\documentclass[mnsc]{informs3}

\OneAndAHalfSpacedXI % current default line spacing
\usepackage{endnotes}

\usepackage{float} 
\usepackage{enumitem}
\usepackage{amsmath}
\usepackage{placeins}   % for \FloatBarrier
\usepackage{bm} 
\usepackage{framed}
\usepackage[outdir=./]{epstopdf}
\usepackage{pgf}
\usepackage{pgfplots}
\pgfplotsset{compat=1.18} 
\usepackage{subfig}

\usepackage{graphicx}
\usepackage{xcolor,cancel}
\usepackage{tabularx}
\usepackage{multirow}
\usepackage{threeparttable}
\usepackage{booktabs}  % added by Sheng

\usepackage{algorithm}
\usepackage{algorithmicx}
\usepackage{algpseudocode}

\usepackage{url}

\newcolumntype{C}[1]{>{\centering\arraybackslash}p{#1}}

\usepackage{natbib}
 \bibpunct[, ]{(}{)}{,}{a}{}{,}%
 \def\bibfont{\small}%
\TheoremsNumberedThrough     % Preferred (Theorem 1, Lemma 1, Theorem 2)
\ECRepeatTheorems

\EquationsNumberedThrough    % Default: (1), (2), ...
\begin{document}
%%%%%%%%%%%%%%%%

%\setcounter{page}{1}
%\DOI{10.1287/xxxx.0000.0000}%

\RUNAUTHOR{xxx}

\RUNTITLE{xxx}

% \TITLE{Bundling Information Goods of Decreasing Value}
\TITLE{Minimizing Type 2 Errors in an Experiment-Rich Regime via Optimal Resource Allocation}

\ARTICLEAUTHORS{%
\AUTHOR{Fenghua Yang, Dae Woong (David) Ham, Stefanus Jasin}
\AFF{Stephen M. Ross School of Business, University of Michigan, Ann Arbor, MI, USA  \\ \EMAIL{yfenghua, daewoong, sjasin@umich.edu}} 
} % end of the block

\ABSTRACT{Randomized experiments (often known as ``A/B tests'') are widely used to 
evaluate product and service innovations. We study how to allocate
limited experimentation resources across $M$ concurrent experiments in an 
experiment-rich regime. Existing work on allocation has predominantly 
focused on minimizing the \emph{worst-case mean squared error} (MSE) of 
estimated treatment effects, which favors experiments with larger (and 
typically unknown) outcome variance. While appropriate for controlling 
estimation accuracy, this objective does not directly capture a common 
managerial priority in screening stages: detecting practically meaningful 
treatment effects with high probability. Motivated by this, we consider the objective of \emph{minimizing the 
worst-case Type~2 error} across all experiments. When the standard 
deviations are known, we characterize the power-optimal allocation and 
show that MSE-based allocations can be highly inefficient for detection, 
even though the two objectives align asymptotically. When the standard 
deviations are unknown and must be learned from pilot data, we show that 
a naive plug-in approach---treating pilot standard deviations as truth---can 
suffer substantial power loss. We propose inflating pilot estimates via 
\emph{correction factors} and develop three optimization-based frameworks for selecting them, each reflecting a different risk 
criterion with distinct managerial implications. Although the resulting 
stochastic programs are computationally challenging at scale, we derive 
tractable surrogate reformulations inspired by robust optimization and 
establish favorable theoretical properties. We further propose 
\textbf{Surrogate-$S$}, a fully data-dependent and implementable procedure 
that computes correction factors using only pilot variance estimates and 
achieves near-oracle performance in numerical experiments. %Finally, we extend the framework to a weighted minimax objective to accommodate heterogeneous experiment priorities.
}

%\KEYWORDS{deterministic inventory theory; infinite linear programming duality;
%  existence of optimal policies; semi-Markov decision process; cyclic schedule}


\maketitle
%%%%%%%%%%%%%%%%%%%%%%%%%%%%%%%%%%%%%%%%%%%%%%%%%%%%%%%%%%%%%%%%%%%%%%

% Samples of sectioning (and labeling) in OPRE
% NOTE: (1) \section and \subsection do NOT end with a period
%       (2) \subsubsection and lower need end punctuation
%       (3) capitalization is as shown (title style).
%
%\section{Introduction.}\label{intro} %%1.
%\subsection{Duality and the Classical EOQ Problem.}\label{class-EOQ} %% 1.1.
%\subsection{Outline.}\label{outline1} %% 1.2.
%\subsubsection{Cyclic Schedules for the General Deterministic SMDP.}
%  \label{cyclic-schedules} %% 1.2.1
%\section{Problem Description.}\label{problemdescription} %% 2.

% Text of your paper here

\section{Introduction} \label{sec:introduction}

The proliferation of large-scale online platforms has fundamentally
changed how firms design, evaluate, and deploy new ideas. A central
enabler of this transformation is the widespread use of randomized
experiments (A/B testing) to guide decisions in product design, web
layout, pricing, advertising, and algorithm development, among many
other applications. Modern platforms routinely run dozens, and often
hundreds, of experiments in parallel, giving rise to what we call an
\emph{experiment-rich regime} \citep{schmit2019optimal}. 

For example, Microsoft was already conducting approximately 250 daily
experiments on its Bing search engine as early as 2013
\citep{Kohavi2013}. Platforms such as Netflix and Booking.com have each
reported running over 10{,}000 experiments annually to optimize their
products and services \citep{Kohavi2017}. More recently, the Google
Search team reported launching 16{,}871 experiments in 2023 alone
\citep{Google2023}. While this growing abundance of parallel
experimentation greatly accelerates organizational learning, it also
creates fundamental challenges for experimental design and the
allocation of limited experimentation resources.

Despite their large active user bases, high-tech companies often face
tight constraints on experimentation resources. A finite pool of users must be allocated across an ever-growing portfolio of
concurrent A/B tests. One fundamental reason is that experimental
traffic cannot be freely reused across tests. Users cannot be
indiscriminately assigned to multiple experiments: Overlapping tests
may interact and confound attribution, so platforms must partition
traffic to preserve valid inference
\citep{tang2010overlapping, pansare_2025}. Consequently, traffic that
might once have been concentrated on a single flagship experiment must
now be split across many. 

At the same time, the total amount of traffic available for
experimentation is inherently limited. Experiments are typically run
under fixed timelines so that results can inform timely product
decisions and feature deployments, which bounds the number of
observations that can be collected. In addition, platforms are often cautious
about exposing large numbers of users to potentially inferior
treatments, further restricting experimental traffic. These limits are
compounded by operational constraints—such as server capacity,
engineering support, and analyst time—which jointly constrain the
overall scale and duration of experimentation.

This reality raises a central design question: 
\vspace{2mm}
\begin{center}
\emph{Given a fixed pool of users (or human subjects), how should they be
allocated \\ across $M$ concurrent experiments to support reliable
decision-making?
} 
\end{center}

\vspace{2mm}
A natural and well-studied approach in the literature is to optimize
allocation for \emph{estimation accuracy}. In particular, many works focus on
minimizing the \emph{worst-case mean squared error (MSE)} of estimated
treatment effects, which leads to allocating more samples to
experiments with higher outcome variance
\citep{antos2008active, carpentier2011upper, deng2012active}. Such
minimax MSE criteria provide uniform guarantees on estimation precision
across experiments and have become a standard benchmark for resource
allocation. However, MSE-based allocations do \emph{not} directly address the statistical
power of detecting true effects. This mismatch arises because
large-scale experimentation typically serves two distinct purposes at
different stages. In an
\emph{initial screening phase}, a platform evaluates many candidate
ideas with the goal of identifying which ones exhibit sufficiently
strong effects to justify further investment. Decisions at this stage
are inherently binary: determining whether an effect exceeds a
meaningful threshold and should be retained. In this setting, a
``miss'' (a false negative) is particularly costly, since a potentially
valuable innovation may be discarded prematurely. By contrast, in a
subsequent \emph{confirmatory phase}, attention shifts to precisely
estimating the magnitude of the effect for a smaller set of surviving
ideas, where estimation accuracy supports forecasting, resource
allocation, and rollout decisions. This paper focuses on the initial
screening phase, where detection performance—rather than estimation
precision—is the primary concern. This motivates optimizing allocation
directly for detection, namely by controlling Type~2 error (see below)
at the portfolio level. Indeed, we show that allocations optimized for
MSE can exhibit poor Type~2 error performance, especially under tight
resource constraints.

Framing the problem in terms of screening naturally leads to a
hypothesis-testing perspective. Any hypothesis test involves two types
of error. A {Type~1 error} occurs when the null hypothesis is
incorrectly rejected, producing a false positive, whereas a
{Type~2 error} occurs when the null hypothesis is incorrectly
accepted when the alternative is true, producing a false negative. Both errors carry significant business consequences. False positives can result in wasted investment, misallocated engineering effort, and erosion of confidence in the experimentation process, whereas false negatives may cause platforms to overlook genuinely valuable innovations, slowing organizational learning and weakening competitive advantage.
While standard
testing procedures explicitly control the Type~1 error by fixing a
significance level~$\alpha$, it provides no comparable guarantees for
the Type~2 error. Addressing this imbalance in the screening phase is
the central focus of this paper.

Motivated by these considerations, we develop an allocation method that
provides explicit guarantees on detection performance across a
portfolio of experiments. Our approach adopts a minimax perspective and
seeks to minimize the maximum Type~2 error across experiments. By doing
so, it prevents any single experiment from becoming severely
underpowered and eliminates weak links in the testing portfolio. As a
result, the platform can ensure a uniform baseline level of statistical
power across all experiments, even under tight resource constraints.

\vspace{2mm}
\textbf{Our results and contributions.} We summarize our key contributions as follows:
\vspace{1mm}
\begin{enumerate}
\item For the benchmark case in which the standard deviations of outcomes are known, we derive a closed-form \emph{power-optimal allocation}. This allocation assigns samples in proportion to the square of the standard-deviation-to-effect-size ratio, thereby equalizing Type~2 errors across experiments. In sharp contrast, the MSE-optimal
allocation depends solely on standard deviations and disregards effect sizes. Numerical results demonstrate that, under tight resource constraints, the power-optimal allocation can substantially increase the probability of detecting true effects relative to the MSE-based benchmark.

\vspace{1mm}
\item For the more realistic case in which the standard deviations are unknown, we follow common practice and assume that the platform can estimate them from pilot data. A natural but naive approach is to substitute the pilot-based estimates as if they were true values and apply the known-standard-deviation allocation. Numerical results show that this \emph{naive plug-in} strategy can lead to substantial power loss. To mitigate this issue, we adapt the concept
of \emph{correction factors}, originally developed for single-experiment settings, to the multi-experiment context. Specifically, we introduce correction factors that inflate the pilot-based standard deviations, providing a safeguard against variance underestimation and the associated loss of power.

\vspace{1mm}
\item We propose three optimization-based frameworks for selecting these correction factors, reflecting distinct risk preferences: (i) \emph{tolerance-based optimization} (\textbf{TOL}), which minimizes the smallest possible tolerance $\delta$ such that, with high probability, the realized maximum Type 2 error remains within $\delta$ of the optimum; (ii) \emph{confidence-based optimization} (\textbf{CONF}), which maximizes the probability of meeting a pre-specified tolerance level; and (iii) \emph{expectation-based optimization} (\textbf{EXP}), which minimizes the expected realized maximum Type 2 error. Together, these frameworks offer a flexible toolkit for balancing
reliability and risk. 

\vspace{1mm}
\item While conceptually useful, unfortunately, the original \textbf{TOL}, \textbf{CONF}, and \textbf{EXP} formulations are computationally intractable in large-scale settings. We, therefore, develop tractable reformulations inspired by \emph{robust optimization} and derive favorable theoretical properties. To operationalize these reformulations in practice, we introduce the \textbf{Surrogate-$S$} method---a fully data-dependent algorithm that computes correction factors using only pilot-based standard deviation estimates. Our numerical experiments demonstrate that \textbf{Surrogate-$S$} is highly competitive with the theoretical oracle benchmark that uses true variances, validating it as a scalable, principled solution for experiment-rich regimes.

\iffalse
\vspace{1mm}
\item Finally, we extend our analysis to settings where experiments may differ in their relative importance. Motivated by practical scenarios where some discoveries carry more business value than others, we study a weighted version of the power-optimal allocation problem. We characterize the optimal allocation in this setting and show how our robust reformulation approach naturally adapts to incorporate weights. This extension demonstrates the flexibility of our framework to handle heterogeneous priorities across experiments.
\fi
\end{enumerate}

\vspace{1mm}
Overall, our paper contributes to the growing literature on optimal design and statistical decision-making in large-scale experimentation. By explicitly focusing on minimizing Type 2 errors under uniform guarantees, we highlight the limitations of estimation-centric allocation rules and provide new tools for platforms operating in experiment-rich regimes. Our results offer both conceptual insights and practical algorithms that can help platforms maximize the value of their testing resources in high-stakes decision environments.

\vspace{2mm}
\textbf{Organization of the paper.} The remainder of the paper is organized as follows. Section~\ref{sec:lit} reviews the related literature, and 
Section~\ref{sec:model} introduces the model and formalizes the 
allocation problem. Section~\ref{sec:known_sigma} studies the benchmark 
case in which the standard deviations are known, characterizes the 
power-optimal allocation, and compares it to the MSE-optimal rule. Section~\ref{sec:unknown sigma 2exp} turns to the more realistic case 
with unknown standard deviations and analyzes a stylized two-experiment 
setting to characterize the structure of the ideal correction factor. 
Building on these structural insights, Section~\ref{sec:approximate k} develops tractable, robust optimization–inspired surrogate reformulations for large experiment portfolios and introduces the data-dependent implementation. 
Section~\ref{sec:closing} concludes the paper. Unless otherwise noted, all proofs and additional technical details 
are provided in the Appendix of the paper.

\section{Literature Review} \label{sec:lit}

Our work in this paper relates primarily to three streams of literature: large-scale experimentation, resource allocation in multi-armed bandits, and sample size allocation with probabilistic guarantees when using estimated variance.

\subsection{Large-Scale Experimentation and A/B Testing}

The rapid growth of large-scale online platforms has transformed A/B testing into a central tool for decision-making in product design, pricing, and algorithm optimization. Industry leaders such as Microsoft, Google, Netflix, Booking.com, and Meta collectively run thousands of experiments annually \citep{Kohavi2013,Google2023,Kohavi2017}. This \emph{experiment-rich regime} accelerates organizational learning but creates significant challenges in allocating finite user traffic across many concurrent tests. Real-world experimentation systems described by \citet{tang2010overlapping} and \citet{Kohavi2013} highlight operational
issues such as traffic segmentation, interference between parallel experiments, and the need for flexible allocation policies. 

More recently, \citet{schmit2019optimal} consider an experiment-rich setting with a sequential stream of users and an effectively infinite pool of potential experiments. Their objective is to minimize the expected time to the first ``discovery'' and they characterize the optimal policy for this goal. While their approach is effective for maximizing the speed of achieving the first discovery in an infinite-horizon setting, it does not address the common scenario in which a fixed, finite budget must be allocated across a known set of concurrent experiments with guarantees on their statistical power. Our work fills this gap by providing the first optimal allocation method with explicit Type~2 error control for all experiments.

%\citet{schmit2019optimal} formalized the experiment-rich setting and proposed allocation policies to manage opportunity costs when many experiments compete for limited traffic. Their framework considers a dynamic problem with a sequential stream of users and a infinite pool of potential experiments, aiming to minimize the expected time to the first discovery. While their approach effectively prioritizes speed and exploration, it does not address the most common scenario where a fixed budget must be allocated across a specific, finite set of concurrent experiments to guarantee uniform performance. While prior studies underscore the importance of allocation, no existing approach guarantees uniform statistical power across experiments. Our framework fills this gap by providing the first optimal allocation method with explicit Type~2 error control for all experiments.

\subsection{Resource Allocation in Multi-Armed Bandits}
A substantial body of work on allocation arises in the active learning literature within the multi-armed bandit (MAB) framework. Early studies focused on estimation-centric objectives such as minimizing the worst-case mean squared error (MSE) of arm mean estimates. For example, \citet{antos2008active} and \citet{carpentier2011upper} advocate allocating samples in proportion to arm variances to achieve uniform estimation accuracy. 
Implementing this variance-proportional allocation requires learning arm variances from data, since they are typically unknown ex ante. This motivates adaptive MAB algorithms that estimate (and upper-bound) variances online and allocate future pulls according to these confidence bounds, e.g., via UCB-style variance indices \citep{audibert2009variance, carpentier2011upper}.

While minimizing MSE is the gold standard for \emph{confirmatory stages} (Phase 2)—where precise impact sizing is required for financial forecasting and final launch decisions—it is less aligned with the objectives of the \emph{initial screening phase} (Phase 1), which is the focus of our work. In this screening regime, the decision-maker's primary goal is to efficiently filter a large funnel of potential innovations to identify surviving candidates for follow-up. As \citet{Kohavi2017} notes, even small design changes—such as Amazon’s revision of a credit card sign-up page—can yield substantial gains but may be difficult to detect. In this context, the primary risk is not imprecise measurement, but rather failing to detect a truly impactful treatment (Type~2 error). Consequently, we shift the objective from MSE minimization to explicitly controlling the maximum Type~2 error. Our results show that, under limited resources, MSE-based allocations can produce markedly lower statistical power than allocations explicitly designed for discovery.

\subsection{Testing with Unknown Variance and Pilot Studies}

When variances are unknown, researchers often conduct a \emph{pilot study} —a small preliminary study run prior to the main experiment—to test key procedures and to obtain initial estimates of nuisance parameters such as outcome variability that inform sample-size planning and allocation \citep{thabane2010tutorial,arain2010pilot}. Prior research
(e.g., \citealt{birkett1994internal,browne1995use,kieser1996use,julious2005sample,
sim2011size,Teare2014}) recommends pilot sizes that balance the need for accurate variance estimation with the need to preserve samples for the main study. A recurring challenge in this literature is that the sampling distribution of the sample variance is right-skewed, with a median strictly less than the true variance. Consequently, the
probability that a pilot-based variance estimate underestimates the true variance is strictly greater than $50\%$, leading to a systematic risk of underpowered designs. 

\citet{browne1995use} and \citet{kieser1996use} address this issue in single-experiment settings by applying \emph{correction factors} that inflate pilot-based variance estimates, thereby guaranteeing target power with a specified probability. These methods, however, do not directly extend to multi-experiment allocation problems with shared resource constraints. We extend this idea to multi-experiment settings by applying correction factors to pilot-based standard deviation estimates, thereby mitigating the power loss inherent in naive plug-in approaches. Our method draws on robust optimization principles, which address uncertainty by ensuring worst-case performance guarantees.

\section{Model} \label{sec:model}

We consider a setting where a platform (or an experimenter) runs \( M \) independent experiments in parallel and must allocate a fixed number \( N \in \mathbb{Z}^+ \) of users (human subjects) across them. For simplicity, we assume that each subject participates in at most 
one experiment, ensuring independence of outcomes across experiments. 
Relaxing this assumption would introduce cross-experiment dependence 
and fundamentally alter the structure of the allocation problem; 
we defer this extension to future work. 

Let \( E = \{e_1, e_2, \dots, e_M\} \) denote the set of experiments. The platform selects an allocation vector \( \vec{n} = (n_1, \dots, n_M) \), 
where \( n_i \) indicates the number of subjects assigned to experiment \( e_i \). An allocation \( \vec{n} \) is said to be feasible if 
\( n_i \geq 0 \) for all \( i \in [M] := \{1, \dots, M\} \), 
and the total number of subjects satisfies 
\( \sum_{i=1}^M n_i \le N \). Without loss of generality, we assume that experiment $e_i$ is associated
with the following one-sided hypothesis test:
\begin{eqnarray*}
H_{i,0}: \mu_i \leq \theta_i \,\, \mbox{ and } \,\,
H_{i,1}: \mu_i > \theta_i 
%\quad \text{(Alternative hypothesis: the true mean outcome exceeds the threshold)}
\end{eqnarray*}
\noindent
where $\mu_i$ denotes the true mean outcome of experiment $e_i$, and
$\theta_i$ is the threshold required for $e_i$ to be considered
effective. If we fail to reject the null hypothesis, we conclude that
there is insufficient statistical evidence that the outcome exceeds the
threshold; otherwise, we conclude that it does exceed the threshold and
a \emph{discovery} has been made. 

We make the following assumptions regarding the statistical framework.

\vspace{1mm}
\begin{assumption}
\label{assumption:2}
To test the hypothesis in experiment \( e_i \), \( i \in [M] \), the platform collects \( n_i \) independent and identically distributed (i.i.d.) samples \( S_i = \{X_{i,1}, \dots, X_{i,n_i}\} \), where each observation is drawn from a Normal distribution \( X_{i,j} \sim N(\mu_i, \sigma_i^2) \). These samples are used to compute a test statistic and decide whether to reject the null hypothesis \( H_{i,0} \).
\end{assumption}

\vspace{1mm}
%\begin{remark} \label{rem:interpretation}
We provide a practical interpretation of the idealized variable
$X_{i,j}$ in Assumption~\ref{assumption:2}. In practice, experiments often
compare outcomes between a treatment group and a control group; our
model abstracts this comparison into a single random variable
$X_{i,j}$ whose mean equals the treatment effect~$\mu_i$. From a managerial perspective, it is helpful to view $X_{i,j}$ as the
\emph{realized incremental value} generated by an individual user
interaction. For example, consider an experiment that evaluates a new
``One-Click Checkout'' button. In this setting, $\mu_i$ represents the
true expected revenue lift per user (e.g., \$0.50), capturing the
underlying effect of interest. The random variable $X_{i,j}$ then
represents the incremental revenue observed in the $j$-th user
session. Individual sessions may exhibit substantial variability due
to noise~$\sigma_i$—for instance, generating no additional revenue or
a large purchase—but aggregating these incremental values across users
recovers the true business impact. Finally, although the Normality assumption in Assumption~\ref{assumption:2}
is idealized, it is standard in large-scale experimentation. Furthermore, our results rely on the sample mean $\bar{X}_i$ and sample variance to have normal and chi-squared distributions, respectively, which hold asymptotically through the 
Central Limit Theorem. Therefore, we view the normality assumption as more ``regular''.

%\end{remark}

\vspace{1mm}
\begin{assumption}
\label{assumption:type1}
For each experiment \( e_i \), \( i \in [M] \), the platform seeks to control the Type 1 error---i.e., the probability of rejecting the null hypothesis when it is in fact true---at a level no greater than \( \alpha \), for some significance level \( \alpha \in (0, 1) \).
\end{assumption}

\vspace{1mm}
Assumption~\ref{assumption:type1} is stated for completeness as a valid hypothesis test controls Type~1 error by design.  Consequently, the allocation decision \( \vec{n} \) primarily affects statistical power. Our objective is therefore to minimize the maximum Type 2 error across all experiments.

This objective contrasts with much of the existing literature, which
typically focuses on minimizing the maximum mean squared error (MSE) in
estimation \citep{antos2008active, carpentier2011upper, deng2012active}.
Optimizing allocation with respect to Type~2 error, rather than MSE,
offers a different perspective that is often more closely aligned with
practical decision-making. Specifically, Type~2 error directly reflects how managers evaluate experimental
outcomes. By definition, it is the probability of failing to detect an
experiment that truly delivers a meaningful effect, where what
constitutes ``meaningful'' is determined by a managerial threshold
$\theta_i$. In the common case $\theta_i = 0$, the objective is to
correctly identify treatments with a statistically significant
positive effect. In many applications, however, managers often require the
effect to exceed a strictly positive threshold in order to justify
implementation costs. For example, a feature may be considered
successful only if it increases user engagement by at least $3\%$, in
which case $\theta_i = 0.03$. This flexibility makes Type~2 error a
natural performance metric in large-scale A/B testing, where the goal
is to screen out changes with insufficient practical impact.

At the same time, Type~2 error and MSE play complementary roles over the
lifecycle of experimentation. In the \emph{initial screening phase}
(Phase~1), the platform must rapidly evaluate a large set of candidate
ideas, and the decision problem is inherently binary: determining
whether an effect exceeds the threshold~$\theta_i$. In this phase, the
primary statistical risk is a Type~2 error, since failing to detect a
genuinely beneficial treatment may eliminate it from further
consideration. By contrast, once a treatment passes this screening
stage, the \emph{confirmatory phase} (Phase~2) focuses on accurately
estimating its effect size to support forecasting, resource allocation,
and launch decisions. In that setting, minimizing MSE is more
appropriate, as estimation precision rather than detection is the main
objective. Our work explicitly targets Phase~1, where power
considerations dominate and where MSE-based allocation rules may be
poorly aligned with the platform’s screening goals.

Below, we begin by discussing the case in which the standard deviation vector of our outcomes
\( \vec{\sigma} = (\sigma_1, \dots, \sigma_M) \) is known, and then we discuss the case where \( \vec{\sigma} \) is unknown. %{\color{red} ST: I commented out the discussion that $\vec{\sigma}$ is estimated using a pooled or two-sample variance estimator because the final formula is different from $s/\sqrt{n}$. This may raise confusion.} %As a reminder, $\vec{\sigma}$ represents the estimated treatment effect variance. Typically, $\vec{\sigma}$ is estimated using a pooled or two-sample variance estimator. In the next section, we assume this is known and show some preliminary results under this simplified case. 

\subsection{The case with known $\vec \sigma$}

In this case, we use the standard $z$-test where the test statistic for experiment \( e_i \) is given by
\[
Z_i = \frac{\bar{X}_i - \theta_i}{\sigma_i / \sqrt{n_i}},
\]

\vspace{1mm}
\noindent
where \( \bar{X}_i \) denotes the sample mean of the observations 
\( X_{i,1}, \dots, X_{i,n_i} \). Under the null hypothesis and Assumption~\ref{assumption:2}, one can view the observations as drawn from a Normal distribution with mean \( \mu_i = \theta_i \) and variance \( \sigma_i^2 \). 
Thus, the test statistic \( Z_i \) follows a standard Normal distribution (under the null), 
i.e., \( Z_i \sim N(0, 1) \). Let \( q_{1-\alpha} \) denote the \( (1 - \alpha) \)-quantile of the standard Normal distribution. We then apply the standard one-sided decision rule that controls Type~1 error (under Assumption~\ref{assumption:type1}): 
\begin{eqnarray*}
\mbox{reject \( H_{i,0} \) if \( Z_i > q_{1-\alpha} \); otherwise, we fail to reject \( H_{i,0} \)}. 
\end{eqnarray*}

%\noindent
To analyze the power of this test, as is standard in the literature (e.g., \citealt{cohen2013statistical}, \citealt{fleiss2013statistical}), we introduce a design parameter $\Delta_i > 0$, referred to as the \emph{minimum detectable gap} (MDG). The MDG represents the smallest excess over the threshold $\theta_i$
that is considered meaningful for experiment $e_i$. Power is then
evaluated under the hypothetical mean $\mu_i = \theta_i + \Delta_i$.
Note that this does not imply that the actual gap $\mu_i - \theta_i$ equals
$\Delta_i$, since $\mu_i$ is unknown, but rather that the experiment is designed to guarantee the desired power whenever the true gap is at least $\Delta_i$. 

Specifically, under the design scenario $\mu_i = \theta_i + \Delta_i$, the test
statistic $Z_i$ has mean $\delta_i = \tfrac{\Delta_i \sqrt{n_i}}{\sigma_i}$
and variance $1$, i.e., $Z_i \sim N(\delta_i,1)$. Consequently, the
Type~2 error is
\begin{eqnarray}
\beta(\sigma_i, n_i; \Delta_i) \ := \ 
\mathbb{P}(Z_i \le q_{1-\alpha} \mid \mu_i = \theta_i + \Delta_i)
\ = \ 
\Phi\!\left(q_{1-\alpha} - \frac{\Delta_i \sqrt{n_i}}{\sigma_i}\right),
\label{eq:power}
\end{eqnarray}
where $\Phi(\cdot)$ denotes the standard Normal CDF. Now, suppose an experiment is designed with $\Delta_i = 1\%$
and $n_i$ is chosen so that the Type~2 error at this design gap is $5\%$.
If the true gap $\mu_i - \theta_i$ is actually larger than $1\%$, then with the
same sample size $n_i$ the noncentrality parameter
$\tfrac{(\mu_i - \theta_i)\sqrt{n_i}}{\sigma_i}$ increases, and the
actual Type~2 error will be strictly smaller than $5\%$.

Choosing $\Delta_i > 0$ is essential: requiring power against
arbitrarily small gaps (i.e., $\Delta_i = 0$) would necessitate
unbounded sample sizes and would rarely be of practical value. In
practice, the choice of $\Delta_i$ is context-dependent. It is often set
to reflect the smallest improvement that justifies implementation costs
(e.g., a 1\% lift in click-through rate or a 0.5\% increase in revenue).
In other settings, available traffic or resource constraints dictate the
feasible $\Delta_i$, since with limited sample size only larger gaps can
be reliably detected. In this way, $\Delta_i$ provides a bridge between
statistical design and managerial priorities.

Recall that the platform must choose the allocation vector \( \vec{n} \) 
to minimize the maximum Type 2 error across all experiments. Formally, the platform solves the following optimization problem:
\begin{eqnarray*}
\textbf{POWER-OPT:} \hspace{15mm} \beta^*(\vec \sigma) \ := \ \min_{\vec n} && \max_{i \in [M]} \ \{\beta(\sigma_i, n_i)\} \\
\mbox{subject to} && \sum_{i=1}^M n_i \le N, \quad n_i \ge 0 \ \ \forall i  \in [M] 
\end{eqnarray*}

%\vspace{1mm}
\noindent
While in practice each \( n_i \) must be an integer, we follow the standard convention in the literature and relax this requirement, allowing \( n_i \in \mathbb{R}_{\ge 0} \) 
to facilitate analytical tractability. We refer to the above problem as \textbf{POWER-OPT} because minimizing the maximum Type 2 error is equivalent to maximizing the minimum power across all experiments. We use $\vec n^*(\vec\sigma)$ to denote the optimal solution to \textbf{POWER-OPT}.

The complete analysis of \textbf{POWER-OPT} together with the comparison with the more standard approach that minimizes the maximum MSE is given in Section \ref{sec:known_sigma}.

\subsection{The case with unknown $\vec \sigma$}

In most applications, the standard deviations $\vec{\sigma}$ are unknown. In this case we cannot directly compute the Type 2 error of experiment \( e_i \) using Equation~\eqref{eq:power}. A common statistical remedy is to replace the $z$-test with a $t$-test and evaluate the Type 2 error using the $t$-distribution in place of the Normal distribution. 

Specifically, we define the new test statistic \( T_i \) as follows:
\begin{eqnarray*}
T_i = \frac{\bar{X}_i - \theta_i}{\tilde{S}_i / \sqrt{n_i}},
\end{eqnarray*}

%\vspace{1mm}
\noindent
where \( \tilde{S}_i \) denotes the classic sample standard deviation computed 
from \( n_i \) i.i.d.\ observations. The corresponding Type 2 error for experiment \( e_i \) is given by
\begin{eqnarray*}
\beta^T(n_i) \ := \ 
\mathbb{P}(T_i \le t_{1-\alpha, n_i-1} \mid \mu_i = \theta_i + \Delta_i),
\end{eqnarray*}

%\vspace{1mm}
\noindent
where \( t_{1-\alpha, n_i - 1} \) denotes the $(1-\alpha)$-quantile of 
the $t$-distribution with \( n_i - 1 \) degrees of freedom. Note that under the alternative \( \mu_i = \theta_i + \Delta_i \), 
the statistic \( T_i \) follows a noncentral $t$-distribution with 
\( n_i - 1 \) degrees of freedom and noncentrality parameter $\delta_i = \frac{\Delta_i \sqrt{n_i}}{\sigma_i}$. 

This leads to the following optimization problem, analogous to 
\textbf{POWER-OPT}:
\begin{eqnarray*}
\min_{\vec n} && \max_{i \in [M]} \ \{\beta^T(n_i)\} \\
\mbox{subject to} && \sum_{i=1}^M n_i \le N, \quad n_i \ge 0 \ \ \forall i  \in [M] 
\end{eqnarray*}

%\vspace{1mm}
\noindent
While this formulation provides a direct way to optimize the allocation 
vector $\vec{n}$, solving it at scale is computationally challenging. 
Evaluating $\beta^T(n_i)$ requires numerical computation of the 
noncentral $t$ cumulative distribution function, and this evaluation 
must be embedded within an optimization over $\vec{n}$. 
As the number of experiments grows, the resulting problem becomes 
computationally intensive and difficult to scale efficiently. For this reason, 
we do not pursue this approach further in the paper.

\vspace{2mm}
\textbf{Pilot experiment and the naive plug-in method.} A more practical alternative—commonly used in the literature (e.g. \citealt{browne1995use,kieser1996use,Whitehead2016, Teresi2022, Kunselman2024}) and widely adopted in real-world applications—is to assume that, for each experiment \( e_i \), the platform has access to an estimate \( S_i \) of the true standard deviation \( \sigma_i \), obtained from a small pilot study of size \( \epsilon_i \ge 2 \). The size of a pilot study can vary substantially depending on the scale of the main experiment—ranging from 15 to 150 participants in clinical trials, to several thousand in experiments conducted by large tech companies with high customer traffic \citep{browne1995use, kieser1996use, Teare2014}. These pilot-based estimates are then used to inform the sample size allocation for the main experiments. Specifically, \( S_i \) is the sample standard deviation computed from \( \epsilon \) i.i.d.\ pilot observations, and satisfies the following distributional property:
\[
\frac{(\epsilon_i - 1) S_i^2}{\sigma_i^2} \ \sim \ \chi^2_{\epsilon_i - 1},
\]

%\vspace{1mm}
\noindent
where \( \chi^2_n \) denotes the chi-squared distribution with \( n \) degrees of freedom. We then use the estimated  
\( \vec{S} = (S_1, \dots, S_M) \) in place of the real unknown $\vec{\sigma}$ to determine the allocation vector 
\( \vec{n} \).

A natural approach is to apply a \emph{naive plug-in} method, where we directly substitute \( \vec{S} \) for \( \vec{\sigma} \) in the \textbf{POWER-OPT} formulation and solve the resulting problem. This method treats \( \vec{S} \) as if it were the true standard deviation vector and proceeds under the $z$-test framework for known variances. Similarly, we follow the approach of \citet{browne1995use, kieser1996use}, which assumes the true variance $\vec{\sigma}$ is well approximated by $\vec{S}$. However, prior work has shown that this naive plug-in method can perform poorly, even in the simplified case with only a single experiment. % and no resource constraint. 

\vspace{2mm}
\textbf{Using correction factors.}
To address the shortcomings of the naive plug-in approach, the literature
proposes inflating the pilot-based standard deviation estimates by a
\emph{correction factor}. This adjustment accounts for estimation
uncertainty and ensures that the resulting design achieves the desired
power with a specified probability
(e.g., \citealt{browne1995use, kieser1996use}). Although the sample
standard deviation is an unbiased estimator of $\sigma_i$ in expectation,
its distribution is skewed and exhibits high variability. As a result,
there is a nontrivial chance—greater than $50\%$—that $S_i$ underestimates
$\sigma_i$, leading to underpowered designs if left uncorrected.
Inflating $S_i$ by an appropriate factor mitigates this risk and provides
more reliable high-probability guarantees. In our setting, this adjustment corresponds to replacing \( \sigma_i \) with \( \sqrt{k_i} S_i \), for some \( k_i\) for all $i \in [M]$, and solving the resulting \textbf{POWER-OPT} problem using these adjusted estimates. This is the approach we adopt in this paper. (Throughout, we will also refer to $\vec k$ as the \emph{inflation} factors.) 

We consider the set of allocation vectors of the form \(\vec n =  \vec{n}^*(\vec{k}, \vec{S}) \), obtained by solving \textbf{POWER-OPT} using the adjusted plug-in standard deviations. That is, for each pair $(\vec k, \vec S)$, the vector \( \vec{n}^*(\vec{k}, \vec{S}) \) is the optimal solution to  
\begin{eqnarray*}
\min_{\vec n} && \max_{i \in [M]} \ \{\beta(\sqrt{k_i} S_i, n_i)\} \\
\mbox{subject to} && \sum_{i=1}^M n_i \le N, \quad n_i \ge 0 \ \ \forall i  \in [M] \nonumber
\end{eqnarray*}

%\vspace{1mm}
\noindent
Under the allocation \( \vec{n}^*(\vec{k}, \vec{S}) \), the maximum Type 2 error across all experiments is given by:
\begin{eqnarray}
\tilde{\beta}^*(\vec{k}, \vec{S}, \vec{\sigma}) \ := \ 
\max_{i \in [M]} 
\big\{ \beta\big(\sigma_i, n^*_i(\vec{k}, \vec{S})\big) \big\}. \label{eq:tildebeta*}
\end{eqnarray}

%\vspace{1mm}
\noindent
We treat the correction vector \( \vec{k} \) as the decision variables to optimize, with the goal of ensuring that the realized maximum Type 2 error \( \tilde{\beta}^*(\vec{k}, \vec{S}, \vec{\sigma}) \) closely approximates the true optimum \( \beta^*(\vec{\sigma}) \). 

Since \( \vec{S} \) is random, the realized \( \tilde{\beta}^*(\vec{k}, \vec{S}, \vec{\sigma}) \) is itself a random quantity. Thus, any meaningful choice of \( \vec{k} \) must account for its distributional behavior. To guide the selection of the correction vector \( \vec{k} \), we introduce three optimization-based frameworks. Each framework reflects a distinct objective the platform may wish to prioritize and leads to a different problem formulation. We refer to them as \emph{tolerance-based} optimization, 
\emph{confidence-based} optimization, and \emph{expectation-based} optimization, respectively. Note that there are different objectives, as opposed to one objective in Section~\ref{sec:known_sigma}, due to the different probabilistic guarantees one can make about the random power \( \tilde{\beta}^*(\vec{k}, \vec{S}, \vec{\sigma}) \). We discuss them below.
\vspace{1mm}
\begin{enumerate}
%\vspace{2mm}
%\noindent
\item \textbf{Tolerance-based Optimization.} This framework prioritizes guaranteeing a high-confidence performance bound. 
Given a desired confidence level $\gamma \in (0,1)$, the goal is to choose 
$\vec{k}$ such that, with probability at least $\gamma$, the realized 
maximum Type 2 error $\tilde{\beta}^*(\vec{k}, \vec{S}, \vec{\sigma})$ 
is no larger than a threshold $\beta^*(\vec\sigma) + \delta$, where as a reminder $\beta^*(\vec\sigma)$ is the oracle best Type~2 error assuming knowledge of $\vec{\sigma}$. Among all such feasible 
allocations, we seek to minimize the smallest possible tolerance $\delta$. 
Formally, this leads to the following optimization problem:
\begin{eqnarray*}
\textbf{TOL:} \hspace{10mm} \delta^*(\gamma, \vec{\epsilon}, \vec \sigma) \ := \ \min_{\vec{k}, \delta} && \delta \\
\text{subject to} && 
\mathbb{P}(
\tilde{\beta}^*(\vec{k}, \vec{S}, \vec{\sigma}) \le \beta^*(\vec\sigma) + \delta) \ge \gamma, \\
&& k_i \ge 1, \quad \forall i \in [M], \\
&& \delta \ge 0.
\end{eqnarray*}
%\vspace{1mm}
\item \textbf{Confidence-based Optimization.} In contrast to \textbf{TOL}, this framework fixes a performance target $\delta$ and seeks to maximize the confidence level under which it can be guaranteed. That is, given a pre-specified tolerance $\delta$, we 
choose $\vec{k}$ to maximize the probability that the realized maximum Type 2 error remains below this threshold. This approach is appropriate when the platform is willing to tolerate a certain level of error but wants to make that guarantee hold as reliably as possible. This corresponds to the following optimization problem:
\begin{eqnarray*}
\textbf{CONF:} \hspace{10mm} \gamma^*(\delta, \vec{\epsilon}, \vec \sigma) \ := \ \max_{\vec{k}, \gamma} && \gamma \\
\text{subject to} && 
\mathbb{P}(
\tilde{\beta}^*(\vec{k}, \vec{S}, \vec{\sigma}) \le \beta^*(\vec\sigma) + \delta) \ge \gamma, \\
&& k_i \ge 1, \quad \forall i \in [M], \\
&& \gamma \in (0,1).
\end{eqnarray*}

\vspace{1mm}
%\noindent
\item \textbf{Expectation-based Optimization.} Rather than controlling the probability of meeting a specific threshold, this framework directly minimizes the expected value of the 
realized worst-case Type 2 error. This objective reflects a risk-neutral perspective, focusing on average-case performance over the randomness in $\vec{S}$. This leads to the following optimization problem:
\begin{eqnarray*}
\textbf{EXP:} \hspace{10mm} g^*(\vec{\epsilon}, \vec \sigma) \ := \ \min_{\vec{k}} && 
\mathbb{E}[
\tilde{\beta}^*(\vec{k}, \vec{S}, \vec{\sigma})] \\
\text{subject to} && 
k_i \ge 1, \quad \forall i \in [M].
\end{eqnarray*}
\end{enumerate}

\vspace{2mm}
%Together, \textbf{TOL}, \textbf{CONF}, and \textbf{EXP} offer complementary perspectives on how to select the correction factors $\vec{k}$---by targeting high-probability guarantees, threshold-based reliability, and expected performance, respectively. These objectives naturally arise from two major tradeoffs. Since $\vec{S}$ is random, the realized optimal Type~2 error $\tilde{\beta}^*(\vec{k}, \vec{S}, \vec{\sigma})$ is also random. Consequently, one may observe \emph{Scenario~1}: $\tilde{\beta}^*(\vec{k}, \vec{S}, \vec{\sigma})$ is very close to the ideal $\beta^*(\sigma)$, but only with relatively small probability; or \emph{Scenario~2}: $\tilde{\beta}^*(\vec{k}, \vec{S}, \vec{\sigma})$ is within a wider tolerance of $\beta^*(\sigma)$, but with probability close to one. This is a fundamental tradeoff: increasing the probability guarantee that $\tilde{\beta}^*(\vec{k}, \vec{S}, \vec{\sigma})$ is close to $\beta^*(\sigma)$ necessarily requires a wider tolerance, and vice versa. Which scenario is preferable depends on the decision risk: in low-stakes or repeatable settings (e.g., exploratory screening where an underperforming allocation can be corrected in subsequent iterations), an experimenter may accept Scenario~1; whereas in high-stakes or one-shot decisions (e.g., an expensive rollout or policy change), the experimenter may prefer Scenario~2. Therefore, we present both the \textbf{TOL} and \textbf{CONF} objectives, which control tolerance and confidence, respectively. {\color{red} ST: this paragraph feels redundant now. Delete it? Or make it shorter, half long}

It is worth noting that while \textbf{TOL}, \textbf{CONF}, and \textbf{EXP} offer valuable conceptual insights, unfortunately, the associated optimization problems are generally intractable for instances with a large number of experiments 
(i.e., large \( M \)). Specifically, even when $\vec{\sigma}$ is
known, both \textbf{TOL} and \textbf{CONF} are chance-constrained
stochastic programs whose feasibility depends on probabilities (over
the pilot randomness~$\vec{S}$) of events defined through
$\tilde{\beta}^*(\vec{k},\vec{S},\vec{\sigma})$. Evaluating these
probabilities requires computing tail probabilities of nonlinear
functions of~$\vec{S}$ that depend on weighted sums of independent
$\chi^2$ random variables, for which no closed-form expressions exist
in the general $M$-experiment case. A direct numerical approach would require Monte Carlo methods nested inside an outer
optimization over~$\vec{k}$, leading to non-smooth, non-convex, and
computationally prohibitive procedures. 

To address these challenges, in Section~\ref{sec:approximate k} we develop \emph{surrogate reformulations} of each objective through the lens of robust optimization. The proposed surrogate reformulations yield, respectively, an upper bound, a lower bound, and an upper bound on the objectives of \textbf{TOL}, \textbf{CONF}, and \textbf{EXP}. As such, they provide \emph{conservative} estimates of the best achievable tolerance, reliability, and expected maximum Type 2 error under any allocation vector of the form \( \vec{n}^*(\vec{k}, \vec{S}) \). These reformulations are computationally tractable and remain close to
the original formulations under suitable conditions. As a result, they
provide a principled and scalable approach for approximating the
optimal correction in large experiment portfolios.

\section{Known $\vec \sigma$}  \label{sec:known_sigma}
%{\color{red} DH: Worth thinking about for organization but can we move 3.2 and all discussion after this section. This way we just cleanly talk about the known variance case and the results than move on.}
In this section, we examine the optimal allocation vector \( \vec{n}^*(\vec{\sigma}) \) and the corresponding optimal maximum Type 2 error \( \beta^*(\vec{\sigma}) \) when $\sigma$ is known. We then compare this power-optimal approach with a traditional allocation strategy that minimizes the total Mean Squared Error (MSE), highlighting key differences in their objectives, underlying trade-offs, and resulting allocations.

\subsection{Optimal Allocation under \textbf{POWER-OPT}} \label{subsec:power_opt_solution}

The following proposition characterizes the optimal allocation vector 
under \textbf{POWER-OPT}, along with the corresponding maximum 
Type 2 error.

\vspace{1mm}
\begin{proposition} \label{prop:power_opt_optimal}
The optimal allocation under \textbf{POWER-OPT} is given by
\[
n_i^*(\vec{\sigma}) \ = \ 
N \cdot \frac{\left( \frac{\sigma_i}{\Delta_i} \right)^2}
{\sum_{j=1}^M \left( \frac{\sigma_j}{\Delta_j} \right)^2},
\quad \forall i \in [M],
\]

%\vspace{1mm}
\noindent
which equalizes the Type 2 error across all experiments, i.e., $\beta(\sigma_1, n_1^*(\vec{\sigma})) 
= \cdots = 
\beta(\sigma_M, n_M^*(\vec{\sigma})) 
= \beta^*(\vec{\sigma}),
$
where the common (and optimal) Type 2 error is given by
\[
\beta^*(\vec{\sigma}) \ = \ 
\Phi\left( q_{1-\alpha} - 
\sqrt{ \frac{N}{\sum_{j=1}^M 
\left( \frac{\sigma_j}{\Delta_j} \right)^2} } \right).
\]
\end{proposition}

\vspace{1mm}
This allocation assigns more samples to experiments with larger variance-to-effect-size ratios, as captured by the quantity \( (\frac{\sigma_i}{\Delta_i})^2 \). Equivalently, we can interpret \( \frac{\sigma_i}{\Delta_i} \) as a measure of the \emph{statistical 
difficulty} of experiment \( e_i \): It quantifies how hard it is to detect the signal \( \Delta_i \) in the presence of noise \( \sigma_i \). Intuitively, an experiment with higher variance (larger \( \sigma_i \)) or a smaller effect size (smaller \( \Delta_i \)) requires more samples to reliably distinguish the effect from random variation. The optimal allocation accounts for this by distributing the total budget in proportion to \( ( \frac{\sigma_i}{\Delta_i} )^2 \), effectively giving more resources to more difficult experiments.

Note that the optimal allocation under \textbf{POWER-OPT} equalizes the Type 2 error across all experiments. Since the goal is to minimize the maximum Type 2 error, the optimal strategy prevents any single experiment from becoming a weak link by over- or under-allocating relative to its statistical difficulty. As a result, it balances effort across experiments in a way that reflects their relative difficulty and ensures robust power guarantees.

\vspace{2mm}
\begin{remark} \label{remark1}
As discussed in Section~\ref{sec:model}, when $\vec{\sigma}$ is unknown, 
the allocation vector $\vec{n}^*(\vec{k},\vec{S})$ is obtained by solving 
\textbf{POWER-OPT} with adjusted plug-in standard deviations, i.e., by 
replacing each $\sigma_i$ with $\sqrt{k_i}S_i$. 
Applying Proposition~\ref{prop:power_opt_optimal} yields the following expression for $\vec{n}^*(\vec{k},\vec{S})$:
\begin{eqnarray}
n_i^*(\vec{k}, \vec S) \;=\; 
N \cdot \frac{k_i \bigl(\tfrac{S_i}{\Delta_i}\bigr)^2}
{\sum_{j=1}^M k_j \bigl(\tfrac{S_j}{\Delta_j}\bigr)^2},
\qquad \forall i \in [M]. 
\label{eq:n*(k,S)}
\end{eqnarray}
%Observe that what ultimately matters are not the absolute values of the $k_i$’s  but their ratios $\tfrac{k_i}{k_j}$. 

\end{remark}

\subsection{Comparison with the MSE-Minimization Approach} \label{subsec:mse_allocation}

In the context of estimating the mean outcomes \( \mu_i \) for each experiment \( e_i \), a popular objective from the literature is to minimize the maximum Mean Squared Error (MSE) across all experiments (e.g., \citealt{antos2008active, carpentier2011upper, deng2012active}). 

For the sample mean estimator 
\( \hat{\mu}_i = \frac{1}{n_i} \sum_{j=1}^{n_i} X_{i,j} \), the MSE is given by:
\[
\mathbb{E}\left[ (\hat{\mu}_i - \mu_i)^2 \right] = \frac{\sigma_i^2}{n_i},
\]

%\vspace{1mm}
\noindent
where \( \sigma_i^2 \) is the variance of the outcome 
\( X_i \sim \mathcal{N}(\mu_i, \sigma_i^2) \), and 
\( n_i \) is the number of units allocated to experiment \( e_i \). The corresponding minimax allocation problem is:
\begin{eqnarray*}
\min_{\vec{n}} && \max_{i \in [M]} \ \left\{\frac{\sigma_i^2}{n_i}\right\} \\
\text{subject to} &&
\sum_{i=1}^M n_i = N, \quad n_i \ge 0 \ \ \forall i \in [M]
\end{eqnarray*}

%\vspace{1mm}
The above formulation seeks to equalize the MSEs across all experiments, 
ensuring the worst-case estimation error is minimized. 
The optimal allocation is given by
\[
n_i^{\text{MSE}} \ := \ N \cdot \frac{\sigma_i^2}{\sum_{j=1}^M \sigma_j^2}, \quad \forall i \in [M]
\]

%\vspace{1mm}
\noindent
\citep[see, e.g.,][]{antos2008active, carpentier2011upper, deng2012active}. Comparing the MSE-optimal allocation \( \vec{n}^{\text{MSE}} \) with the power-optimal allocation \( \vec{n}^*(\vec{\sigma}) \) from Proposition \ref{prop:power_opt_optimal}, we observe a key difference in how the two approaches handle effect sizes \( \Delta_i \). Specifically, while both strategies allocate more samples to experiments with larger variances \( \sigma_i^2 \), only the power-optimal allocation accounts for signal strength by incorporating \( \Delta_i \) through the hypothesis-testing framework. As a result, MSE-based allocation may under-invest in statistically difficult experiments, leading to poor 
Type 2 error control. This highlights a fundamental distinction: MSE-optimal allocation is well-suited for estimation objectives, whereas power-optimal allocation is tailored for decision-making under uncertainty—especially when the goal is to identify experiments with \( \mu_i \ge \theta_i \). 

\begin{figure}[htbp] % Use [htbp] for flexible placement: here, top, bottom, or page
    \centering
    \includegraphics[width=0.7\textwidth]{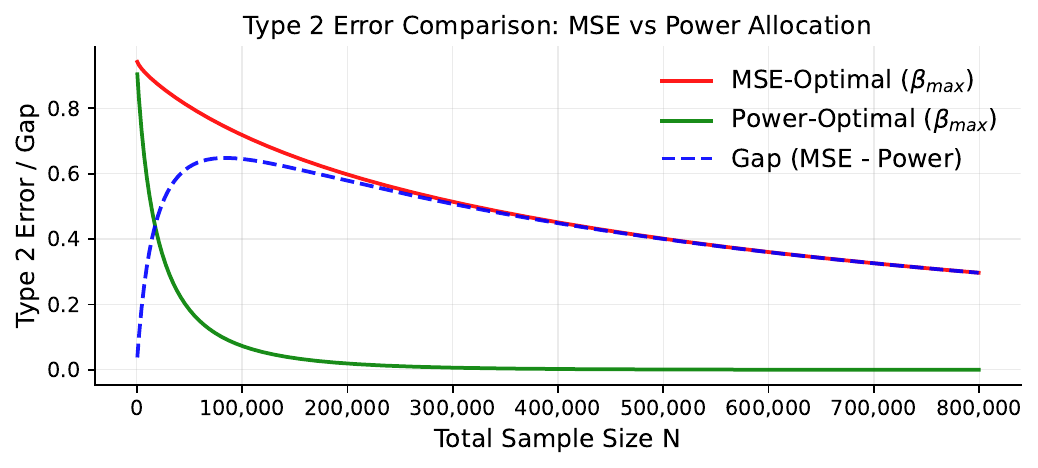} 
    \caption{Comparison of worst-case Type 2 error under power-optimal and MSE-optimal allocations as a function of \( N \). Parameters: \( M = 50 \), \( \alpha = 0.05 \), \( \sigma_i \) and \( \Delta_i \) are randomly generated from $[0.5,2]$ and $[0.01, 1]$, repeated for $R=1,000$ times.}
    \label{fig:type2-gap}
\end{figure}

Figure~\ref{fig:type2-gap} reports simulation results comparing the two
allocation strategies across different total sample sizes~$N$. The red
curve shows the worst-case Type~2 error under the benchmark
MSE-optimal allocation, whereas
the green curve shows the substantially lower error achieved by the
power-optimal allocation derived in
Section~\ref{subsec:power_opt_solution}. The blue dashed curve plots the
difference between the two and quantifies the efficiency gain from
optimizing directly for detection power.

The figure highlights three regimes as the total budget~$N$ varies.
When resources are scarce (small~$N$), neither allocation provides
enough samples to reliably detect effects. As a result, the Type~2
error under both methods approaches the trivial upper bound
$1-\alpha$, and the performance gap is negligible. At the other
extreme, where $N$ is very large, both
allocations achieve near-certain detection. As the Type~2 errors of
both methods converge to zero, the difference between them naturally
vanishes. The most pronounced contrast arises in the intermediate,
resource-constrained regime, which is often the most relevant in
practice. For example, at a total budget of $N = 80{,}000$—corresponding
to an average of $1{,}600$ samples per experiment—the MSE-optimal
allocation yields a worst-case Type~2 error of nearly $0.75$. In
contrast, the power-optimal allocation reduces this error to
approximately $0.10$, a gap of about $65$ percentage points. Even at
$N = 200{,}000$, the gap remains close to $0.60$, indicating that
MSE-based allocation requires substantially more traffic to achieve the
same level of detection performance as the power-optimal approach.

\section{Unknown $\vec \sigma$: Exact Analysis of a Two-Experiment Setting } \label{sec:unknown sigma 2exp}

We now turn to the more realistic case in which $\vec{\sigma}$ is unknown. 
Before addressing the fully general setting, we first study a simplified 
two-experiment case ($M = 2$) with identical pilot sample sizes 
($\epsilon_1 = \epsilon_2 = \epsilon$). 
The purpose of this section is to characterize the structure of the 
\emph{ideal} correction factor under an oracle benchmark in which the 
true standard deviations are known. 
Although this oracle solution is not directly implementable in practice, 
it yields sharp and interpretable insights into how optimal variance 
inflation depends on the relative difficulty indices 
$\sigma_i/\Delta_i$ and on the decision-maker’s risk criterion 
(\textbf{TOL}, \textbf{CONF}, or \textbf{EXP}).

Under this symmetric two-experiment setting, the optimization problems 
introduced in Section~3 reduce to a univariate problem in the ratio 
\[
r := \frac{k_1}{k_2}.
\]
By exploiting the symmetry of the $F$-distribution under equal degrees 
of freedom, we obtain closed-form and interpretable expressions for 
the optimal inflation in the \textbf{TOL} and \textbf{CONF} formulations. 
Although deriving an explicit solution for the \textbf{EXP} formulation 
is challenging, we show that the same qualitative structure of the 
optimal correction factor extends to that case as well.

\subsection{Analysis of \textbf{TOL} and \textbf{CONF}}

For analytical tractability, we start with providing alternative formulations of \textbf{TOL} and \textbf{CONF}.

Recall that we focus on allocations of the form $\vec{n} = \vec{n}^*(\vec{k}, \vec{S})$ for some $\vec{k}$, where $\vec{n}^*(\vec{k}, \vec{S})$ is given in (\ref{eq:n*(k,S)}). When $M = 2$, the corresponding maximum Type~2 error becomes:
\begin{eqnarray*}
\tilde{\beta}^*(\vec{k}, \vec{S}, \vec{\sigma})  
= \max_{i \in [2]} \big\{ \beta(\sigma_i, n_i^*(\vec{k}, \vec{S})) \big\} 
%& = & \Phi\left(q_{1-\alpha} - \sqrt{N} \cdot \min\left\{\frac{\Delta_1}{\sigma_1} \sqrt{\frac{k_1 \left(\frac{S_1}{\Delta_1} \right)^2}{\sum_{j=1}^2 k_j \left( \frac{S_j}{\Delta_j} \right)^2}}, \ \frac{\Delta_2}{\sigma_2} \sqrt{\frac{k_2 \left(\frac{S_2}{\Delta_2} \right)^2}{\sum_{j=1}^2 k_j \left( \frac{S_j}{\Delta_j} \right)^2}}  \right\} \right) \\[1mm]
%& = & \Phi\left(q_{1-\alpha} - \min_{i \in [2]}  \left\{\frac{\Delta_i}{\sigma_i}  \sqrt{n^*_i(\vec k, \vec S)} \right\} \right) \\[1mm]
& = & \Phi\left( q_{1-\alpha} - \sqrt{N} \cdot \sqrt{ \frac{1}{\max(U_1, U_2)} } \right),
\end{eqnarray*}

%\vspace{1mm}
\noindent
where
\[
U_1 := \left( \frac{\sigma_1}{\Delta_1} \right)^2 + \frac{1}{r} \cdot \left( \frac{\sigma_2}{\Delta_2} \right)^2 \cdot \frac{Y_2}{Y_1}, \qquad 
U_2 := \left( \frac{\sigma_2}{\Delta_2} \right)^2 + r \cdot \left( \frac{\sigma_1}{\Delta_1} \right)^2 \cdot \frac{Y_1}{Y_2}, 
\]

\vspace{1mm}
\noindent
and %$a_i := \left( \frac{\sigma_i}{\Delta_i} \right)^2$ and 
$Y_i \sim \chi^2_{\epsilon - 1}$ for $i \in [2]$. Thus, the probability $\mathbb{P}(\tilde{\beta}^*(\vec{k}, \vec{S}, \vec{\sigma}) \leq \beta^*(\vec{\sigma}) + \delta)$ can be expressed as $\mathbb{P}(\max(U_1, U_2) \leq d(\delta))$, where the critical threshold $d(\delta)$ is given by
\begin{eqnarray}
d(\delta) := \frac{N}{\left( q_{1-\alpha} - \Phi^{-1}(\beta^*(\vec\sigma) + \delta) \right)^2}. \label{eq:d(delta)}
\end{eqnarray}

\vspace{1mm}
\noindent
Note that since \( \tilde\beta^*(\vec k, \vec S, \vec \sigma) \le 1 - \alpha \) for all \( \vec k \ge \vec 1 \) and \( \vec S \), without loss of generality, we can focus on the case \( \delta \in (0, 1 - \alpha - \beta^*(\vec\sigma)) \). Focusing on this interval also has the benefit that $d(\delta)$ is monotonically increasing. Moreover, \(d(\delta) \in \big(\sum_{j \in [2]}(\frac{\sigma_j}{\Delta_j})^2, \infty \big)\) almost surely. 

Given the above set-up, \textbf{TOL} can be expressed as:
\begin{eqnarray*}
\delta^*(\gamma, \epsilon, \vec{\sigma}) \ \ 
:= \ \ \min_{r, \delta} && \delta \\
\text{subject to} && \mathbb{P}\left( \max(U_1, U_2) \leq d(\delta) \right) \geq \gamma, \\
&& r > 0, \quad \delta \geq 0
\end{eqnarray*}

\noindent
Similarly, \textbf{CONF} can be expressed as
\begin{eqnarray*}
\gamma^*(\delta, \epsilon, \vec{\sigma}) 
\ \ := \ \ \max_{r, \gamma} && \gamma \\
\mbox{subject to} && \mathbb{P}\left( \max(U_1, U_2) \leq d(\delta) \right) \geq \gamma, \\
&& r > 0, \quad \gamma \in (0,1)
\end{eqnarray*}

\noindent

%\vspace{1mm}
The probability $\mathbb{P}\left( \max(U_1, U_2) \leq d(\delta) \right)$ plays a central role in the analysis of \textbf{TOL} and \textbf{CONF}. In what follows, we first examine key properties of this probability function. First, note that the event \( \max(U_1, U_2) \le d(\delta) \) is equivalent to $U_1 \le d(\delta)$ and $U_2 \le d(\delta)$. Letting $a_i := (\frac{\sigma_i}{\Delta_i})^2$ for $i \in [2]$ and substituting the definitions of \( U_1 \) and \( U_2 \), this condition is equivalent to
\[
\frac{a_2}{r (d(\delta) - a_1)}  
 \le  
\Theta  
 \le  
\frac{d(\delta) - a_2}{r a_1},
\]

\vspace{1mm}
\noindent
where \( \Theta = \frac{Y_1}{Y_2} \sim F_{\nu, \nu} \) follows an $F$-distribution with degrees of freedom \( (\nu, \nu) \), and \( \nu = \epsilon - 1 \). Let \( F_{F_{\nu,\nu}} \) denote the cumulative distribution function (CDF) of this distribution. Then,
\begin{eqnarray*}
\mathbb{P} \left( \max(U_1, U_2) \le d(\delta) \right) 
& = & \mathbb{P} \left( 
\frac{a_2}{r (d(\delta) - a_1)} \le \Theta \le 
\frac{d(\delta) - a_2}{r a_1} 
\right) \\[1mm]
& = & F_{F_{\nu, \nu}} \left( \frac{d(\delta) - a_2}{r a_1} \right)
- F_{F_{\nu, \nu}} \left( \frac{a_2}{r (d(\delta) - a_1)} \right).
\end{eqnarray*}

\vspace{1mm}
\noindent
%Since \( \tilde\beta^*(\vec k, \vec S, \vec \sigma) \le 1 - \alpha \) almost surely for all \( \vec k \ge \vec 1 \) and \( \vec S \), we focus on the case \( \delta \in (0, 1 - \alpha - \beta^*(\vec\sigma)) \). In this interval, 
Since \( d(\delta) \) is increasing for \( \delta \in (0, 1 - \alpha - \beta^*(\vec\sigma)) \) (this maps to $d(\delta) \in \big( \sum_{j \in [2]} (\frac{\sigma_j}{\Delta_j})^2, \infty \big)$), we can re-parameterize the above expression as a function of \( (r,d) \) instead of \( (r,\delta) \). Define:
\[
H(r,d) \ := \ 
F_{F_{\nu, \nu}} \left( \frac{d - a_2}{r a_1} \right)
- F_{F_{\nu, \nu}} \left( \frac{a_2}{r (d - a_1)} \right).
\]

\vspace{1mm}
\noindent
The following lemma characterizes the properties of \( H(r,d) \).

\vspace{1mm}
\begin{lemma} \label{lem:H(r,d)}
The function \( H(r,d) \) has the following properties:
\vspace{1mm}
\begin{itemize}
\item[\emph{(i)}] For fixed \( r \in (0, \infty) \), \( H(r,d) \) is increasing in 
\( d \in \big( \sum_{j \in [2]} (\frac{\sigma_j}{\Delta_j})^2, \infty \big) \).

\item[\emph{(ii)}] For fixed \( d \in \big( \sum_{j \in [2]} (\frac{\sigma_j}{\Delta_j})^2, \infty \big) \), \( H(r,d) \) is unimodal in \( r \). Specifically, $\lim_{r \to 0^+} H(r,d) = 0$, $\lim_{r \to \infty} H(r,d) = 0$, and there exists a unique maximizer \( r(d) = \argmax_{r \in (0, \infty)} H(r,d)\) where % \in (0,\infty) \) such that $H(r^*(d), d) = \max_{r > 0} \ H(r, d)$ and $H(r^*(d), d) > 0$.
\begin{eqnarray*}
r(d) = \sqrt{\frac{a_2}{d-a_1} \cdot \frac{d-a_2}{a_1}}.
\end{eqnarray*}
\end{itemize}
\end{lemma}

\vspace{1mm}
Property~(i) implies that, for any fixed $r$, $H(r,d)$ is nondecreasing in $d$; hence larger $d$ yields larger (or equal) values of $H(r,d)$, making the constraint $H(r,d)\ge \gamma$ easier to satisfy.
 Property~(ii) highlights a key structural feature: for any fixed~$d$ in the given interval, the function $H(r,d)$ is unimodal in~$r$, attaining a unique maximum at some $r^*(d) \in (0, \infty)$. This implies that among all values of $r$, only $r(d)$ maximizes the probability $H(r,d)$. Together, these two properties provide the foundation for identifying the smallest $d$ (call it $d^*$) such that $H(r,d) \ge \gamma$ for some $r > 0$. As we will show below, the value of $d^*$ must correspond to choosing $r = r(d^*)$, thereby enabling an explicit characterization of the optimal tolerance.

Using Lemma~\ref{lem:H(r,d)}, we next establish an explicit expression for the optimal inflation ratio \( r^* \) in the case with two experiments and equal pilot sizes.

\vspace{1mm}
\begin{proposition} \label{prop:optimal_k_M2}
Let $F^{-1}_{\nu, \nu}(p)$ denote the $p$-quantile of the $F$-distribution 
with degrees of freedom $(\nu, \nu)$, where $\nu = \epsilon - 1$. Then, in the two-experiment setting, we have:
\vspace{1mm}
\begin{enumerate}
    \item[\emph{(i)}] For \textbf{TOL} with target confidence $\gamma \in (0,1)$, the tolerance $\delta$ is minimized at 
    %\begin{eqnarray*}
    $r^* = \sqrt{m_1 m_2}$, 
    %\end{eqnarray*}
    %\noindent
    where $m_1 = \frac{a_2}{d^* - a_1}$,  $m_2 = \frac{d^* - a_2}{a_1}$, and the critical threshold $d^* = d(\delta^*(\gamma, \epsilon,\vec \sigma))$ is given by
    \[
    d^* = \frac{a_1 + a_2 + \sqrt{(a_1 - a_2)^2 + 4 a_1 a_2 \cdot \big(F^{-1}_{\nu, \nu}(\frac{1+\gamma}{2})\big)^2}}{2}.
    \]

     \vspace{1mm}
    \item[\emph{(ii)}] For \textbf{CONF} with threshold $\beta^*(\vec\sigma) + \delta \in (\beta^*(\vec\sigma), 1-\alpha)$, the confidence level $\gamma$ is maximized at 
    %\begin{eqnarray*}
    $r^* = \sqrt{m_1 m_2}$,  
    %\end{eqnarray*}
    %\noindent
    where $m_1 = \frac{a_2}{d(\delta) - a_1}$ and $m_2 = \frac{d(\delta) - a_2}{a_1}$.
   \end{enumerate}
\end{proposition}

\vspace{1mm}
Note that the optimal ratio $r^* = \sqrt{m_1 m_2}$ depends on the relative values of $a_1$ and $a_2$, which reflect the statistical difficulty of the two experiments. When $a_1 = a_2$, we have $m_1 = m_2$ and hence $r^* = 1$. As the gap between $a_1$ and $a_2$ widens, however, $m_1$ and $m_2$ become asymmetric, pulling $r^*$ away from $1$. This result clarifies why the naive plug-in approach (see Section \ref{sec:model}), which treats the sample standard deviations \( \vec{S} \) as if they were the true variances and implicitly assumes \( k_1 = k_2 = 1 \), i.e., \( r = 1 \), can be suboptimal. Proposition \ref{prop:optimal_k_M2} shows that the optimal inflation ratio generally differs from 1, especially when the experiments differ in statistical difficulty. This highlights the value of using a calibrated correction factor rather than relying on the naive plug-in method.

The following corollary formalizes this behavior by describing how \( r^* \) deviates from 1 based on the relative statistical difficulty of the two experiments (see also Figure~\ref{fig:r_fig}) .

\vspace{1mm}
\begin{corollary} \label{cor:optimal_2exp}
The following hold for \textbf{TOL} and \textbf{CONF}. In the two-experiment setting,
\begin{enumerate}
\item[\emph{(i)}] If $\frac{\sigma_1}{\Delta_1} = \frac{\sigma_2}{\Delta_2}$, then $r^* = 1$;
\item[\emph{(ii)}] If $\frac{\sigma_1}{\Delta_1} < \frac{\sigma_2}{\Delta_2}$, then $r^* > 1$;
\item[\emph{(iii)}] If $\frac{\sigma_1}{\Delta_1} > \frac{\sigma_2}{\Delta_2}$, then $r^* < 1$.
\end{enumerate}
\end{corollary}

\vspace{1mm}
This corollary offers key insights into the optimal inflation strategy in the two-experiment setting. When experiment~1 is statistically easier (more difficult) than experiment~2—i.e., it has a smaller (larger) variance-to-signal ratio—then the optimal inflation ratio satisfies $r^* = \frac{k_1}{k_2} > (<) 1$. That is, the correction factor for the easier (more difficult) experiment is inflated more (less) than that of the harder (easier) one. At first glance, this may seem counterintuitive. One might expect the more difficult or uncertain experiment to receive greater inflation. However, the objective is not to guard each experiment individually, but to control the probability that the maximum Type~2 error across the two experiments exceeds a target threshold. To understand why this leads to $r^* > 1$, recall that the worst-case Type~2 error depends on the random quantities $U_1$ and $U_2$, given by
\[
U_1 = a_1 + \frac{1}{r} a_2 \cdot \frac{Y_2}{Y_1}, \qquad 
U_2 = a_2 + r a_1 \cdot \frac{Y_1}{Y_2},
\]

%\vspace{1mm}
\noindent
where $Y_1, Y_2 \sim \chi^2_{\epsilon - 1}$ independently. Although the ratios $\frac{Y_2}{Y_1}$ and $\frac{Y_1}{Y_2}$ have symmetric distributions, the asymmetry in $a_1$ and $a_2$ breaks this balance. Since the more difficult experiment (larger $a_i$) amplifies the impact of random fluctuations in these ratios, the optimizer reduces the risk by using a larger correction for the easier experiment. For example, if \( a_1 > a_2 \), then the term \( a_1 \frac{Y_1}{Y_2} \) is more volatile than \( a_2 \frac{Y_2}{Y_1} \), leading the optimizer to set  \( r < 1 \). This reduces the chance that either $U_1$ or $U_2$ becomes too large. In short, the easier experiment is deliberately over-inflated  to stabilize the overall variability and minimize the maximum Type~2 error across the two experiments. 

The previous corollary characterize how the optimal inflation ratio depends on the relative statistical difficulty of the experiments. The following corollary introduces how the optimal inflation ratio also adapts to the experimenter's tolerance or confidence preference.  

\vspace{1mm}
\begin{corollary} \label{cor:r_with_delta}
Suppose \( \frac{\sigma_1}{\Delta_1} < \frac{\sigma_2}{\Delta_2} \). In the two-experiment setting, %we have:
\begin{enumerate}
    \item[\emph{(i)}] For \textbf{TOL}, %\( r^* > 1 \) for all \( \gamma \in (0, 1) \), and \
 \(r^* \) is increasing in the confidence level \( \gamma \in (0,1) \);
    
    \item[\emph{(ii)}] For \textbf{CONF}, %\( r^* > 1 \) for all \( \delta \in (0, 1 - \alpha - \beta^*(\vec{\sigma})) \), and 
    \( r^* \) is increasing in the tolerance \( \delta \in (0, 1 - \alpha - \beta^*(\vec\sigma)) \).
\end{enumerate}
\noindent
By symmetry, if  \( \frac{\sigma_1}{\Delta_1} > \frac{\sigma_2}{\Delta_2} \), then $r^*$ is decreasing in $\gamma$ (for \textbf{TOL}) and $\delta$ (for \textbf{CONF}).
\end{corollary}

\vspace{1mm}
This corollary shows that the optimal inflation ratio is driven not only by the statistical difficulty of the experiments, but also by the  decision-maker's risk preference—reflected in the choice of confidence  level \( \gamma \) (for \textbf{TOL}) or tolerance \( \delta \) (for \textbf{CONF}). Under the assumed ordering \( \frac{\sigma_1}{\Delta_1} < \frac{\sigma_2}{\Delta_2} \),  experiment 1 is statistically easier than experiment 2, so  
Corollary~\ref{cor:optimal_2exp} implies that \( r^* > 1 \).  As the experimenter demands higher reliability (larger \( \gamma \)) or permits greater deviation from optimality (larger \( \delta \)), the  critical threshold—\( d(\delta^*) \) for \textbf{TOL} or \( d(\delta) \)  for \textbf{CONF}—increases. A larger threshold makes it easier for the condition \( \max(U_1, U_2) \le d(\delta^*) \) (or \( d(\delta) \)) to be satisfied. Mathematically, this means that the set of feasible inflation ratios \( r \)  
satisfying the constraint  
\[
\mathbb{P}( \max(U_1, U_2) \le d(\delta^*) ) \ge \gamma 
\quad \text{or} \quad 
\mathbb{P}( \max(U_1, U_2) \le d(\delta) ) \ge \gamma
\]

%\vspace{1mm}
\noindent
becomes strictly larger as the critical threshold increases. To take advantage of this added flexibility, the optimizer adjusts  the inflation ratio \( r^* \) further away from 1—by inflating the easier  experiment even more relative to the harder one. This asymmetry effectively shifts risk away from the harder experiment, reducing the probability that either \( U_1 \) or \( U_2 \) becomes too large, and  thereby helps maintain tighter control over the collective maximum Type 2 error.
Figure~\ref{fig:r_fig} visualizes the key insights from Corollaries~\ref{cor:optimal_2exp} and \ref{cor:r_with_delta}, showing how the optimal inflation ratio varies with both the relative statistical difficulty of the experiments and the experimenter's risk preference---captured by the confidence level $\gamma$ in \textbf{TOL} and the tolerance $\delta$ in \textbf{CONF}.

\begin{figure}[!htbp]
    \centering
    \subfloat[\textbf{TOL} Objective ($\epsilon=20$)]{%
        \includegraphics[width=0.7\textwidth]{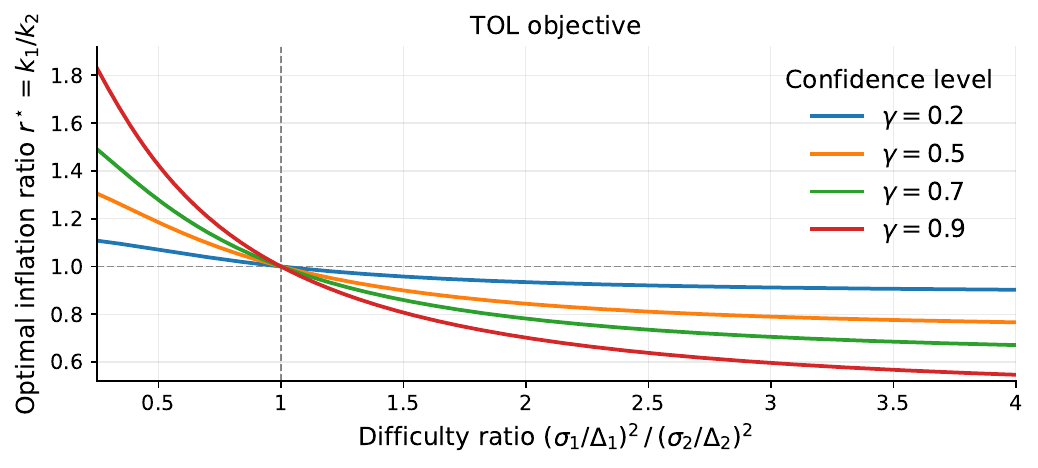}%
        \label{fig:cor1_tol}%
    }

    \vspace{0.5cm}

    \subfloat[\textbf{CONF} Objective ($N=200, \epsilon=20$)]{%
        \includegraphics[width=0.7\textwidth]{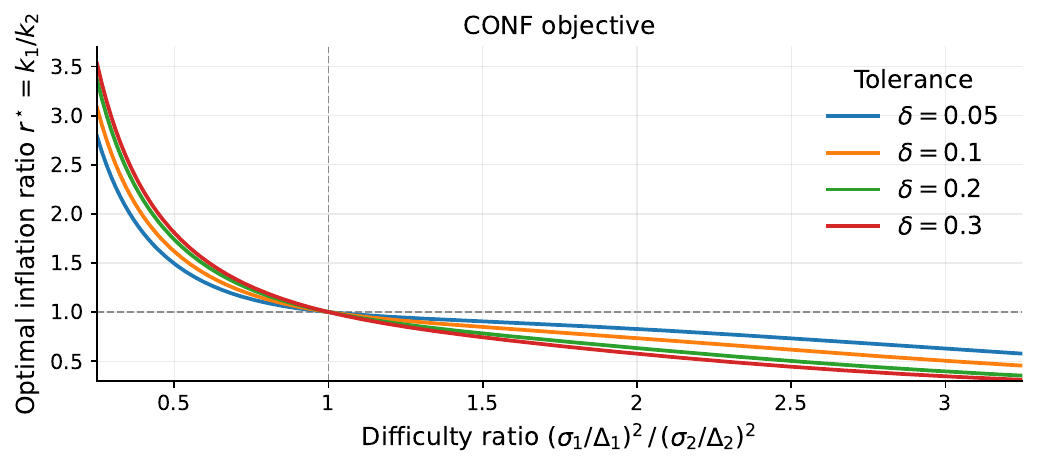}%
        \label{fig:cor1_conf}%
    }

    \caption{Optimal inflation ratio $r^\star$ for the \textbf{TOL} and \textbf{CONF}
    objectives under varying difficulty ratios ($a_1/a_2$), confidence levels,
    and tolerances.}
    \label{fig:r_fig}
\end{figure}

\subsection{Analysis of \textbf{EXP}}

We now discuss \textbf{EXP}. Based on our discussion in the previous subsection, it can be expressed as
\begin{eqnarray*}
g^*(\epsilon, \vec{\sigma}) \,\, := \,\, \min_{r} && 
\mathbb{E}\!\left[
\Phi\!\left(
q_{1-\alpha}
- \sqrt{N}\cdot \sqrt{\frac{1}{\max(U_1, U_2)}}
\right)
\right] \\
\text{subject to } && r > 0.
\end{eqnarray*}

\noindent
Unlike \textbf{TOL} and \textbf{CONF}, it is intractable to derive the exact solution to \textbf{EXP}. But we show that the key insight of Corollary \ref{cor:optimal_2exp} for \textbf{TOL} and \textbf{CONF} similarly holds for \textbf{EXP}.

\vspace{1mm}
\begin{proposition} \label{prop:optimal_exp}
The following hold for \textbf{EXP}. In the two-experiment setting, we have:
\begin{enumerate}
\item[\emph{(i)}] If $\frac{\sigma_1}{\Delta_1} = \frac{\sigma_2}{\Delta_2}$, then $r^* = 1$;
\item[\emph{(ii)}] If $\frac{\sigma_1}{\Delta_1} < \frac{\sigma_2}{\Delta_2}$, then $r^* > 1$;
\item[\emph{(iii)}] If $\frac{\sigma_1}{\Delta_1} > \frac{\sigma_2}{\Delta_2}$, then $r^* < 1$.
\end{enumerate}
\end{proposition}

\vspace{1mm}
The proof of Proposition~\ref{prop:optimal_exp} is technically involved and therefore deferred to the Appendix. Its intuition parallels that of Corollary~\ref{cor:optimal_2exp}, even though the 
\textbf{EXP} formulation minimizes the \emph{expected} maximum Type~2 error rather than 
imposing a high-probability guarantee. The optimizer seeks to balance the
distribution of the random variables $U_1$ and $U_2$, whose maxima
determine the realized power. When the two experiments have equal
difficulty indices $\frac{\sigma_i}{\Delta_i}$, symmetry implies $r^* = 1$.
When one experiment is statistically easier, say
$\frac{\sigma_1}{\Delta_1} < \frac{\sigma_2}{\Delta_2}$, the harder
experiment contributes more volatility to $\max(U_1,U_2)$. To mitigate
this asymmetry in expectation, the optimizer deliberately inflates the
easier experiment more ($r^*>1$), thereby smoothing the distribution of
$\max(U_1,U_2)$ and reducing its tail heaviness. In short, although
\textbf{EXP} focuses on average-case performance rather than
probabilistic guarantees, the same structural principle emerges as in
\textbf{TOL} and \textbf{CONF}.

Figure~\ref{fig:exp} corroborates these theoretical insights. The plot displays the optimal inflation ratio $r^*$ against the difficulty ratio $(\sigma_1/\Delta_1)^2/(\sigma_2/\Delta_2)^2$ for pilot sample sizes $\epsilon \in \{20, 50, 100, 500\}$. Consistent with Proposition~\ref{prop:optimal_exp}, an inverse relationship emerges: as $e_1$ becomes relatively easier (difficulty ratio $< 1$), the optimal strategy prescribes stronger inflation ($r^* > 1$). Critically, the figure highlights the role of estimation uncertainty. For the smallest pilot size ($\epsilon=20$), $r^*$ deviates significantly from unity, reflecting the need for asymmetric correction to buffer against pilot variability. As $\epsilon$ increases to $500$, the curve flattens toward $r^* \approx 1$, confirming that asymmetric inflation specifically mitigates small-sample risk; as pilot precision improves, the corrective mechanism attenuates.

\begin{figure}[htbp]
    \centering
    \includegraphics[width=0.7\textwidth]{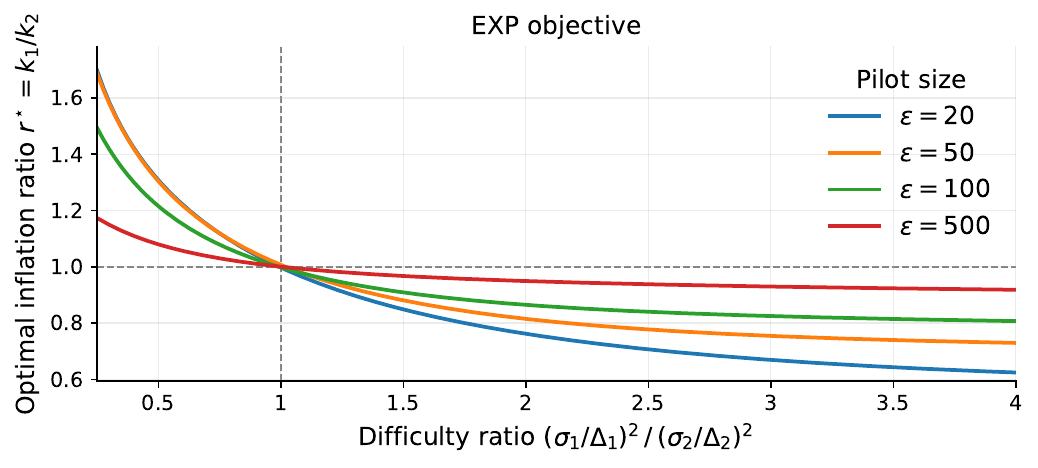}
    \caption{Optimal inflation ratio $r^*$ under \textbf{EXP} $(N=200)$ as a function of difficulty ratio $(a_1/a_2)$ for varying pilot sizes $\epsilon$. The deviation from $r^*=1$ is most pronounced for small $\epsilon$.}
    \label{fig:exp}
\end{figure}

\section{Unknown $\vec\sigma$: Approximations for the General Case} \label{sec:approximate k}

We now develop tractable approximations for the general case with $M$ 
experiments. While Section~5 characterized the structure of the oracle 
correction factor in a stylized two-experiment setting, directly optimizing 
the \textbf{TOL}, \textbf{CONF}, or \textbf{EXP} formulations in Section~\ref{sec:model} 
becomes computationally challenging in large portfolios. The resulting 
problems involve high-dimensional stochastic objectives or chance 
constraints over the pilot randomness, making direct optimization 
difficult to scale. To address this challenge, we construct surrogate reformulations inspired 
by robust optimization. The key idea is to replace the original stochastic 
criteria with deterministic upper bounds that preserve the structural 
logic of optimal variance inflation while avoiding repeated evaluation of 
complex probability expressions. 

In this section, we proceed in three steps. We first present preliminary 
observations that motivate our reformulations. We then develop and analyze 
the surrogate problems, establishing their structural properties and 
tractability. Finally, we leverage these insights to construct a fully 
data-dependent implementation, in which the unknown standard deviations 
are replaced by their pilot-based estimates and the resulting surrogate 
problem is solved directly.

\subsection{Preliminary observations}

To see the connection between the expression of $\vec n^*(\vec k, \vec S)$ in (\ref{eq:n*(k,S)})  and robust optimization, consider the following uncertainty set for $S_i$, which corresponds to the confidence interval for the variance $\sigma_i^2$ at confidence level $c_i \in [0,1)$:
\begin{eqnarray*}
CI_i(S_i, \epsilon_i, c_i) = \left[ \frac{(\epsilon_i-1)S_i^2}{\chi^2_{1 - (1 - c_i)/2, \epsilon_i-1}}, \frac{(\epsilon_i-1)S_i^2}{\chi^2_{(1 - c_i)/2, \epsilon_i-1}} \right] = \left[ \frac{(\epsilon_i-1)S_i^2}{\chi^2_{(1 + c_i)/2, \epsilon_i-1}}, \frac{(\epsilon_i-1)S_i^2}{\chi^2_{(1 - c_i)/2, \epsilon_i-1}} \right],
\end{eqnarray*}

%\vspace{1mm}
\noindent
where $\chi^2_{\xi,n}$ is the $\xi$-quantile of chi-squared distribution with $n$ degrees of freedom. For each $i \in [M]$, define the lower and upper scaling factors 
$\underline{\varphi}(\epsilon_i, c_i)$ and 
$\overline{\varphi}(\epsilon_i, c_i)$ as follows:
\begin{eqnarray*}
\underline{\varphi}(\epsilon_i, c_i) 
&:=& \frac{\epsilon_i - 1}{\chi^2_{(1 + c_i)/2, \ \epsilon_i - 1}}, 
\quad \text{and} \quad
\overline{\varphi}(\epsilon_i, c_i) 
:= \frac{\epsilon_i - 1}{\chi^2_{(1 - c_i)/2, \ \epsilon_i - 1}}.
\end{eqnarray*}

%\vspace{1mm}
\noindent
Using this notation, the confidence interval for the variance 
$\sigma_i^2$ can be rewritten as:
\begin{eqnarray*}
CI_i(S_i, \epsilon_i, c_i) 
= \left[ 
\underline{\varphi}(\epsilon_i, c_i) S_i^2, \ 
\overline{\varphi}(\epsilon_i, c_i) S_i^2 
\right].
\end{eqnarray*}

%\vspace{1mm}
Now, consider the following robust optimization problem:
\begin{eqnarray*}
\textbf{R-POWER-OPT:} \hspace{10mm} 
\tilde{\beta}^{robust}(\vec{S}, {\vec{\epsilon}}, \vec{\sigma}, \vec{c}) 
:= \min_{\vec{n}} && \max_{i \in [M]} \ 
\left\{ 
\max_{\tilde{\sigma}_i^2 \in CI_i(S_i, \epsilon_i, c_i)} 
\beta(\tilde{\sigma}_i, n_i) 
\right\} \\
\text{subject to} && \sum_{i=1}^M n_i \le N, \quad n_i \ge 0 \ \ \forall i \in [M],
\end{eqnarray*}

%\vspace{1mm}
\noindent
Since $\beta(\tilde{\sigma}_i, n_i)$ is increasing in $\tilde{\sigma}_i$, 
the worst case within the uncertainty set occurs at the upper bound. 
Thus, the problem simplifies to:
\begin{eqnarray*}
\textbf{R-POWER-OPT:} \hspace{10mm} 
\tilde{\beta}^{robust}(\vec{S}, {\vec{\epsilon}}, \vec{\sigma}, \vec{c}) 
= \min_{\vec{n}} && \max_{i \in [M]} \ \{
\beta( \sqrt{\overline{\varphi}(\epsilon_i, c_i)} S_i, \ n_i) \} \\
\text{subject to} 
&& \sum_{i=1}^M n_i \le N, \quad n_i \ge 0 \ \ \forall i \in [M]
\end{eqnarray*}

%\vspace{1mm}
\noindent
This is equivalent to the original \textbf{POWER-OPT} formulation after replacing $\sigma_i$ with $\sqrt{\overline{\varphi}(\epsilon_i, c_i)} S_i$ for each $i \in [M]$. In other words, the optimal solution to \textbf{R-POWER-OPT} is given by $\vec{n}^*(\vec{\ell}, \vec{S})$, where $\ell_i := \overline{\varphi}(\epsilon_i, c_i)$ for all $i \in [M]$. The corresponding optimal value is:
\begin{eqnarray*}
\tilde{\beta}^{\text{robust}}(\vec{S}, {\vec{\epsilon}}, \vec{\sigma}, \vec{c}) 
= \tilde{\beta}^*(\vec{\ell}, \vec{S}, \vec{\Sigma}(\vec{\ell}, \vec{S})),
\end{eqnarray*}

%\vspace{1mm}
\noindent
where we define 
$\Sigma_i(\vec{\ell}, \vec{S}) := \sqrt{\ell_i} S_i$ for each $i \in [M]$. 
Under this notation, the true standard deviation vector can be written as 
$\vec{\Sigma}(\vec{1}, \vec{\sigma})$.

We remark that $\tilde{\beta}^*(\vec{\ell}, \vec{S}, \vec{\sigma})$ (a random quantity induced by the randomness of $\vec{S}$) is now effectively a fixed quantity analogous to \textbf{POWER-OPT} by only considering values of $\sigma$ in the confidence set constructed above. This allows us to upper bound  $\tilde{\beta}^*(\vec{\ell}, \vec{S}, \vec{\sigma})$ by $\tilde{\beta}^{\text{robust}}(\vec{S}, {\vec{\epsilon}}, \vec{\sigma}, \vec{c})$, with a high probability. To formalize this, define the event 
$\mathcal{E}(\vec{S}, {\vec{\epsilon}}, \vec{c})$ as follows:
\begin{eqnarray}
\mathcal{E}(\vec{S}, {\vec{\epsilon}}, \vec{c}) 
:= \bigcap_{i=1}^M \left\{ 
\sigma_i^2 \in CI_i(S_i, \epsilon_i, c_i) 
\right\}. \label{event E}
\end{eqnarray}

%\vspace{1mm}
\noindent
That is, $\mathcal{E}(\vec{S}, {\vec{\epsilon}}, \vec{c})$ is the event that each true variance $\sigma_i^2$ lies within its respective confidence interval $CI_i(S_i, \epsilon_i, c_i)$. The next lemma establishes the desired probabilistic guarantee, which we have obtained by construction.

\vspace{1mm}
\begin{lemma} \label{lem_A1}
Suppose that $\ell_i = \overline{\varphi}_i(\epsilon_i, c_i)$ for all 
$i \in [M]$. On the set $\mathcal{E}(\vec{S}, {\vec{\epsilon}}, \vec{c})$, 
we have 
$
\tilde{\beta}^*(\vec{\ell}, \vec{S}, \vec{\sigma}) 
\le 
\tilde{\beta}^{\text{robust}}(\vec{S}, {\vec{\epsilon}}, \vec{\sigma}, \vec{c}).
$
Moreover, this event occurs with probability 
$\mathbb{P}(\mathcal{E}(\vec{S}, {\vec{\epsilon}}, \vec{c})) = \prod_{i=1}^M c_i$.
\end{lemma}

\vspace{1mm}
Lemma~\ref{lem_A1} provides the foundation for constructing tractable
approximations of the optimal correction vector $\vec{k}$.
Rather than analyzing the random quantity
$\tilde{\beta}^*(\vec{\ell},\vec{S},\vec{\sigma})$ directly, we
consider the auxiliary random variable
$\tilde{\beta}^{\text{robust}}(\vec{S},\vec{\epsilon},\vec{\sigma},\vec{c})$,
which serves as a high-probability upper bound on
$\tilde{\beta}^*(\vec{\ell},\vec{S},\vec{\sigma})$ with coverage probability
$\prod_i c_i$. To operationalize this idea, we set each correction factor
$k_i$ to the deterministic quantity
\[
k_i=\ell_i:=\overline{\varphi}(\epsilon_i,c_i), \qquad i\in[M],
\]
and parameterize the search over the corresponding confidence levels
$\vec{c}$.

The feasible domain of $\vec{c}$ can be characterized from the following
observations. As noted in Remark~\ref{remark1}, only the ratios
$\tfrac{k_i}{k_j}$ matter rather than their absolute magnitudes.

\vspace{1mm}
The next lemma shows that $\tilde{\beta}^{\text{robust}}(\vec{S},\epsilon,\vec{\sigma},\vec{c})$ can be bounded above by a closed-form expression independent of the random vector~$\vec{S}$. This property will play a crucial role in the surrogate reformulations developed in the next subsection.

\vspace{1mm}
\begin{lemma} \label{lem_A2}
Let $\kappa(\epsilon_i, c_i)
:= \frac{\overline{\varphi}(\epsilon_i, c_i)}{\underline{\varphi}(\epsilon_i, c_i)}$ for all $i \in [M]$. On the event $\mathcal{E}(\vec{S}, {\vec{\epsilon}}, \vec{c})$, we have:
\begin{eqnarray*}
\Phi\left(
q_{1-\alpha} - 
\sqrt{
\frac{N}
{\sum_{i=1}^M \left( \frac{\sigma_i}{\Delta_i} \right)^2}
}
\right)
\ \le \ 
\tilde{\beta}^{\text{robust}}(\vec{S}, {\vec{\epsilon}}, \vec{\sigma}, \vec{c})
\ \le \ 
\Phi\left(
q_{1-\alpha} - 
\sqrt{
\frac{N}
{\sum_{i=1}^M 
\kappa_i(\epsilon_i, c_i) 
\left( \frac{\sigma_i}{\Delta_i} \right)^2}
}
\right),
\end{eqnarray*}
\end{lemma}

\vspace{1mm}
The next lemma highlights useful properties of $\kappa(\epsilon_i, c_i)$. %The proof follows directly from the definition and standard properties of the chi-squared distribution, and is omitted for brevity.

\vspace{1mm}
\begin{lemma} \label{lem_kappa}
The function $\kappa(\epsilon_i, c_i)$ has the following properties: 
\begin{enumerate}
\item[\emph{(i)}] For all $c_i \in [0,1)$, $\lim_{\epsilon_i \to \infty}$ $\kappa(\epsilon_i, c_i) = 1$. 

\item[\emph{(ii)}] For all $\epsilon_i \in [2, \infty)$, $\kappa(\epsilon_i, c_i)$ is continuous, differentiable and convex in $c_i$, increasing in $c_i$ with $\lim_{c_i \to 0}$ $ \kappa(\epsilon_i, c_i) = 1$ and $\lim_{c_i \to 1} \kappa(\epsilon_i, c_i) = \infty$.

\item[\emph{(iii)}] Consider the equation $\kappa(\epsilon_i, c_i) = x$ for a fixed $x \in (1, \infty)$. As $\epsilon_i \to \infty$, we must have $c_i \to 1$. In other words, if we consider $c_i$ as a function of $\epsilon_i$ and $x$, we have $\lim_{\epsilon_i \to \infty} c_i(\epsilon_i, x) = 1$. 
\end{enumerate}
\end{lemma}

\vspace{1mm}
By Lemma~\ref{lem_kappa}, the upper and lower bounds in Lemma~\ref{lem_A2} coincide with $\beta^*(\sigma)$ as $\epsilon_i \to \infty$ for all $i \in [M]$. Since the lower bound equals the true optimum $\beta^*(\vec{\sigma})$, this implies that $\tilde{\beta}^{\text{robust}}(\vec{S}, \vec{\epsilon}, \vec{\sigma}, \vec{c})$ is a highly accurate approximation to $\beta^*(\vec{\sigma})$ when the pilot sample sizes are sufficiently large. % when the number of observations per experiment is sufficiently large.
We are now ready to introduce our reformulated optimization frameworks.
\subsection{Surrogate Reformulations of  \textbf{TOL}, \textbf{CONF}, and \textbf{EXP}}
The key idea behind our reformulations is to leverage the upper bound 
established in Lemma~\ref{lem_A2} to construct tractable approximations—either upper 
or lower bounds, depending on the objective—for the original formulations.
Below, we introduce surrogate reformulations of the three problems \textbf{TOL}, \textbf{CONF}, and \textbf{EXP}, which we denote by \textbf{R-TOL}, \textbf{R-CONF}, and \textbf{R-EXP}, respectively.
\vspace{1mm}
\begin{eqnarray*}
\textbf{R-TOL:} \hspace{10mm} \delta^R(\gamma, \vec{\epsilon}, \vec \sigma) \ := \ \min_{\vec{c}} && \Phi\left(q_{1-\alpha} - \sqrt{\frac{N}{\sum_{i=1}^M \kappa(\epsilon_i, c_i) \left(\frac{\sigma_i}{\Delta_i}\right)^2}}\right) - \beta^*(\vec\sigma) \\
\text{subject to} && 
\prod_{i=1}^M c_i \ge \gamma, \\
&& c_i \in [0,1) \ \ \forall i \in [M] \\[2mm]
\textbf{R-CONF:} \hspace{10mm} \gamma^R(\delta, \vec{\epsilon}, \vec \sigma) \ := \ \max_{\vec{c}} && \prod_{i=1}^M c_i \\
\text{subject to} && 
\delta \ \ge \ \Phi\left(q_{1-\alpha} - \sqrt{\frac{N}{\sum_{i=1}^M \kappa(\epsilon_i, c_i) \left(\frac{\sigma_i}{\Delta_i}\right)^2}}\right) - \beta^*(\vec\sigma), \\[1mm]
&&  c_i \in [0,1)  \ \ \forall i \in [M] \\[2mm]
\textbf{R-EXP:} \hspace{10mm} g^R(\vec{\epsilon}, \vec \sigma) \ := \ \min_{\vec{c}} && 
1 + \left[ \Phi\left(q_{1-\alpha} - \sqrt{\frac{N}{\sum_{i=1}^M \kappa(\epsilon_i, c_i) \left(\frac{\sigma_i}{\Delta_i}\right)^2}}\right) - 1\right] \cdot \prod_{i=1}^M c_i \\[1mm]
\text{subject to} && 
 c_i \in [0,1)  \ \ \forall i \in [M]
\end{eqnarray*}

\vspace{1mm}
A key advantage of \textbf{R-TOL} and \textbf{R-CONF} is that, after the reparameterization $x_i=\log c_i$, both problems reduce to deterministic convex programs with separable structure and a single coupling constraint; see \eqref{alternate R_TOL} and \eqref{alternate R_CONF} below. 
Moreover, \textbf{R-EXP} inherits the same tractability: for any fixed value of $\prod_{i=1}^M c_i=\gamma$, the inner minimization over $\vec c$ coincides with \textbf{R-TOL}, so \textbf{R-EXP} can be solved efficiently via a one-dimensional search over $\gamma$ with convex subproblems.
In contrast, directly optimizing \textbf{TOL}/\textbf{CONF}/\textbf{EXP} requires repeated evaluation of probabilities or expectations over the pilot randomness $\vec S$, which generally has no closed form for large $M$ and must be approximated numerically, leading to nested simulation and poor scalability. 

With these computational benefits in mind, we next define how each surrogate formulation induces a concrete set of correction factors (and hence an allocation rule). 
Let \( \vec{c}^{\pi} \) denote an optimal solution to formulation \( \pi \), where 
\( \pi \in \{\textbf{R-TOL}, \textbf{R-CONF}, \textbf{R-EXP}\} \). The corresponding correction vector \( \vec{k}^{\pi} \) is then defined entrywise by 
\[
k^{\pi}_i  :=  \overline{\varphi}(\epsilon_i, c^{\pi}_i) 
\quad \text{for all } i \in [M].
\]

\noindent
For clarity, we use $\vec{k}^{\,\mathrm{orig},\pi}$ to denote the optimal solution to the original formulation $\pi \in \{\textbf{TOL},$ $ \textbf{CONF}, \textbf{EXP}\}$. In the remainder of this subsection, we will analyze the desirable properties of our proposed surrogate reformulation and how it can closely ``approximate'' the original intractable formulations—\textbf{TOL}, \textbf{CONF}, and \textbf{EXP}. 

We begin by characterizing the relationship between \textbf{TOL} and \textbf{R-TOL}.

\vspace{1mm}
\begin{proposition}[\textbf{TOL vs. R-TOL}] \label{prop_A1}
Suppose $\gamma \in (0,1)$. The pair $(\vec k, \delta) = (\vec k^{\textbf{R-TOL}},  \delta^R(\gamma, \vec{\epsilon}, \vec{\sigma}))$ is feasible for \textbf{TOL} and, therefore,  $\delta^*(\gamma, \vec{\epsilon}, \vec{\sigma}) 
\le \delta^R(\gamma, \vec{\epsilon}, \vec{\sigma})$. Moreover, 
\begin{enumerate}
\item[\emph{(i)}] As $\gamma \to 0$, %$c^{\textbf{R-TOL}}_i \to 0$ for all $i \in [M]$, 
$\delta^*(\gamma, \vec{\epsilon}, \vec{\sigma}) \to 0$ and $\delta^R(\gamma, \vec{\epsilon}, \vec{\sigma}) \to 0$;
\item[\emph{(ii)}] As $\gamma \to 1$, %$c^{\textbf{R-TOL}}_i \to 1$ for all $i \in [M]$, 
$\delta^*(\gamma, \vec{\epsilon}, \vec{\sigma}) \to 1 - \alpha - \beta^*(\vec \sigma)$  and $\delta^R(\gamma, \vec{\epsilon}, \vec{\sigma}) \to 1 - \alpha - \beta^*(\vec \sigma)$;
\item[\emph{(iii)}] As $\epsilon_i \to \infty$ for all $i \in [M]$, $\delta^*(\gamma, \vec{\epsilon}, \vec{\sigma}) \to 0$ and $\delta^R(\gamma, \vec{\epsilon}, \vec{\sigma}) \to 0$.
\end{enumerate}    
\end{proposition}

\vspace{1mm}
Proposition~\ref{prop_A1} tells us that the optimal solution of \textbf{R-TOL} 
is always feasible for \textbf{TOL}. Indeed, by Lemma \ref{lem_A2}, any pair $(\vec k, \delta)$ with $k_i = \overline{\varphi}(\epsilon_i, c_i)$ for all $i \in [M]$ and
\begin{eqnarray*}
\delta 
= \Phi\left( 
q_{1-\alpha} - 
\sqrt{ 
\frac{N}
{\sum_{i=1}^M \kappa(\epsilon_i, c_i) 
\left( \frac{\sigma_i}{\Delta_i} \right)^2}
}
\right) - \beta^*(\vec\sigma)
\end{eqnarray*}

%\vspace{1mm}
\noindent
for some $\vec c$ satisfying $\prod_{i=1}^M c_i \ge \gamma$ is a feasible solution for \textbf{TOL}. Consequently, we have $\delta^*(\gamma, \vec{\epsilon}, \vec{\sigma}) 
\le \delta^R(\gamma, \vec{\epsilon}, \vec{\sigma})$. Thus, \textbf{R-TOL} can be viewed 
as a conservative approximation of \textbf{TOL}. Notably, the gap between the two 
objectives becomes small in three regimes: (i) when \( \gamma \) is small, 
(ii) when $\gamma$ is large, and (iii) when \( \vec{\epsilon} \) is large. Below, we provide intuition for these limits: % explanations for the limiting behaviors described in Proposition~\ref{prop_A1}:
\vspace{1mm}
\begin{enumerate}
\item When \( \gamma \) is close to 0, the required probability guarantee is weak, making the constraint in \textbf{TOL} easy to satisfy. Hence, the tolerance \( \delta \) can be made small (near zero) while still maintaining feasibility. Similarly, \textbf{R-TOL} can also achieve a near-zero objective by setting \( c_i \to 0\), which yields $\kappa(\epsilon_i, c_i) \to 1$, for all $i \in [M]$. 

\vspace{1mm}
\item When $\gamma$ is close to 1, the required probability guarantee becomes highly stringent and demands near-complete coverage. Since \( \tilde{\beta}^*(\vec{k}, \vec{S}, \vec{\sigma}) \) is a continuous random variable with support in \( (\beta^*(\vec{\sigma}), 1 - \alpha] \), we have \( \delta^*(\gamma, \vec{\epsilon}, \vec{\sigma}) \to 1 - \alpha - \beta^*(\vec{\sigma}) \). In this regime, setting $c_i \to 1$ for all $i \in [M]$ forces $\kappa(\epsilon_i, c_i) \to \infty$, so \textbf{R-TOL} yields the same limit. %\textbf{R-TOL} also has the same limit by setting \( c_i \) close to 1, which yields $\kappa(\epsilon_i, c_i) \to \infty$, for all \( i \in [M] \).

\vspace{1mm}
\item When \( \epsilon_i \to \infty \) for all \( i \in [M] \), the pilot standard deviation \( S_i \) converges to the true standard deviations \( \sigma_i \) almost surely. Thus, no correction is needed for \textbf{TOL} and  \( \delta^*(\gamma, \vec{\epsilon}, \vec{\sigma}) \to 0\). As for \textbf{R-TOL}, the chi-squared confidence interval for each \( \sigma_i^2 \) collapses to its true value, giving $\kappa(\epsilon_i, c_i) \to 1$ for any fixed $c_i$. Hence, any choice of \( \vec{c} \in [0,1)^M \) satisfying \( \prod_{i=1}^M c_i = \gamma \) is sufficient to meet the probabilistic guarantee.
\end{enumerate}

\vspace{1mm}
Similar to Proposition \ref{prop_A1}, the next two propositions characterize the relationship between \textbf{CONF} and \textbf{R-CONF}, and between \textbf{EXP} and \textbf{R-EXP}.

\vspace{1mm}
\begin{proposition}[\textbf{CONF vs. R-CONF}] \label{prop_A2}
Suppose that $\delta \in (0, 1-\alpha - \beta^*(\vec\sigma))$. The pair $(\vec k, \gamma) = (\vec k^{\textbf{R-CONF}}, \gamma^R(\delta, \vec{\epsilon}, \vec{\sigma}))$ is feasible for \textbf{CONF} and $\gamma^*(\delta, \vec{\epsilon}, \vec{\sigma}) 
\ge \gamma^R(\delta, \vec{\epsilon}, \vec{\sigma})$. Moreover, 
\begin{enumerate}
\item[\emph{(i)}] As $\delta \to 0$, %$c^{\textbf{R-CONF}}_i \to 0$ for all $i \in [M]$, 
$\gamma^*(\delta, \vec{\epsilon}, \vec{\sigma}) \to 0$ and $\gamma^R(\delta, \vec{\epsilon}, \vec{\sigma}) \to 0$;
\item[\emph{(ii)}] As $\delta \to 1-\alpha - \beta^*(\vec\sigma)$, %$c^{\textbf{R-CONF}}_i \to 1$ for all $i \in [M]$, 
$\gamma^*(\delta, \vec{\epsilon}, \vec{\sigma}) \to 1$ and $\gamma^R(\delta, \vec{\epsilon}, \vec{\sigma}) \to 1$;
\item[\emph{(iii)}] As $\epsilon_i \to \infty$ for all $i \in [M]$, $\gamma^*(\delta, \vec{\epsilon}, \vec{\sigma}) \to 1$ and $\gamma^R(\delta, \vec{\epsilon}, \vec{\sigma}) \to 1$.
\end{enumerate}
\end{proposition}

\begin{proposition}[\textbf{EXP vs. R-EXP}] \label{prop_A3}
 The solution $\vec k = \vec k^{\textbf{R-EXP}}$ is feasible for \textbf{EXP} and $g^*(\vec{\epsilon}, \vec{\sigma}) $ $ \le  g^R(\vec{\epsilon}, \vec{\sigma})$. In addition, as $\epsilon_i \to \infty$ for all $i \in [M]$, $g^*(\vec{\epsilon}, \vec{\sigma}) \to \beta^*(\vec\sigma)$ and $g^R(\vec{\epsilon}, \vec{\sigma}) \to \beta^*(\vec\sigma)$.
\end{proposition}

\vspace{1mm}
The insights from Propositions~\ref{prop_A2} and \ref{prop_A3} parallel those of 
Proposition~\ref{prop_A1}. For \textbf{CONF}, the surrogate \textbf{R-CONF} achieves a conservative lower bound on the achievable confidence level. When the tolerance $\delta \to 0$, the requirement becomes most stringent, making it difficult to guarantee high confidence; in this regime, both $\gamma^*(\delta, \vec{\epsilon}, \vec{\sigma})$ and $\gamma^R(\delta, \vec{\epsilon}, \vec{\sigma})$ converge to zero. As tolerance $\delta$ increases, the guarantee becomes easier to satisfy, and both formulations converge to full confidence. 
Similarly, for \textbf{EXP}, the surrogate \textbf{R-EXP} provides an upper bound on 
expected Type~2 error. In both cases, the gap between the surrogates and the original formulations vanishes as pilot sizes grow, ensuring that the robust approximations remain faithful to their exact counterparts in large-sample regimes. Together with Proposition~\ref{prop_A1}, these results show that our surrogate frameworks serve as 
reasonable proxies for the original formulations.

The next proposition demonstrates that the optimal \( \vec{c} \) under each reformulation also exhibits a structure similar to the exact solution in the two-experiment setting analyzed in Corollary~\ref{cor:optimal_2exp}.
\vspace{1mm}

\begin{proposition} \label{prop_A4}
Suppose that $\epsilon_i = \epsilon$ for all $i \in [M]$ (i.e., we have identical pilot sample sizes). For each $\pi \in \{\textbf{R-TOL}, \textbf{R-CONF},$ $ \textbf{R-EXP}\}$, we have:
\begin{enumerate}
\item[\emph{(i)}] If $\frac{\sigma_i}{\Delta_i} = \frac{\sigma_j}{\Delta_j}$, then $c^{\pi}_i = c^{\pi}_j$;
\item[\emph{(ii)}] If $\frac{\sigma_i}{\Delta_i} < \frac{\sigma_j}{\Delta_j}$, then $c^{\pi}_i \ge c^{\pi}_j$, with strict inequality when both are interior solutions;
\item[\emph{(iii)}] If $\frac{\sigma_i}{\Delta_i} > \frac{\sigma_j}{\Delta_j}$, then $c^{\pi}_i \le c^{\pi}_j$, with strict inequality when both are interior solutions.
\end{enumerate}
\end{proposition}

\subsection{Solving \textbf{R-TOL}, \textbf{R-CONF}, and \textbf{R-EXP}}
\label{subsec:solve_r-tol}
We now discuss how to solve the surrogate reformulations. 

Note that solving \textbf{R-TOL} is equivalent to the solving the following optimization problem:
%\vspace{1mm}
\begin{eqnarray*}
\min_{\vec{c}} && \sum_{i=1}^M \kappa(\epsilon_i, c_i) 
\left(\frac{\sigma_i}{\Delta_i}\right)^2 \\
\text{subject to} && \prod_{i=1}^M c_i \ge \gamma, \\
&& c_i \in  [0,1) \ \ \forall i \in [M]
\end{eqnarray*}
%\vspace{1mm}
\noindent
Let \( x_i = \log c_i \) and define \( g(x_i) := \kappa(\epsilon_i, e^{x_i}) \) for all \( i \in [M] \). By construction,  \( x_i \in (-\infty, 0) \) and \( g(x_i) \) is increasing in \( x_i \), with $\lim_{x_i \to -\infty} g(x_i) = 1$ and $\lim_{x_i \to 0} g(x_i) = \infty$. The problem becomes:
%\vspace{1mm}
\begin{eqnarray} \label{alternate R_TOL}
\min_{\vec{x}} && \sum_{i=1}^M g(x_i) 
\left(\frac{\sigma_i}{\Delta_i}\right)^2 \label{OBJ1} \\
\text{subject to} && \sum_{i=1}^M x_i \ge \log \gamma, \nonumber \\
&& x_i \in  (-\infty, 0) \ \ \forall i \in [M] \nonumber
\end{eqnarray}
%\vspace{1mm}
\noindent
By Lemma~\ref{lem_kappa}(ii), $\kappa(\epsilon_i,c_i)$ is convex and increasing in $c_i$.
Therefore $g(x_i)=\kappa(\epsilon_i,e^{x_i})$ is convex in $x_i$, and the objective in \eqref{alternate R_TOL} is convex and separable.
Since the constraint set is also convex, \eqref{alternate R_TOL} is a (deterministic) convex optimization problem and can be solved efficiently using standard solvers.

\textbf{R-CONF} shares the same structural ingredients as \textbf{R-TOL}. Applying the change of variables \(x_i=\log c_i\) and \(g(x_i)=\kappa(\epsilon_i,e^{x_i})\), and using the definition of \(d(\delta)\) in \eqref{eq:d(delta)}, it can be written as % equivalently as
\begin{eqnarray} \label{alternate R_CONF}
\max_{\vec{x}} && \sum_{i=1}^M x_i \\
\text{subject to} && 
\sum_{i=1}^M g(x_i)
\left(\frac{\sigma_i}{\Delta_i}\right)^2 \le d(\delta), \nonumber \\
&& x_i \in (-\infty, 0) \ \ \forall i \in [M] \nonumber
\end{eqnarray}

%\vspace{1mm}
\noindent
The above is a convex problem and can be solved efficiently using any off-the-shelf solver.

As for \textbf{R-EXP}, for any fixed value of $\prod_{i=1}^M c_i = \gamma$, 
the inner minimization over $\vec{c}$ is equivalent to solving 
\textbf{R-TOL}. Therefore, \textbf{R-EXP} can be solved by 
performing a one-dimensional line search over $\gamma \in (0,1)$: 
for each candidate $\gamma$, solve the corresponding 
\textbf{R-TOL} subproblem to obtain the minimal inner value, and 
then choose the $\gamma$ that minimizes the overall objective. 
Each subproblem is convex and can be efficiently solved using 
any off-the-shelf solver.

\subsection{The Proposed Approach}

While the reformulations \textbf{R-TOL}, \textbf{R-CONF}, and 
\textbf{R-EXP} provide a computationally tractable framework, their 
objectives and constraints are still expressed in terms of the true 
standard deviations $\vec{\sigma}$, which are unknown in practice. 
To translate these structural results into an operational procedure, 
we propose the \textbf{Surrogate-}$S$ method. 

The key idea is to use the pilot study as a plug-in estimator for the 
unknown variance parameters. Specifically, we replace each 
$\sigma_i$ in the surrogate formulations with its pilot-based estimate 
$S_i$, and solve the resulting optimization problem exactly as before. 
This substitution produces a fully data-dependent procedure that 
retains the tractability of the surrogate programs while eliminating 
the need for knowledge of the true $\vec{\sigma}$.

Formally, replacing $\sigma_i$ with $S_i$ yields the following empirical 
optimization problems:

\vspace{1mm}
\noindent \textbf{Empirical R-TOL (Surrogate-$S$):}
\begin{eqnarray} \label{empirical_R_TOL}
\min_{\vec{x}} && \sum_{i=1}^M g(x_i) \left(\frac{S_i}{\Delta_i}\right)^2 \\
\text{subject to} && \sum_{i=1}^M x_i \ge \log \gamma, \nonumber \\
&& x_i \in (-\infty, 0) \ \ \forall i \in [M] \nonumber
\end{eqnarray}

\vspace{1mm}
\noindent \textbf{Empirical R-CONF (Surrogate-$S$):}
\begin{eqnarray} \label{empirical_R_CONF}
\max_{\vec{x}} && \sum_{i=1}^M x_i \\
\text{subject to} && \sum_{i=1}^M g(x_i) \left(\frac{S_i}{\Delta_i}\right)^2 \le d(\delta), \nonumber \\
&& x_i \in (-\infty, 0) \ \ \forall i \in [M] \nonumber
\end{eqnarray}

\vspace{1mm}
\noindent \textbf{Empirical R-EXP (Surrogate-$S$):}
Similar to \textbf{R-EXP} but replace each $\sigma_i$ with $S_i$.

\vspace{1mm}
To operationalize the proposed method in practice, the platform proceeds as in Figure~\ref{fig:surrogate_s_flowchart}. 
It first collects the pilot variance estimates $S_i^2$ from samples of 
size $\epsilon_i$, together with the managerial effect-size thresholds 
$\Delta_i$. Given these inputs, the platform solves the selected empirical 
surrogate formulation (e.g., \eqref{empirical_R_TOL} for \textbf{R-TOL}) 
using a standard convex optimization solver to obtain the optimal 
decision variables $\vec{x}^*$. These variables are then mapped to confidence levels via 
$c_i^* = \exp(x_i^*)$, which determine the corresponding robust 
correction factors $k_i^* = \overline{\varphi}(\epsilon_i, c_i^*)$ 
derived from the $\chi^2$ upper bounds. Finally, substituting the 
resulting correction factors $\vec{k}^*$ and the pilot estimates 
$\vec{S}$ into the power-optimal allocation formula 
\eqref{eq:n*(k,S)} yields the allocation of the total experimentation 
budget $N$ across experiments.

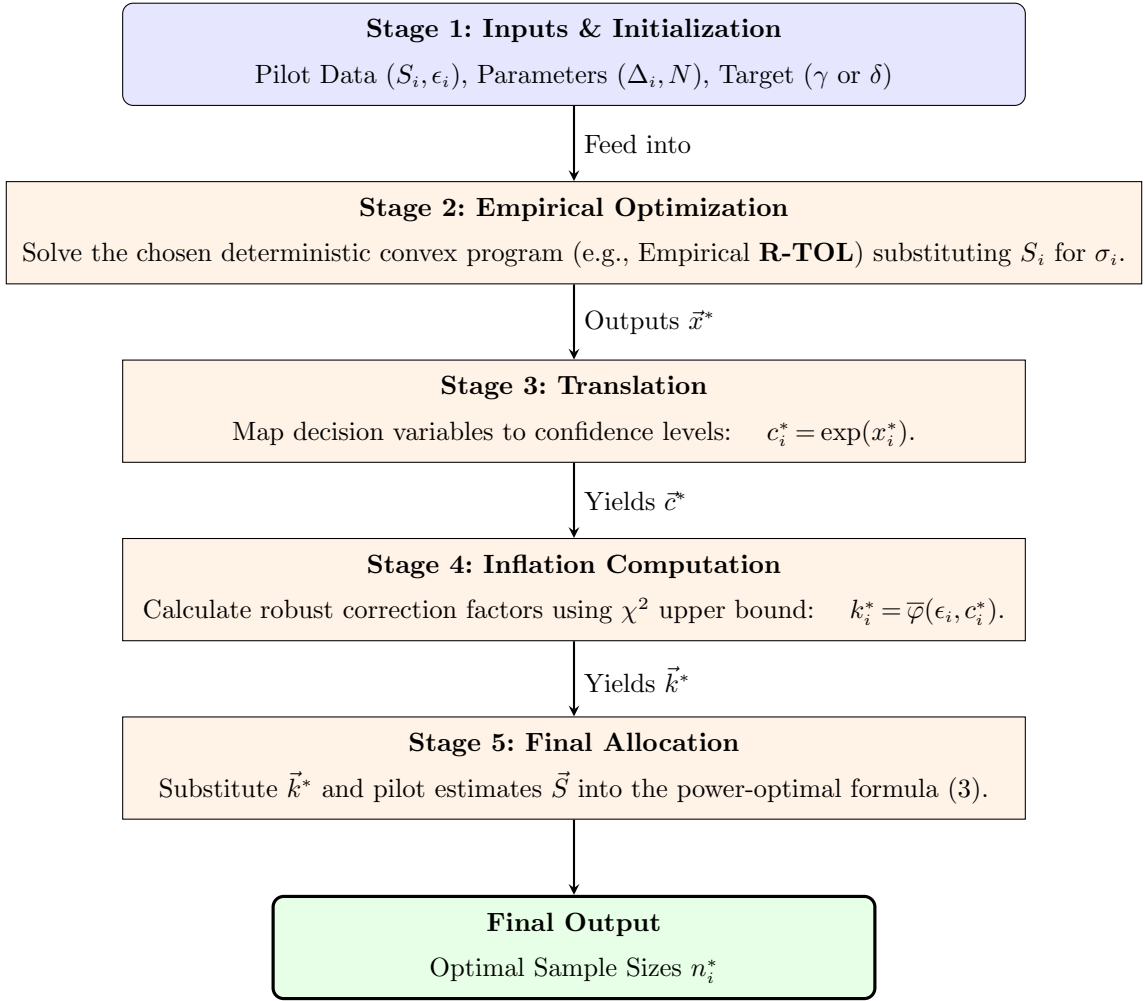
\begin{figure}[htbp]
\centering

\usetikzlibrary{positioning,calc}

\tikzset{
  input/.style = {rectangle, rounded corners, draw=black, fill=blue!10,
                  minimum width=12cm, inner sep=6pt, align=center},
  process/.style = {rectangle, draw=black, fill=orange!10,
                    minimum width=12cm, inner sep=6pt, align=center},
  output/.style = {rectangle, rounded corners, draw=black, fill=green!10,
                   very thick, minimum width=8cm, inner sep=6pt, align=center},
  arrow/.style = {thick, ->, >=stealth},
  font=\small
}

\begin{tikzpicture}[node distance=10mm]

\node (start) [input] {
  \textbf{Stage 1: Inputs \& Initialization}\\
  Pilot Data ($S_i,\epsilon_i$), Parameters ($\Delta_i,N$), Target ($\gamma$ or $\delta$)
};

\node (opt) [process, below=of start] {
  \textbf{Stage 2: Empirical Optimization}\\
  Solve the chosen deterministic convex program (e.g., Empirical \textbf{R-TOL})
  substituting $S_i$ for $\sigma_i$.
};

\node (trans) [process, below=of opt] {
  \textbf{Stage 3: Translation}\\
  Map decision variables to confidence levels: \quad $c_i^*=\exp(x_i^*)$.
};

\node (inf) [process, below=of trans] {
  \textbf{Stage 4: Inflation Computation}\\
  Calculate robust correction factors using $\chi^2$ upper bound: \quad
  $k_i^*=\overline{\varphi}(\epsilon_i,c_i^*)$.
};

\node (alloc) [process, below=of inf] {
  \textbf{Stage 5: Final Allocation}\\
  Substitute $\vec{k}^*$ and pilot estimates $\vec{S}$ into the power-optimal formula~(3).
};

\node (end) [output, below=of alloc] {
  \textbf{Final Output}\\
  Optimal Sample Sizes $n_i^*$
};

\draw [arrow] (start) -- (opt);
\draw [arrow] (opt) -- (trans);
\draw [arrow] (trans) -- (inf);
\draw [arrow] (inf) -- (alloc);
\draw [arrow] (alloc) -- (end);

\node[right] at ($(start.south)!0.5!(opt.north)$)  {Feed into};
\node[right] at ($(opt.south)!0.5!(trans.north)$)  {Outputs $\vec{x}^*$};
\node[right] at ($(trans.south)!0.5!(inf.north)$)  {Yields $\vec{c}^*$};
\node[right] at ($(inf.south)!0.5!(alloc.north)$)  {Yields $\vec{k}^*$};

\end{tikzpicture}

\caption{End-to-End Process Flow of the Surrogate-$S$ Method. The procedure transforms raw pilot data into final sample size allocations through a sequence of convex optimization and deterministic mappings.}
\label{fig:surrogate_s_flowchart}
\end{figure}

\vspace{2mm}
To illustrate the practical impact of our surrogate reformulation, we present three simulation plots corresponding to the objectives \textbf{R-TOL}, \textbf{R-CONF}, and \textbf{R-EXP}. In the discussion below, we compare the performance of three distinct allocation strategies: 
\begin{itemize}
\item The \textbf{Naive Plug-in (Blue line)}, which serves as the uncorrected baseline (no inflation factor, i.e., $k_i = 1$) where $S_i$ is substituted for $\sigma_i$; 
\item The \textbf{Oracle Surrogate-$\sigma$ (Orange line)}, representing the theoretical benchmark where robust factors are tuned using the true $\sigma_i$; and,
\item Our proposed \textbf{Surrogate-$S$ (Green line)}, which implements the robust formulation using only the pilot estimates $S_i$.
\end{itemize}

\begin{figure}[htbp]
    \centering
    \includegraphics[width=0.7\textwidth]{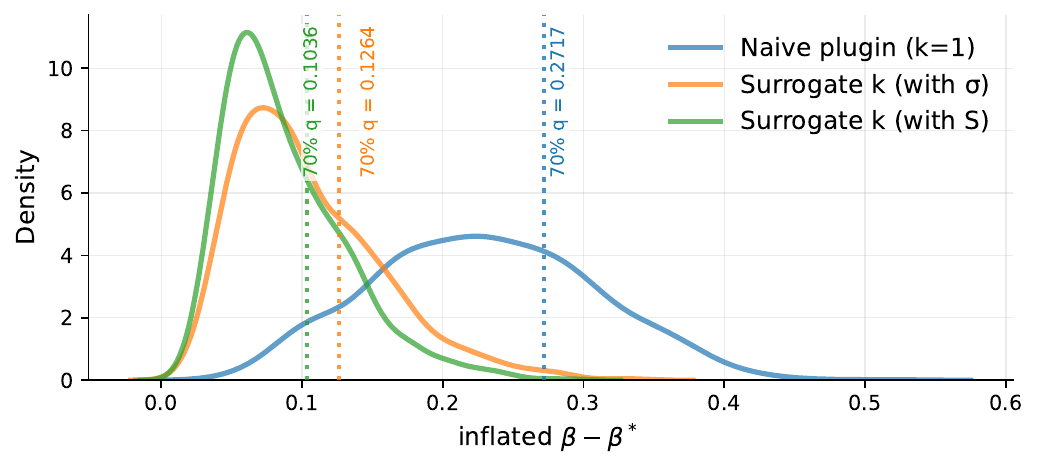}
    \caption{Distribution of Maximum Type~2 Error Relative to $\beta^*(\vec{\sigma})$ for the \textbf{R-TOL} objective ($\gamma = 0.7$). The vertical dotted lines mark the 70th percentile cutoffs. The curves compare the \textbf{Naive Plug-in (Blue)} with no correction factor, the \textbf{Oracle Surrogate-$\sigma$ (Orange)}, and our proposed \textbf{Surrogate-$S$ (Green)}.}
    \label{fig:rtol}
\end{figure}

\vspace{1mm}
Figure~\ref{fig:rtol} evaluates the \textbf{R-TOL} objective, where the platform requires a $70\%$ confidence level ($\gamma=0.7$) that the realized Type~2 error remains within a specific bound. The vertical dotted lines quantify the ``cost'' of this reliability: they mark the minimum tolerance $\delta$ needed to ensure the error stays below that threshold in $70\%$ of the trials. To achieve this $70\%$ confidence guarantee, the Naive Plug-in method (Blue) forces the user to accept a substantial excess error margin of $\approx 0.27$. In contrast, our \textbf{Surrogate-$S$} method (Green) satisfies the same $70\%$ confidence requirement with a strictly tighter tolerance of $\approx 0.10$. This demonstrates that for a fixed level of confidence, our robust formulation reduces the necessary error margin by over $60\%$ relative to the naive baseline.

\begin{figure}[htbp]
    \centering
    \includegraphics[width=0.7\textwidth]{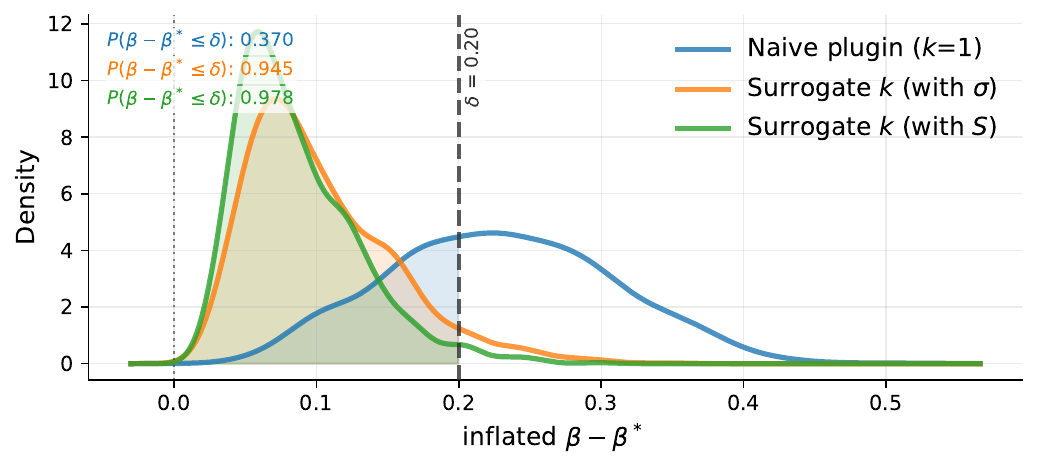}
    \caption{Distribution of Maximum Type~2 Error Relative to $\beta^*(\vec{\sigma})$ for the \textbf{R-CONF} objective ($\delta = 0.2$). The dashed vertical line marks the tolerance threshold. The curves compare the \textbf{Naive Plug-in (Blue)}, the \textbf{Oracle Surrogate-$\sigma$ (Orange)}, and our proposed \textbf{Surrogate-$S$ (Green)}.}
    \label{fig:rconf}
\end{figure}

Figure~\ref{fig:rconf} evaluates the \textbf{R-CONF} objective, where the platform sets a strict tolerance limit of $\delta=0.2$ on the excess error (marked by the vertical dashed line). The goal is to maximize the probability---or confidence level---that the realized Type~2 error stays within this pre-specified bound (i.e., to the left of the dashed line). The Naive method (Blue) fails to reliably meet this constraint, with significant probability mass leaking beyond the threshold; specifically, it achieves a realized confidence of only $37.0\%$, meaning it violates the error limit in most trials. Conversely, our \textbf{Surrogate-$S$} method (Green) successfully concentrates the distribution within the allowable region, achieving a realized confidence level of $97.8\%$. This demonstrates that the robust approach effectively guarantees that the error remains within the manager's tolerance, whereas the naive approach frequently violates it.
\begin{figure}[htbp]
    \centering
    \includegraphics[width=0.7\textwidth]{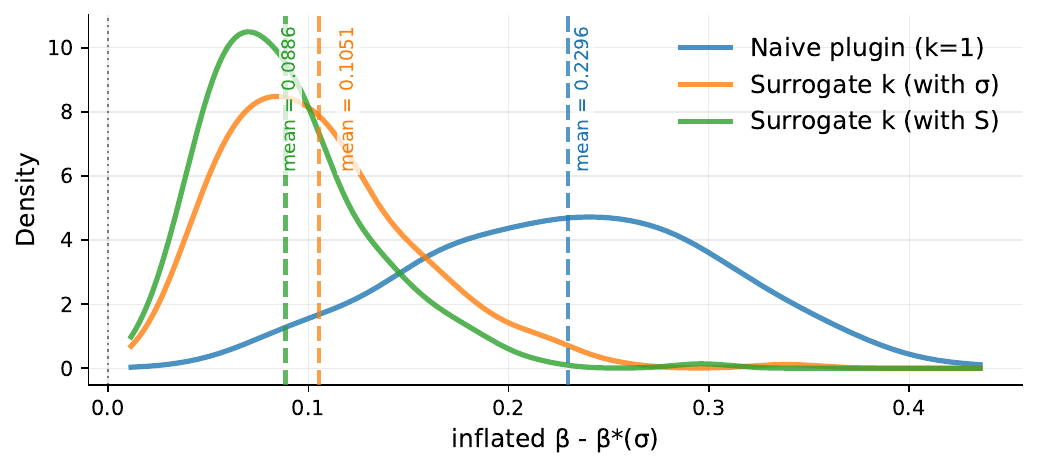}
    \caption{Distribution of Maximum Type~2 Error Relative to $\beta^*(\vec{\sigma})$ for the \textbf{R-EXP} objective. The vertical dashed lines indicate the mean excess error. The curves compare the \textbf{Naive Plug-in (Blue)}, the \textbf{Oracle Surrogate-$\sigma$ (Orange)}, and our proposed \textbf{Surrogate-$S$ (Green)}.}
    \label{fig:rexp}
\end{figure}

Finally, Figure~\ref{fig:rexp} evaluates the \textbf{R-EXP} objective, which adopts a risk-neutral perspective by minimizing the \emph{expected} worst-case Type~2 error. The vertical dashed lines mark the mean of each distribution, representing the average excess error a platform would incur over repeated applications. The Naive approach (Blue) results in a significantly higher average cost, with a mean excess error of approximately $0.23$. In contrast, our \textbf{Surrogate-$S$} method (Green) shifts the entire distribution toward zero, greatly reducing this average cost to roughly $0.09$. This reduction of over $60\%$ demonstrates that accounting for estimation uncertainty via robust inflation factors yields substantial performance gains on average, not just in extreme cases. Lastly, we note in Figures~\ref{fig:rtol}-~\ref{fig:rexp} our proposed Surrogate-S cost is comparable (and often better) to the Oracle Surrogate approach, showing the robust formulation loses little with the plug-in approach.

\section{Conclusion and Future Directions} \label{sec:closing}

In this paper, we studied the allocation of limited experimentation 
resources across a large portfolio of parallel tests in an 
experiment-rich regime. Departing from traditional allocation rules 
that prioritize estimation accuracy or average performance, we 
adopted a minimax perspective that directly controls the maximum 
Type~2 error across experiments. This objective ensures that no 
experiment is systematically underpowered and that detection 
reliability is balanced across the portfolio. We showed that MSE-optimal allocations, while natural for uniform 
estimation accuracy, can perform poorly from a discovery standpoint, 
particularly under tight budgets. To address this gap, we introduced 
correction factors that inflate pilot-based variance estimates, 
mitigating the risk of underestimating variability and the associated 
power loss. We formalized this idea through three optimization-based 
frameworks—\textbf{TOL}, \textbf{CONF}, and \textbf{EXP}—each reflecting a distinct risk 
criterion. Although these formulations lead to high-dimensional 
stochastic programs, we developed robust optimization–inspired 
surrogate reformulations that are tractable, scalable, and retain the 
structural logic of optimal variance inflation. Building on this 
analysis, we proposed \textbf{Surrogate-$S$}, a fully data-dependent 
and implementable procedure that achieves near-oracle performance in 
numerical experiments. Taken together, our results underscore the importance of explicitly 
controlling Type~2 error in large-scale experimentation systems. 
More broadly, they illustrate how principled optimization tools can 
be used to align statistical guarantees with managerial objectives in 
resource-constrained environments.

Several directions for future research merit further investigation. 
First, extending the framework to adaptive or sequential allocation 
policies could enhance efficiency when experimentation unfolds over 
time. Second, incorporating dependence across experiments—such as 
overlapping user populations or interference effects—would enrich 
the model and better reflect real-world platforms. Third, exploring 
connections with broader causal inference frameworks may expand the 
applicability of our methods to more complex experimental designs. 
These directions point toward a more comprehensive theory of 
large-scale experimental design under uncertainty and resource 
constraints.

%\newpage
\bibliographystyle{ormsv080} % outcomment this and next line in Case 1
\bibliography{jpi} 

%% Here starts the e-companion (EC)
%%%%%%%%%%%%%%%%%%%%%%%%%%%%%%%%%%%%%%%%%%%%%%%%%%%%%%%%%%
\ECSwitch

\ECDisclaimer
%%%%%%%%%%%%%%%%%%%%%%%%%%%%%%%%%%%%%%%%%%%%%%%%%%%%%%%%%%

%%% EC.1 matches Main Section 4 %%%
\section{Proofs and Supplemental Material for Section \ref{sec:known_sigma}} \label{ec:known_sigma}

\subsection{Proof of Proposition~\ref{prop:power_opt_optimal}}

By definition, the Type 2 error function \( \beta(\sigma_i, n_i) \) is continuous and strictly decreasing in \( n_i \). Under the minimax objective, the optimal solution must equalize the Type 2 errors across all experiments. This can be shown by a simple contradiction argument. Suppose \( \beta(\sigma_i, n_i^*) > \beta(\sigma_j, n_j^*) \) 
for some pair \( i \ne j \). Then, it is possible to shift a small amount of allocation 
from experiment \( e_i \) to experiment \( e_j \), which would strictly reduce the maximum Type 2 error between the two experiments. This contradicts the optimality of \( \vec{n}^* \). 
Hence, in the optimal solution, all Type 2 errors must be equal.

Assume the optimal allocation \( \vec{n}^* \) satisfies:
\[
\beta(\sigma_1, n_1^*) = \cdots = \beta(\sigma_M, n_M^*) = \beta^*, 
\quad \text{and} \quad \sum_{i=1}^M n_i^* = N.
\]
(If \( \sum_{i=1}^M n_i^* < N \), then there are unused resources that can be redistributed across the experiments to further reduce the maximum Type 2 error. This contradicts optimality, so any optimal solution must fully exhaust the budget, i.e., \( \sum_{i=1}^M n_i^* = N \).) By definition of \( \beta(\sigma_i, n_i) \), we have:
\[
\Phi\left( q_{1-\alpha} - \frac{ \Delta_i \sqrt{n_i^*} }{ \sigma_i } \right) = \beta^*,
\quad \forall i \in [M].
\]
Applying \( \Phi^{-1}(\cdot) \) to both sides and rearranging gives:
\[
n_i^* = \left( \frac{ \sigma_i }{ \Delta_i } \right)^2 
\left( q_{1-\alpha} - \Phi^{-1}(\beta^*) \right)^2.
\]
Summing over all experiments:
\[
\sum_{i=1}^M n_i^* = 
\left( q_{1-\alpha} - \Phi^{-1}(\beta^*) \right)^2 
\cdot \left[\sum_{i=1}^M \left( \frac{ \sigma_i }{ \Delta_i } \right)^2 \right] = N.
\]
Now, let $C := \sum_{i=1}^M \left( \frac{ \sigma_i }{ \Delta_i } \right)^2$. Then, we have $
C \left( q_{1-\alpha} - \Phi^{-1}(\beta^*) \right)^2 = N$, or equivalently
\[
\beta^* = \Phi\left( q_{1-\alpha} - \sqrt{ \frac{N}{C} } \right).
\]
We conclude that the optimal Type 2 error is:
\[
\beta^*(\vec{\sigma}) = 
\Phi\left( q_{1-\alpha} - \sqrt{ 
\frac{N}{\sum_{j=1}^M \left( \frac{ \sigma_j }{ \Delta_j } \right)^2 } } \right).
\]
Substituting back into the expression for \( n_i^* \), we obtain:
\[
n_i^* = \left( \frac{ \sigma_i }{ \Delta_i } \right)^2 
\cdot \frac{N}{C} 
= N \cdot \frac{ 
\left( \frac{ \sigma_i }{ \Delta_i } \right)^2 
}{ 
\sum_{j=1}^M \left( \frac{ \sigma_j }{ \Delta_j } \right)^2 
}.
\]
This completes the proof. \hfill $\blacksquare$

%%% EC.2 matches Main Section 5 %%%
\section{Proofs and Supplemental Material for Section \ref{sec:unknown sigma 2exp}} \label{ec:unknown sigma 2exp}

\subsection{Proof of Lemma~\ref{lem:H(r,d)}}

\underline{Proof of Property (i)}

Property~(i) of the lemma follows directly from the definition of \( H(r, d) \). Since \( F_{F_{\nu, \nu}} \) is the CDF of the F-distribution, it is monotonically increasing. As \( d \) increases, the upper limit increases while the lower limit decreases. Hence, the interval expands, and \( H(r,d) \) increases with \( d \).

\vspace{2mm}
\noindent
\underline{Proof of Property (ii)}

For property (ii), $\lim_{r \to 0^+} H(r,d) = 0$ and $\lim_{r \to \infty} H(r,d) = 0$ follow directly from the definition of $H(r,d)$. Now, we consider the maximizer $r(d)$.  Taking the partial derivative of \( H(r,d) \) with respect to \( r \) yields
\begin{align*}
\frac{\partial H}{\partial r}
&= \frac{a_2}{r^2 (d - a_1)} 
   f_F\left( \frac{a_2}{r (d - a_1)} \right)
 - \frac{d - a_2}{r^2 a_1} 
   f_F\left( \frac{d - a_2}{r a_1} \right),
\end{align*}
where $f_F$ is the pdf of $F_{F_{v,v}}$. Setting the derivative to zero yields
\begin{align*}
\frac{a_2}{r^2 (d - a_1)} 
f_F\left( \frac{a_2}{r (d - a_1)} \right)
&= \frac{d - a_2}{r^2 a_1} 
f_F\left( \frac{d - a_2}{r a_1} \right).
\end{align*}
Multiplying both sides by \( r \), we get:
\begin{align*}
\frac{a_2}{r(d - a_1)} 
f_F\left( \frac{a_2}{r (d - a_1)} \right)
&= \frac{d - a_2}{r a_1} 
f_F\left( \frac{d - a_2}{r a_1} \right).
\end{align*}
Now, let $x_1 := \frac{a_2}{r (d - a_1)}$ and $x_2 := \frac{d - a_2}{r a_1}$. Then, the above equation becomes:
\begin{align*}
g(x_1) \ := \ x_1 f_F(x_1)
\ = \ x_2 f_F(x_2) \ := \ g(x_2).
\end{align*}

\begin{claim} \label{claim_Appendix1}
The function $g(x)$ has the following properties:
\begin{itemize}
    \item [\emph{(i)}] On the set $x \in (0, \infty)$, $g(x)$ on is unimodal with a unique global maximum at $x=1$.
    \item[\emph{(ii)}] On the set $x \in (0, \infty)$, $g(x'_1) = g(x'_2)$ if and only if $x'_1 = x'_2$, or $x'_1 \cdot x'_2$ = 1.
\end{itemize}
\end{claim}

\textsc{Proof.} We first show property (i) for $g(x)$. Let \( a = \frac{\nu}{2} > 0 \). Then, 
\[
f_F(x) = c \cdot x^{a - 1} (1 + x)^{-2a}, \quad x > 0,
\]
where \( c = \frac{\Gamma(2a)}{\Gamma(a)^2} a^a \) is the normalizing constant. Thus,
\[
g(x) = x f_F(x) = c \cdot x^a (1 + x)^{-2a}.
\]
To determine unimodality, we find the critical points by computing the derivative. Consider the natural logarithm of $g(x)$:
\[
\ln g(x) = \ln c + a \ln x - 2a \ln (1 + x).
\]
Differentiating with respect to \( x \):
\[
\frac{g'(x)}{g(x)} = \frac{d}{dx} \ln g(x) = \frac{a}{x} - \frac{2a}{1 + x}.
\]
Setting the derivative to zero to find critical points yileds:
\[
\frac{a}{x} - \frac{2a}{1 + x} = 0 \implies \frac{1}{x} = \frac{2}{1 + x} \implies 1 + x = 2x \implies x = 1.
\]
Since \( a > 0 \) and \( x > 0 \), \( x = 1 \) is the only critical point. 
Given \( g'(x) > 0 \) for \( 0 < x < 1 \), \( g'(x) = 0 \) at \( x = 1 \), and \( g'(x) < 0 \) for \( x > 1 \), we conclude that $g(x)$ is unimodal.  

We next show property (ii), i.e., that $g(x'_1) = g(x'_2)$ if and only if either $x'_1 = x'_2$ or $x'_1 x'_2 = 1$. Using the expression \( g(x) = c \cdot x^a (1 + x)^{-2a} \), evaluate \( g \) at \( \frac{1}{x} \):
\[
g\left( \frac{1}{x} \right) = c \cdot \left( \frac{1}{x} \right)^a \left( 1 + \frac{1}{x} \right)^{-2a} = c \cdot x^{-a} \left( \frac{x + 1}{x} \right)^{-2a} = c \cdot x^{-a} \cdot \frac{x^{2a}}{(x + 1)^{2a}} = c \cdot \frac{x^a}{(x + 1)^{2a}} = g(x).
\]
Thus, \( g(x) = g\left( \frac{1}{x} \right) \) for all \( x > 0 \). By property (i), \( g(x) \) is unimodal with a maximum at \( x = 1 \), increasing on \( (0, 1) \), and decreasing on \( (1, \infty) \). We conclude that \( g(x'_1) = g(x'_2) \) if and only if \( x'_1 = x'_2 \) or \( x'_2 = \frac{1}{x'_1} \). \hfill $\blacksquare$

Recall that $r = r(d)$ must satisfy $g(x_1) = g(x_2)$. By Claim \ref{claim_Appendix1}, this happens if and only if $x_1 = x_2$ or $x_1 x_2 = 1$. Since $x_1 \not= x_2$, we must have $x_1 x_2 = 1$. This yields
\begin{eqnarray*}
\frac{a_2}{r(d - a_1)} \cdot \frac{d - a_2}{r a_1} = 1 \implies r = \sqrt{\frac{a_2}{d-a_1} \cdot \frac{d-a_2}{a_1}}.
\end{eqnarray*}
This completes the proof. \hfill $\blacksquare$

\subsection{Proof of Proposition~\ref{prop:optimal_k_M2}}

\underline{Proof for \textbf{TOL}}

Since \( d = d(\delta) \) is increasing in \(\delta\), minimizing \(\delta\) is equivalent to minimizing \(d\). So, we will focus on minimizing $d$ instead of $\delta$. The optimal solution for \textbf{TOL} satisfies $\delta^* \in (0, 1 - \alpha - \beta^*(\vec\sigma))$ almost surely. Thus, $d^* = d(\delta^*) \in \big( \sum_{j \in [2]} (\frac{\sigma_j}{\Delta_j})^2, \infty \big)$ and we can apply Lemma \ref{lem:H(r,d)}.

By definition, $d^*$ is the smallest $d$ such that there exists $r$ satisfying $H(r, d^*) \ge \gamma$. We claim that the optimal inflation ratio $r^*$ is given by $r^* = r(d^*)$. We prove by contradiction. Suppose $H(\hat r, d^*) = \gamma$ for some $\hat r \not= r^*(d^*)$. By the unimodality of $H(r,d)$ (Lemma \ref{lem_kappa} part (ii)), $H(r(d^*), d^*) > H(\hat r, d^*)$, which implies $H(r^*(d^*), d^*) > \gamma$. Moreover, by the property of $H(r,d)$ (Lemma \ref{lem:H(r,d)} part (i)), there exists $d' < d^*$ such that $H(r(d^*), d^*) > H(r(d^*), d') \ge \gamma$. This contradicts our definition that $d^*$ is the smallest $d$. We conclude that we must have $r^* = r(d^*)$. 

The formula for $r^* = r(d^*)$ follows directly from Lemma \ref{lem:H(r,d)}. As for deriving the value of $d^*$, we utilize the following properties:
\begin{itemize}
    \item [(i)] $F_{F_{\nu, \nu}}(\frac{m_2}{r^*}) - F_{F_{\nu, \nu}}(\frac{m_1}{r^*}) = F_{F_{\nu, \nu}}(\sqrt{\frac{m_2}{m_1}}) - F_{F_{\nu, \nu}}(\sqrt{\frac{m_1}{m_2}}) = \gamma$ 
    \item[(ii)] $F_{F_{\nu, \nu}}(\sqrt{\frac{m_1}{m_2}}) + F_{F_{\nu, \nu}}(\sqrt{\frac{m_2}{m_1}}) = 1$
\end{itemize}

Property (i) follows because, at optimal solution, we must have $H(r^*, d^*) = \gamma$ (if $H(r^*, d^*) > \gamma$, we can always reduce $d^*$ and get a smaller feasible solution, contradicting the optimality of $d^*$). Property (ii) comes from the property of the $F$-distribution for symmetric degree of freedom $(\nu, \nu)$. Specifically, it is known that if $X \sim F_{\nu, \nu}$, then $1/X \sim F_{\nu, \nu}$. This has the following implication:
\begin{align*}
    1 - F_{F_{\nu,\nu}}(x) = P(X > x) = P\left(\frac{1}{X} < \frac{1}{x}\right) = F_{F_{\nu,\nu}}\left(\frac{1}{x}\right) .
\end{align*}
Thus, we have $F_{F_{\nu,\nu}}(x) + F_{F_{\nu,\nu}}(1/x) = 1$. 

Now, solving the following system of equations:
\begin{eqnarray*}
F_{F_{\nu, \nu}}\left(\sqrt{\frac{m_2}{m_1}}\right) - F_{F_{\nu, \nu}}\left(\sqrt{\frac{m_1}{m_2}}\right) & = & \gamma \\
F_{F_{\nu, \nu}}\left(\sqrt{\frac{m_1}{m_2}}\right) + F_{F_{\nu, \nu}}\left(\sqrt{\frac{m_2}{m_1}}\right) & = & 1
\end{eqnarray*}
yields
\begin{eqnarray*}
F_{F_{\nu, \nu}}\left(\sqrt{\frac{m_2}{m_1}}\right) \ = \ \frac{1 + \gamma}{2}.
\end{eqnarray*}
Or, equivalently,
\begin{eqnarray*}
\frac{(d^*-a_1)(d^*-a_2)}{a_1 a_2} \ = \ \left(F^{-1}_{\nu, \nu}\left(\frac{1 + \gamma}{2}\right)\right)^2
\end{eqnarray*}
Solving for $d^*$ yields
\[
d^* = \frac{a_1 + a_2 + \sqrt{(a_1 - a_2)^2 + 4 a_1 a_2 \cdot \left(F^{-1}_{\nu, \nu}(\frac{1+\gamma}{2})\right)^2}}{2}.
\]

\vspace{2mm}
\noindent
\underline{Proof of \textbf{CONF}}

For \textbf{CONF} with tolerance \(\beta^*(\vec\sigma) + \delta\), we want to maximize the confidence:
\[
\gamma = \mathbb{P}\left( \tilde{\beta}^*(\vec{k}, \vec{S}, \vec{\sigma}) \leq \beta^*(\vec\sigma) + \delta \right).
\]
Compared with \textbf{TOL}, the $\delta$ here is fixed constant. We are essentially solving $\max_r H(r,d(\delta))$, and its optimal solution is given by Lemma \ref{lem:H(r,d)}. \hfill \(\blacksquare\)

\subsection{Proof of Corollary~\ref{cor:optimal_2exp}}

\underline{Proof for \textbf{TOL}}

Recall the expression of $d^* = d(\delta^*)$ for $\textbf{TOL}$ from Proposition \ref{prop:optimal_k_M2}, i.e.,
\[
d^* = \frac{a_1 + a_2 + \sqrt{(a_1 - a_2)^2 + 4 a_1 a_2 \cdot \left(F^{-1}_{\nu, \nu}(\frac{1+\gamma}{2})\right)^2}}{2}.
\]
We first argue that $d^* > a_1 + a_2$. To do this, it is sufficient to show that 
\[
F^{-1}_{\nu, \nu}\left(\frac{1+\gamma}{2} \right) > 1,\quad \forall \gamma \in (0, 1)
\]
But this is true because the median of F-distribution with symmetric degree of freedom $(\nu, \nu)$ is always 1. To see this, note that, by a similar argument as in the proof of Proposition \ref{prop:optimal_k_M2}, if $X\sim F_{\nu,\nu}$, then $1/X \sim F_{\nu, \nu}$. By definition, the median of $F_{\nu, \nu}$ is the number $m$ such that $\mathbb{P}(X \le m) = 0.5$. So,
$$
0.5 = \mathbb{P}(X \le m) = \mathbb{P}\left(\frac{1}{X} \le m\right) = \mathbb{P}\left(X \ge \frac{1}{m}\right),
$$
which implies
$$
\mathbb{P}\left(X \le \frac{1}{m}\right) = 1 - \mathbb{P}\left(X \ge \frac{1}{m}\right)=  0.5.
$$
Thus,
$$
\mathbb{P}(X \le m) = \mathbb{P}\left(X \le \frac{1}{m}\right).
$$
Since $X$ is a continuous random variable, we must have $m = \frac{1}{m}$ and, therefore, $m = 1$.

We can now express $d^* = a_1 + a_2 + a^*$ for some $a^* > 0$. Substituting this into the expression for $r^*$ yields 
\[
r^* = \sqrt{m_1 \cdot m_2} = 
\sqrt{\frac{1 + \frac{a^*}{a_1}}{1 + \frac{a^*}{a_2}}} = 
\sqrt{1 + \frac{\frac{1}{a_1} - \frac{1}{a_2}}{\frac{1}{a^*} + \frac{1}{a_2}}}.
\]
The results for \textbf{TOL} immediately follows: (i) if $a_1 = a_2$, then $r^* = 1$; (ii) if $a_1 < a_2$, then $r^* > 1$; (iii) if $a_1 > a_2$, then $r^* < 1$.

\underline{Proof for \textbf{CONF}}

The proof for \textbf{CONF} is similar. Recall that \( d(\delta) \) is increasing for \( \delta \in (0, 1 - \alpha - \beta^*(\vec\sigma)) \). Moreover, plugging \(\delta=0\) into \eqref{eq:d(delta)} gives
\[
d(0)=\frac{N}{\left(q_{1-\alpha}-\Phi^{-1}(\beta^*(\vec\sigma))\right)^2}.
\]
Using Proposition~\ref{prop:power_opt_optimal} with \(M=2\), we have
\(\beta^*(\vec\sigma)=\Phi\!\left(q_{1-\alpha}-\sqrt{\tfrac{N}{a_1+a_2}}\right)\),
so \(q_{1-\alpha}-\Phi^{-1}(\beta^*(\vec\sigma))=\sqrt{\tfrac{N}{a_1+a_2}}\), and therefore \(d(0)=a_1+a_2\). Thus, we must have
\[
d(\delta) > a_1 + a_2, \quad \forall \delta \in (0, 1 - \alpha - \beta^*(\vec\sigma)).
\]
In other words, for a given $ \delta \in (0, 1 - \alpha - \beta^*(\vec\sigma))$, we can write $d(\delta) = a_1 + a_2 + a(\delta)$ for some $a(\delta) > 0$. The rest of the proof is similar to that of \textbf{TOL}. 

\subsection{Proof of Corollary~\ref{cor:r_with_delta}}

From our analysis in the proof of Corollary \ref{cor:optimal_2exp}, for \textbf{TOL}, we have
\[
r^* = \sqrt{m_1 \cdot m_2} = 
\sqrt{\frac{1 + \frac{a^*(\gamma)}{a_1}}{1 + \frac{a^*(\gamma)}{a_2}}} = 
\sqrt{1 + \frac{\frac{1}{a_1} - \frac{1}{a_2}}{\frac{1}{a^*(\gamma)} + \frac{1}{a_2}}}.
\]
where we use $a^*(\gamma)$ to denote its dependency on $\gamma$. From the definition of \( d^* \) in Proposition~\ref{prop:optimal_k_M2}, we know that \( d^* \) increases with \( \gamma \). Since \( a^*(\gamma) \) is increasing in \( d^* \), it follows that \( a^*(\gamma) \) also increases with \( \gamma \). As a result, \( r^* \) increases with \( \gamma \).

As for \textbf{CONF}, again, from our analysis in Proposition \ref{prop:optimal_k_M2},
\[
r^* = \sqrt{m_1 \cdot m_2} = 
\sqrt{\frac{1 + \frac{a(\delta)}{a_1}}{1 + \frac{a(\delta)}{a_2}}} = 
\sqrt{1 + \frac{\frac{1}{a_1} - \frac{1}{a_2}}{\frac{1}{a(\delta)} + \frac{1}{a_2}}}.
\]
By the definition of $d(\delta)$, it is increasing for \( \delta \in (0, 1 - \alpha - \beta^*(\vec\sigma)) \). Therefore, as $\delta$ increases, $a(\delta)$ also increases. As a result, $r^*$ increases with $\delta$. \hfill\(\blacksquare\)

\subsection{Proof of Proposition~\ref{prop:optimal_exp}}

Fix two experiments with difficulty indices
\[
a_i := \left(\frac{\sigma_i}{\Delta_i}\right)^2, \qquad i=1,2.
\]
Let $r := k_1/k_2$ denote the inflation ratio, and set $s := \log r \in \mathbb{R}$.
We will analyze the objective as a function of $s$.

\textbf{Step 1. Distribution of $W$.}
Let $Y_1,Y_2 \overset{\text{i.i.d.}}{\sim} \chi^2_\nu$ with $\nu := \epsilon-1$,
and define
\[
W := \log\!\Big(\frac{Y_1}{Y_2}\Big).
\]
\emph{Symmetry.} Since $Y_1$ and $Y_2$ are i.i.d., the ratio $Y_1/Y_2$
has the same distribution as $Y_2/Y_1$. Hence
\[
W = \log\!\Big(\tfrac{Y_1}{Y_2}\Big)
\;\overset{d}{=}\;
\log\!\Big(\tfrac{Y_2}{Y_1}\Big) = -W.
\]
Thus the distribution of $W$ is symmetric about $0$.

\emph{Unimodality.} The density of $W$ (up to normalization) is
\[
f_W(w) \;\propto\; e^{(\nu/2)w}(1+e^w)^{-\nu}, \qquad w \in \mathbb{R}.
\]
Its log-derivative is
\[
\frac{d}{dw}\log f_W(w)
= \frac{\nu}{2} - \nu \frac{e^w}{1+e^w}.
\]
For $w<0$ the derivative is positive, for $w=0$ it equals zero,
and for $w>0$ it is negative.
Hence $g_W$ is strictly increasing on $(-\infty,0)$,
strictly decreasing on $(0,\infty)$, and attains its unique maximum at $w=0$.
Therefore, $W$ has a symmetric and strictly unimodal density with mode $0$.

\textbf{Step 2. Expression for the maximum error.}
For a given $s$, the quantities $U_1(r)$ and $U_2(r)$ can be expressed as
\[
U_1(r) = a_1 + a_2 e^{-(W+s)}, \qquad
U_2(r) = a_2 + a_1 e^{W+s}.
\]
Define the deterministic function
\[
M(y) := \max\{a_1+a_2 e^{-y}, \; a_2+a_1 e^y\}, \qquad y \in \mathbb{R}.
\]
Then, $\max(U_1(r),U_2(r)) = M(W+s)$.

The expected Type~2 error under the \textbf{EXP} formulation is
\begin{eqnarray*}
J(s) :=
\mathbb{E}\left[
\Phi \left(q_{1-\alpha} - \sqrt{N}\cdot \sqrt{\frac{1}{M(W+s)}}\right)
\right].
\end{eqnarray*}

\textbf{Step 3. Structure of $M(y)$.}
Let
$f_1(y) := a_1 + a_2 e^{-y}$ and $f_2(y) := a_2 + a_1 e^y$. We first analyze each component.  
\begin{itemize}
\item $f_1(y)$ is strictly decreasing in $y$ since 
$f_1'(y) = -a_2 e^{-y} < 0$ for all $y \in \mathbb{R}$.
\item $f_2(y)$ is strictly increasing in $y$ since 
$f_2'(y) = a_1 e^y > 0$ for all $y \in \mathbb{R}$.
\end{itemize}

The two curves intersect at a unique point $y^* \in \mathbb{R}$, which is 
the solution to
\[
f_1(y^*) = f_2(y^*) 
\;\;\Longleftrightarrow\;\;
a_1 + a_2 e^{-y^*} = a_2 + a_1 e^{y^*}.
\]
Rearranging, we have
\[
a_1 - a_2 = a_1 e^{y^*} - a_2 e^{-y^*}.
\]
The right-hand side is strictly increasing in $y$ (its derivative with respect to $y$ is always positive), hence there is a unique solution $y^*$.

Therefore,
\[
M(y) =
\begin{cases}
f_1(y), & y \le y^*, \\[1ex]
f_2(y), & y \ge y^*,
\end{cases}
\]
with $M(y^*) = f_1(y^*) = f_2(y^*)$. Since $f_1$ is strictly decreasing and 
$f_2$ is strictly increasing, $M(y)$ is strictly decreasing on 
$(-\infty,y^*]$ and strictly increasing on $[y^*,\infty)$. 
Thus $M(y)$ has a unique global minimum at $y^*$.

We now compute $y^*$. We start with 
$
a_1 - a_2 = a_1 e^{y^*} - a_2 e^{-y^*}.
$
Multiplying both sides by $e^{y^*}$ and letting $z=e^{y^*}>0$ yields
$
a_1 z^2 - (a_1 - a_2) z - a_2 = 0.
$
Solving this quadratic equation gives
\[
z = \frac{(a_1 - a_2) + \sqrt{(a_1 - a_2)^2 + 4 a_1 a_2}}{2 a_1}
     = \frac{a_2}{a_1}.
\]
Hence
\[
y^* = \log\!\left(\frac{a_2}{a_1}\right).
\]

\textbf{Step 4. Properties of $H(y)$.}
Define
\[
H(y) := \Phi\!\left(q_{1-\alpha}-\frac{\sqrt{N}}{\sqrt{M(y)}}\right).
\]
Note that $H(y)$ is strictly increasing in $M(y)$. 
Because $M(y)$ is strictly decreasing on $(-\infty,y^*]$ and strictly 
increasing on $[y^*,\infty)$, $H(y)$ is also strictly decreasing on 
$(-\infty,y^*]$, strictly increasing on $[y^*,\infty)$, and attains its 
unique minimum at $y^*$. 

\textbf{Step 5. Sublevel structure and auxiliary lemma.}

\emph{(5.1) Sublevel sets of $M$ and $H$.}
As established earlier, $M(y)=\max\{f_1(y),f_2(y)\}$ with 
$f_1(y)=a_1+a_2e^{-y}$ strictly decreasing and $f_2(y)=a_2+a_1e^{y}$ strictly increasing.
They intersect at a unique point $y^*$, which is the unique global minimizer of $M$.
Since $M(y)$ is strictly decreasing on $(-\infty,y^*]$ and strictly increasing on $[y^*,\infty)$,
for any $m>M(y^*)$, the set $\{y:M(y)\le m\}$ is a closed interval containing $y^*$.
Specifically, for any $m\ge M(y^*)$,
\[
\{y:M(y)\le m\}
=\{y:f_1(y)\le m\}\cap\{y:f_2(y)\le m\}
=[\ell(m),u(m)],
\]
where
\[
\ell(m)=-\log\!\Big(\frac{m-a_1}{a_2}\Big),\qquad
u(m)=\log\!\Big(\frac{m-a_2}{a_1}\Big).
\]
(Here, $\ell(m)$ solves $f_1(y)=m$ and $u(m)$ solves $f_2(y)=m$.)

The center and half-length of the interval are given by
\[
C(m):=\frac{\ell(m)+u(m)}{2}
=\frac12\log\!\left(\frac{a_2(m-a_2)}{a_1(m-a_1)}\right),
\qquad
R(m):=\frac{u(m)-\ell(m)}{2}>0.
\]
Since $M(y^*) = f_1(y^*) = a_1 + a_2 e^{-y^*}$ and $M(y^*) = f_2(y^*) = a_2 + a_1 e^{y^*}$, we have
\[
C(M(y^*)) = \frac12\log\left(\frac{a_2(M(y^*)-a_2)}{a_1(M(y^*)-a_1)}\right) = \frac12\log\left(\frac{a_2 a_1 e^{y^*}}{a_1 a_2 e^{-y^*}}\right) = y^*.
\]
As $m\to\infty$,
\[
C(m)\to \frac12\log\Big(\frac{a_2}{a_1}\Big) = \frac12 y^*.
\]
Moreover,
\[
C'(m)=\tfrac12\!\left(\frac{1}{m-a_2}-\frac{1}{m-a_1}\right),
\]
so $C(m)$ is strictly increasing in $m$ when $a_2>a_1$, strictly decreasing when $a_1>a_2$,
and constant when $a_1=a_2$.

Because $H(y)$ is strictly increasing in $M(y)$, its sublevel sets inherit those of $M$:
for every $\tau\in[H(y^*),1-\alpha)$ there exists a unique $m(\tau)\in[M(y^*),\infty)$ such that
\[
\{y:H(y)\le\tau\}=\{y:M(y)\le m(\tau)\}=[L(\tau),U(\tau)],
\]
where $L(\tau):=\ell(m(\tau))$ and $U(\tau):=u(m(\tau))$.
Define their center and half-length:
\[
c(\tau):=\frac{L(\tau)+U(\tau)}{2}=C(m(\tau)),\qquad
r(\tau):=\frac{U(\tau)-L(\tau)}{2}=R(m(\tau))>0.
\]

When $a_1=a_2$, the function $M(y)$ is symmetric about $y^*=0$, so all sublevel sets are
centered at $y^* = 0$. Indeed, by definition, $C(m)$ is always 0. So, $c(\tau)\equiv y^*=0$ for all $\tau$. 

When $a_1\neq a_2$, the centers of the sublevel intervals move away from $y^*$.
To see this, recall the formulas of $C(m)$ and $C'(m)$. If $a_2>a_1$, then $C(m)$ is strictly increasing for all $m>M(y^*)$.
Since $m(\tau)>M(y^*)$ for every $\tau\in(H(y^*),1-\alpha)$ and
$c(\tau)=C(m(\tau))$, we obtain
\[
c(\tau)>C(M(y^*))=y^* \qquad\text{whenever } a_2>a_1.
\]
Analogously, if $a_1>a_2$, then $C(m)$ is strictly decreasing for all $m>M(y^*)$ and hence $c(\tau)<y^*$ for all such $\tau$.
Consequently,
\[
\operatorname{sign}\big(c(\tau)-y^*\big)=\operatorname{sign}(a_2-a_1)
\qquad\forall\,\tau\in(H(y^*),1-\alpha).
\]

\emph{(5.2) Auxiliary lemma for symmetric unimodal densities.}
For any fixed $c\in\mathbb{R}$ and $r>0$, define
\[
F_{c,r}(s):=\mathbb{P}\big(W+s\in[c-r,c+r]\big)
= \int_{c-r}^{c+r} f_W(x-s)\,dx
= F_W(c+r-s)-F_W(c-r-s),
\]
where $F_W$ is the cdf of $W$.
Since $f_W$ is continuous, symmetric ($f_W(x)=f_W(-x)$), and strictly unimodal
with mode $0$, $F_{c,r}$ is differentiable in $s$ and
\[
\frac{d}{ds}F_{c,r}(s)=f_W(c-r-s)-f_W(c+r-s).
\]
We analyze the sign of $\dfrac{d}{ds}F_{c,r}(s)$ with $a:=c-s$. Recall that $r > 0$.
\begin{itemize}
\item \emph{Case 1: $s<c$ ($a>0$).}
We split according to the sign of $a-r$.
\begin{itemize}
\item {If $a\ge r$,} then $a-r\ge0$ and $a+r>0$, so both arguments are nonnegative.
Since $a-r<a+r$ and $f_W$ is strictly decreasing on $(0,\infty)$, we have
$f_W(a-r)>f_W(a+r)$; if $a=r$ this remains true because $f_W(0)$ is the strict maximum.
Hence
$
\frac{d}{ds}F_{c,r}(s)=f_W(a-r)-f_W(a+r)>0.
$

\item {If $0<a<r$,} then $a-r<0<a+r$.
By symmetry, $f_W(a-r)=f_W(-(a-r))=f_W(r-a)$ with $0<r-a<a+r$.
Since $f_W$ is strictly decreasing on $(0,\infty)$, $f_W(r-a)>f_W(a+r)$, and therefore
$
\frac{d}{ds}F_{c,r}(s)=f_W(a-r)-f_W(a+r)>0.
$
\end{itemize}

We conclude that $\dfrac{d}{ds}F_{c,r}(s)>0$ for all $s<c$.

\item \emph{Case 2: $s=c$ ($a=0$).} In this case, we have $
\frac{d}{ds}F_{c,r}(s)=f_W(-r)-f_W(r)=0
$
by symmetry.

\item \emph{Case 3: $s>c$ ($a<0$).}
Write $a=-b$ with $b>0$ and consider $\dfrac{d}{ds}F_{c,r}(s)=f_W(-b-r)-f_W(-b+r)$.
We split according to the sign of $-b+r$.
\begin{itemize}
\item {If $r\le b$,} then $-b-r<-b+r\le 0$, so both arguments are negative.
Since $f_W$ is strictly increasing on $(-\infty,0)$, we have $f_W(-b-r)<f_W(-b+r)$, hence
$
\frac{d}{ds}F_{c,r}(s)<0.
$

\item {If $r>b$,} then $-b-r<0< -b+r$, and by symmetry
$f_W(-b-r)=f_W(b+r)$ and $f_W(-b+r)=f_W(r-b)$ with $0<r-b<b+r$.
As $f_W$ is strictly decreasing on $(0,\infty)$, $f_W(b+r)<f_W(r-b)$; therefore
$
\frac{d}{ds}F_{c,r}(s)=f_W(-b-r)-f_W(-b+r)<0.
$
\end{itemize}

We conclude that $\dfrac{d}{ds}F_{c,r}(s)<0$ for all $s>c$.
\end{itemize}

Combining the three cases,
\begin{equation}\label{eq:monot-interval-prob}
\frac{d}{ds}F_{c,r}(s)
\begin{cases}
>0, & s<c,\\[0.3ex]
=0, & s=c,\\[0.3ex]
<0, & s>c,
\end{cases}
\end{equation}
so for any $r > 0$, $F_{c,r}(s)$ is strictly increasing for $s<c$, strictly decreasing for $s>c$,
and attains its unique maximum at $s=c$.

\textbf{Step 6. Monotonicity and minimization of $J(s)$.}

\emph{(6.1) Representation and derivative.}
Because $H$ is continuous with range contained in $[H(y^*),1-\alpha)$, we can write:
\begin{align*}
J(s)&=\mathbb{E}[H(W+s)]
=H(y^*)+\int_{H(y^*)}^{1-\alpha}\Big(1-\mathbb{P}\big(H(W+s)\le\tau\big)\Big)\,d\tau\\
&=H(y^*)+\int_{H(y^*)}^{1-\alpha}\Big(1-F_{c(\tau),r(\tau)}(s)\Big)\,d\tau.
\end{align*}
Since $|F'_{c(\tau),r(\tau)}(s)|\le2\|f_W\|_\infty$ uniformly in $s$,
differentiation under the integral sign is justified, giving
\[
J'(s)=-\int_{H(y^*)}^{1-\alpha}\frac{d}{ds}F_{c(\tau),r(\tau)}(s)\,d\tau.
\]

\emph{(6.2) Existence of a minimizer.}
As $|s|\to\infty$, $H(W+s)\to1-\alpha$ for each fixed $W$. Since $H$ is bounded from above by $1-\alpha$, we can apply dominated convergence theorem and get
\[
\lim_{|s|\to\infty}J(s) =  \lim_{|s|\to\infty} \mathbb{E}[H(W+s)] =  \mathbb{E}\left[\lim_{|s|\to\infty}H(W+s) \right] = 1-\alpha.
\]
Since $H(y)$ is bounded below for all $y$, the mapping
$J(s)=\int_{\mathbb{R}}H(w+s)f_W(w)\,dw$
is continuous and bounded below.
Moreover, because $H(y)<1-\alpha$ for all finite $y$,
we have $J(s)<1-\alpha$ for all finite $s$.
Hence the global minimum of $J$ is strictly smaller than its limiting value at the tails and it must happen at some finite point $s^*\in\mathbb{R}$.

\emph{(6.3) Sign of $J'(s)$ and characterization of $s^*$.}
From (6.1),
\[
J'(s)=-\int_{H(y^*)}^{1-\alpha}\frac{d}{ds}F_{c(\tau),r(\tau)}(s)\,d\tau.
\]
For each fixed $\tau$, the function $F_{c(\tau),r(\tau)}(s)$ is unimodal and
attains its maximum at $s=c(\tau)$ by \eqref{eq:monot-interval-prob}, i.e.,
 $\frac{d}{ds}F_{c(\tau),r(\tau)}(s)$ is positive for $s<c(\tau)$
and negative for $s>c(\tau)$.

Recall from (5.1) that
\[
\operatorname{sign}\big(c(\tau)-y^*\big)=\operatorname{sign}(a_2-a_1)
\qquad\forall\,\tau\in(H(y^*),1-\alpha).
\]
We analyze the sign of $J'(s)$ for each parameter regime.
\begin{itemize}
\item \emph{Case 1: $a_2>a_1$.}
Then $c(\tau)>y^*>0$ for all $\tau$.
For every $s\le0$, we automatically have $s<c(\tau)$, hence
$\frac{d}{ds}F_{c(\tau),r(\tau)}(s)>0$ for all $\tau$.
Substituting into the integral expression for $J'(s)$ gives
\[
J'(s)=-\int_{H(y^*)}^{1-\alpha}\!\frac{d}{ds}F_{c(\tau),r(\tau)}(s)\,d\tau<0
\quad\text{for all }s\le0.
\]
Thus $J$ is strictly decreasing on $(-\infty,0]$.
Since $J$ is continuous and tends to the same finite limit $1-\alpha$ at both tails,
no minimizer can lie in $(-\infty,0]$, so the global minimum satisfies $s^*>0$.
That is, the optimal inflation ratio $r^*=e^{s^*}>1$.

\item \emph{Case 2: $a_1>a_2$.}
The argument is symmetric with the previous case. Here, $c(\tau)<y^*<0$ for all $\tau$.
For every $s\ge0$, we have $s>c(\tau)$, so
$\frac{d}{ds}F_{c(\tau),r(\tau)}(s)<0$, and hence $J'(s)>0$ for all $s\ge0$.
Therefore $J$ is strictly increasing on $[0,\infty)$, and the global minimum
must satisfy $s^*<0$ (equivalently $r^*<1$).

\item \emph{Case 3: $a_1=a_2$.}
Then $c(\tau)\equiv y^*=0$ for all $\tau$.
Each $F_{0,r(\tau)}(s)$ is symmetric and maximized at $s=0$,
so $\frac{d}{ds}F_{0,r(\tau)}(s)$ is positive for $s<0$, negative for $s>0$,
and zero at $s=0$.
Consequently $J'(s)>0$ for $s<0$, $J'(s)<0$ for $s>0$, and $J'(0)=0$.
Hence $J$ is even and strictly increasing in $|s|$, attaining its unique
global minimum at $s^*=0$.
\end{itemize}

\textbf{Step 6. Conclusions for the inflation ratio.}
Recall $s=\log r$. From Step~5:
\begin{itemize}
\item If $a_1=a_2$, then $s^*=0$ and $r^*=e^{s^*}=1$.
\item If $a_2>a_1$, then $s^*>0$ and $r^*=e^{s^*}>1$.
\item If $a_1>a_2$, then $s^*<0$ and $r^*=e^{s^*}<1$.
\end{itemize}
This completes the proof. \hfill\(\blacksquare\)

%%% EC.3 matches Main Section 6 %%%
\section{Proofs and Supplemental Material for Section \ref{sec:approximate k}} \label{ec:approximate k}

\subsection{Proof of Lemma~\ref{lem_A1}}

By the definition of the event $\mathcal{E}(\vec{S}, \vec{\epsilon}, \vec{c})$, 
we have $\sigma_i^2 \le \ell_i S_i^2$ for all $i \in [M]$, since 
$\ell_i = \overline{\varphi}(\epsilon_i, c_i)$ corresponds to the upper bound of the confidence interval for $\sigma_i^2$. Because the function $\beta(\tilde{\sigma}_i, n_i)$ is increasing in 
$\tilde{\sigma}_i$, it follows that, on the event $\mathcal{E}(\vec{S}, 
\vec{\epsilon}, \vec{c})$,
\begin{eqnarray*}
\tilde{\beta}^*(\vec{\ell}, \vec{S}, \vec{\sigma}) 
\ \le \ 
\tilde{\beta}^*(\vec{\ell}, \vec{S}, \vec{\Sigma}(\vec{\ell}, \vec{S})) 
\ = \ 
\tilde{\beta}^{\text{robust}}(\vec{S}, \vec{\epsilon}, \vec{\sigma}, \vec{c}).
\end{eqnarray*}
As for the probability statement, note that each individual event 
$\left\{ \sigma_i^2 \in CI_i(S_i, \epsilon_i, c_i) \right\}$ occurs 
with probability $c_i$ by the definition of the confidence interval. 
Since we assume that the experiments are independent, we have:
\[
\mathbb{P}(\mathcal{E}(\vec{S}, \vec{\epsilon}, \vec{c})) 
= \prod_{i=1}^M \mathbb{P}(\sigma_i^2 \in CI_i(S_i, \epsilon_i, c_i)) 
= \prod_{i=1}^M c_i.
\]
This completes the proof. \hfill $\blacksquare$

\subsection{Proof of Lemma~\ref{lem_A2}}

Suppose $\ell_i = \overline{\varphi}(\epsilon_i, c_i)$ for all 
$i \in [M]$. By the definition of 
$\tilde{\beta}^{\text{robust}}(\vec{S}, \vec{\epsilon}, 
\vec{\sigma}, \vec{c})$, we have:
\begin{eqnarray*}
\tilde{\beta}^{\text{robust}}(\vec{S}, \vec{\epsilon}, \vec{\sigma}, \vec{c}) 
&=& \tilde{\beta}^*(\vec{\ell}, \vec{S}, \vec{\Sigma}(\vec{\ell}, \vec{S})) \\
&=& \max_{i \in [M]} 
\left\{ \beta( \sqrt{\ell_i} S_i,\ n^*_i(\vec{k}, \vec{S})) \right\} \\
&=& \Phi\left( 
q_{1-\alpha} - 
\sqrt{ 
\frac{N}
{\sum_{i=1}^M \ell_i \left( \frac{S_i}{\Delta_i} \right)^2}
}
\right),
\end{eqnarray*}
where the last equality follows from the closed-form expression of 
$\beta(\cdot)$ given in (\ref{eq:power}).

Now, on the event $\mathcal{E}(\vec{S}, \vec{\epsilon}, \vec{c})$, we have:
\[
\underline{\varphi}(\epsilon_i, c_i) S_i^2 
\ \le\ \sigma_i^2 
\ \le\ \overline{\varphi}(\epsilon_i, c_i) S_i^2,
\]
which implies:
\[
\frac{\sigma_i^2}{\overline{\varphi}(\epsilon_i, c_i)} 
\ \le\ S_i^2 
\ \le\ \frac{\sigma_i^2}{\underline{\varphi}(\epsilon_i, c_i)}.
\]
Substituting this into the expression above, and using the fact that 
$\Phi(\cdot)$ is increasing, we obtain:
\begin{eqnarray*}
\tilde{\beta}^{\text{robust}}(\vec{S}, \vec{\epsilon}, \vec{\sigma}, \vec{c}) 
&=& 
\Phi\left( 
q_{1-\alpha} - 
\sqrt{ 
\frac{N}
{\sum_{i=1}^M \ell_i \left( \frac{S_i}{\Delta_i} \right)^2}
}
\right) \\
&\le&
\Phi\left( 
q_{1-\alpha} - 
\sqrt{ 
\frac{N}
{\sum_{i=1}^M \kappa_i(\epsilon_i, c_i) 
\left( \frac{\sigma_i}{\Delta_i} \right)^2}
}
\right),
\end{eqnarray*}
where 
$\kappa(\epsilon_i, c_i) := 
\overline{\varphi}(\epsilon_i, c_i) \big/ 
\underline{\varphi}(\epsilon_i, c_i)$.

As for the lower bound, observe that on the event 
$\mathcal{E}(\vec{S}, \vec{\epsilon}, \vec{c})$, we have:
\[
\tilde{\beta}^{\text{robust}}(\vec{S}, \vec{\epsilon}, \vec{\sigma}, \vec{c}) 
\ \ge\ 
\tilde{\beta}^*(\vec{\ell}, \vec{S}, \vec{\sigma}) 
\ \ge\ 
\beta^*(\vec{\sigma}),
\]
where the first inequality follows from Lemma~\ref{lem_A1}, and the second 
from the fact that $\beta(\tilde{\sigma}_i, n_i)$ is increasing in 
$\tilde{\sigma}_i$. This completes the proof. \hfill $\blacksquare$

\subsection{Proof of Lemma \ref{lem_kappa}}

\underline{Part (i)}

Recall that for $\nu = \epsilon_i - 1$, we have
\[
\kappa(\epsilon_i, c_i)
= 
\frac{\chi^2_{(1+c_i)/2,\nu}}
     {\chi^2_{(1-c_i)/2,\nu}},
\]
where $\chi^2_{p,\nu}$ denotes the $p$-quantile of the chi-squared distribution with $\nu$ degrees of freedom.

Let $X_\nu \sim \chi^2_\nu$. It is known that 
\[
\frac{X_\nu}{\nu} \xrightarrow{p} 1
\]
as $\nu \to \infty$. Hence, the distribution of $X_\nu/\nu$ converges weakly to a degenerate distribution at~1.
For distributions with continuous and strictly increasing CDFs, convergence in distribution to a point mass
implies convergence of their quantiles to that point. Therefore, for every fixed $p \in (0,1)$,
\[
\lim_{\nu \to \infty} \frac{\chi^2_{p,\nu}}{\nu} = 1.
\]
Applying this to the numerator and denominator of $\kappa(\epsilon_i,c_i)$ gives
\[
\lim_{\epsilon_i \to \infty} 
\kappa(\epsilon_i, c_i)
=\lim_{\nu \to \infty}
\frac{\chi^2_{(1+c_i)/2,\nu}/\nu}
     {\chi^2_{(1-c_i)/2,\nu}/\nu}
=1.
\]
Hence $\displaystyle \lim_{\epsilon_i \to \infty} \kappa(\epsilon_i,c_i)=1$ for all $c_i\in[0,1)$.

\underline{Part (ii)}

\textbf{Continuity and Differentiability.}
The quantile function $\chi^2_{p,\nu}$ is continuous and differentiable in~$p$ for
$p\in(0,1)$.  For $c_i\in[0,1)$, the arguments $(1\pm c_i)/2$ lie in $(0,1)$, and
the denominator $\chi^2_{(1-c_i)/2,\nu}$ is strictly positive.
Hence $\kappa(\epsilon_i,c_i)$ is continuous and differentiable in~$c_i$.

\textbf{Convexity.}
Fix $\epsilon_i\ge2$ and write $\nu=\epsilon_i-1$.  Let $F_\nu$, $f_\nu$, and
$Q_\nu=F_\nu^{-1}$ denote the CDF, PDF, and quantile function of
$\chi^2_\nu$.  For $c_i\in[0,1)$,
\[
\kappa(\epsilon_i,c_i)
=\frac{Q_\nu\!\left(\frac{1+c_i}{2}\right)}
       {Q_\nu\!\left(\frac{1-c_i}{2}\right)}.
\]
We show that $c_i\mapsto\kappa(\epsilon_i,c_i)$ is strictly convex on $[0,1)$.

Define
\[
p_+(c):=\tfrac{1+c}{2},\qquad p_-(c):=\tfrac{1-c}{2},\qquad
\phi(c):=\log\kappa(\epsilon_i,c)
        =\log Q_\nu(p_+(c))-\log Q_\nu(p_-(c)).
\]
Since $\kappa=\exp(\phi)$ and $\exp$ is increasing and convex, it suffices to
show that $\phi$ is convex; indeed,
$\kappa''(c)=\kappa(c)\big(\phi''(c)+(\phi'(c))^2\big)>0$
whenever $\phi''(c)\ge0$.

Because $f_\nu>0$ and continuously differentiable on $(0,\infty)$, the inverse
function theorem yields
\[
Q_\nu'(p)=\frac{1}{f_\nu(Q_\nu(p))},\qquad p\in(0,1).
\]
Hence,
\begin{align*}
\phi'(c)
 &=\frac{1}{2}\frac{Q_\nu'(p_+(c))}{Q_\nu(p_+(c))}
  +\frac{1}{2}\frac{Q_\nu'(p_-(c))}{Q_\nu(p_-(c))}
  =\tfrac12\,g(p_+(c))+\tfrac12\,g(p_-(c)),\\
g(p)&:=\frac{Q_\nu'(p)}{Q_\nu(p)}
     =\frac{1}{Q_\nu(p)\,f_\nu(Q_\nu(p))}>0.
\end{align*}
Differentiating again and using $p_+'(c)=\tfrac12$,
$p_-'(c)=-\tfrac12$ gives
\[
\phi''(c)=\tfrac14\big[g'(p_+(c))-g'(p_-(c))\big].
\]
Thus $\phi''(c)\ge0$ will follow once we show that $g'$ is strictly increasing
on $(0,1)$.

Let $x:=Q_\nu(p)>0$.  Since $Q_\nu'(p)=1/f_\nu(x)$, define
$h(x):=(x f_\nu(x))^{-1}$, so that $g(p)=h(Q_\nu(p))$.  By the chain rule,
\[
g'(p)=h'(Q_\nu(p))\,Q_\nu'(p)
     =\frac{h'(x)}{f_\nu(x)},\qquad
h'(x)=-\frac{f_\nu(x)+x f_\nu'(x)}{(x f_\nu(x))^2}.
\]
Therefore
\[
g'(p)=-\frac{f_\nu(x)+x f_\nu'(x)}{x^2 f_\nu(x)^3}.
\]
For the $\chi^2_\nu$ density
$f_\nu(x)=C x^{\nu/2-1}e^{-x/2}$ ($C>0$),
\[
\frac{f_\nu'(x)}{f_\nu(x)}=\frac{\nu/2-1}{x}-\frac12
\quad\Rightarrow\quad
f_\nu(x)+x f_\nu'(x)
  =f_\nu(x)\!\left(\frac{\nu}{2}-\frac{x}{2}\right)
  =\frac{f_\nu(x)}{2}(\nu-x).
\]
Hence
\[
g'(p)=\frac{x-\nu}{2x^2 f_\nu(x)^2}
      =:H(x).
\]
To establish monotonicity, observe that
$f_\nu(x)^2=C^2 x^{\nu-2}e^{-x}$, so
\[
H(x)=\frac{e^{x}}{2C^2}\frac{x-\nu}{x^{\nu}}
   =K\,e^{x}(x-\nu)x^{-\nu},
   \qquad K=\frac{1}{2C^2}>0.
\]
Let $S(x)=(x-\nu)e^{x}x^{-\nu}$.
Then
\[
\frac{S'(x)}{S(x)}
  =\frac{1}{x-\nu}+1-\frac{\nu}{x},\qquad
H'(x)=K\,e^{x}x^{-\nu}(x-\nu)
      \!\left(\frac{1}{x-\nu}+1-\frac{\nu}{x}\right)
      =K\,e^{x}x^{-\nu}\frac{x+(x-\nu)^2}{x}>0
\]
for all $x>0$.  Thus $H$ is strictly increasing on $(0,\infty)$.
Because $Q_\nu$ is strictly increasing ($Q_\nu'(p)=1/f_\nu(Q_\nu(p))>0$),
the composition $g'(p)=H(Q_\nu(p))$ is strictly increasing on $(0,1)$.

Consequently, for $c\in(0,1)$ we have $p_+(c)>p_-(c)$ and
\[
\phi''(c)
  =\tfrac14\big[g'(p_+(c))-g'(p_-(c))\big]
  >0.
\]
At $c=0$, $p_+=p_-=\tfrac12$, so $\phi''(0)=0$ while $\phi'(0)=g(\tfrac12)>0$.
Finally,
\[
\kappa''(c)=\kappa(c)\big(\phi''(c)+(\phi'(c))^2\big)>0,
\]
showing that $\kappa(\epsilon_i,\cdot)$ is strictly convex on $[0,1)$.

\textbf{Monotonicity.}
As $c_i$ increases, the upper quantile
$\chi^2_{(1+c_i)/2,\nu}$ increases while the lower quantile
$\chi^2_{(1-c_i)/2,\nu}$ decreases; hence
$\kappa(\epsilon_i,c_i)$ increases strictly in~$c_i$.

\textbf{Limits.}
As $c_i\to0^+$, both arguments of the quantiles approach~0.5, so
\[
\lim_{c_i\to0^+}\kappa(\epsilon_i,c_i)
  =\frac{\chi^2_{0.5,\nu}}{\chi^2_{0.5,\nu}}=1.
\]
As $c_i\to1^-$, the numerator argument tends to~1 and the denominator argument
to~0.  Because $\chi^2_{p,\nu}\to\infty$ as $p\to1$ and
$\chi^2_{p,\nu}\to0$ as $p\to0$, we obtain
\[
\lim_{c_i\to1^-}\kappa(\epsilon_i,c_i)
  =\frac{\infty}{0}=\infty.
\]

\underline{Part (iii)}

Fix $x$. For notational brevity, we use $c_i(\epsilon_i) := c_i(\epsilon_i, x)$. We prove the claim by contradiction. Suppose that $\lim_{\epsilon_i \to \infty} c_i(\epsilon_i) \neq 1$.  
Then, there exists a subsequence, denoted by $\{\epsilon_{i,k}\}$, such that 
$\epsilon_{i,k} \to \infty$ and $c_i(\epsilon_{i,k}) \to c^*$ for some $c^* < 1$.
From Part~(i), for every fixed $c \in [0,1)$ we have 
$\lim_{\epsilon_i \to \infty} \kappa(\epsilon_i,c) = 1$.  
Since $c_i(\epsilon_{i,k}) \to c^* < 1$, we can examine the limit of 
$\kappa(\epsilon_{i,k}, c_i(\epsilon_{i,k}))$ along this subsequence.  
For any $\delta > 0$, sufficiently large $k$ ensures that 
$c^* - \delta < c_i(\epsilon_{i,k}) < c^* + \delta$.  
Because $\kappa(\epsilon_i,c)$ is strictly increasing in~$c$ 
(from Part~(ii)), we have
\[
\kappa(\epsilon_{i,k},\, c^*-\delta)
   < \kappa(\epsilon_{i,k},\, c_i(\epsilon_{i,k}))
   < \kappa(\epsilon_{i,k},\, c^*+\delta).
\]
Taking limits as $k \to \infty$ and using Part~(i) gives
\[
\lim_{k\to\infty}\kappa(\epsilon_{i,k},\, c^*-\delta)
   \le
   \lim_{k\to\infty}\kappa(\epsilon_{i,k},\, c_i(\epsilon_{i,k}))
   \le
   \lim_{k\to\infty}\kappa(\epsilon_{i,k},\, c^*+\delta),
\]
and both bounds equal~1.  
Hence
\[
\lim_{k\to\infty}\kappa(\epsilon_{i,k},\, c_i(\epsilon_{i,k})) = 1.
\]
However, by definition $c_i(\epsilon_i)$ satisfies 
$\kappa(\epsilon_i, c_i(\epsilon_i)) = x > 1$ for all~$\epsilon_i$, so
$\kappa(\epsilon_{i,k}, c_i(\epsilon_{i,k})) = x$ for every $k$.
This contradicts the limit above, which forces $x=1$.  
Therefore our assumption is false, and we conclude that
\[
\lim_{\epsilon_i \to \infty} c_i(\epsilon_i) = 1.
\]
\hfill$\square$

\subsection{Proof of Proposition~\ref{prop_A1}}

For a fixed confidence level $\gamma \in (0,1)$, consider any vector 
$\vec{c} \in [0,1)^M$ such that $\prod_{i=1}^M c_i \ge \gamma$, 
and let $k_i = \overline{\varphi}(\epsilon_i, c_i)$ for all $i \in [M]$. 
Then, by Lemma~\ref{lem_A2}, the pair $(\vec{k}, \delta)$ where
\begin{eqnarray*}
\delta 
= \Phi\left( 
q_{1-\alpha} - 
\sqrt{ 
\frac{N}
{\sum_{i=1}^M \kappa(\epsilon_i, c_i) 
\left( \frac{\sigma_i}{\Delta_i} \right)^2}
}
\right) - \beta^*(\vec\sigma)
\end{eqnarray*}
is a feasible solution for \textbf{TOL}. 
Thus, by construction, we have 
$\delta^*(\gamma, \vec{\epsilon}, \vec{\sigma}) 
\le \delta^R(\gamma, \vec{\epsilon}, \vec{\sigma})$.

\underline{Part (i)}

It is sufficient to prove that $\delta^R(\gamma, \vec{\epsilon}, \vec{\sigma}) \to 0$ as $\gamma \to 0$. The claim $\delta^*(\gamma, \vec{\epsilon}, \vec \sigma) \to 0$ will immediately follow from  $\delta^*(\gamma, \vec{\epsilon}, \vec{\sigma}) 
\le \delta^R(\gamma, \vec{\epsilon}, \vec{\sigma})$. To prove $\delta^R(\gamma, \vec{\epsilon}, \vec{\sigma}) \to 0$ as $\gamma \to 0$, note that $\vec c = (\gamma^{1/M}, \dots, \gamma^{1/M})$ is feasible to \textbf{TOL} for small $\gamma$ and, therefore, we can bound
\begin{eqnarray*}
\delta^R(\gamma, \vec{\epsilon}, \vec{\sigma}) \ \le \ \Phi\left( 
q_{1-\alpha} - 
\sqrt{ 
\frac{N}
{\sum_{i=1}^M \kappa(\epsilon_i, \gamma^{1/M}) 
\left( \frac{\sigma_i}{\Delta_i} \right)^2}
}
\right) - \beta^*(\vec\sigma) \ := \ \bar\delta(\gamma, \vec\epsilon, \vec\sigma).
\end{eqnarray*}
Since $\kappa(\epsilon_i, 0) = 1$ by definition, as $\gamma \to 0$, $\kappa(\epsilon_i, \gamma^{1/M}) \to 1$ for all $i \in [M]$. Consequently, we have $\bar\delta(\gamma, \vec\epsilon, \vec\sigma) \to 0$, which implies $\delta^R(\gamma, \vec\epsilon, \vec\sigma) \to 0$.

\underline{Part (ii)}

By definition, we always have $\delta^*(\gamma, \vec{\epsilon}, \vec{\sigma}) \le 1 - \alpha - \beta^*(\vec \sigma)$ and $\delta^R(\gamma, \vec{\epsilon}, \vec{\sigma}) \le 1 - \alpha - \beta^*(\vec \sigma)$. The claim $\delta^R(\gamma, \vec{\epsilon}, \vec{\sigma}) \to 1 - \alpha - \beta^*(\vec \sigma)$ as $\gamma \to 1$ follows directly by noticing that we need to set $c_i \to 1$ for all $i \in [M]$ to satisfy the first constraint in \textbf{R-TOL} as $\gamma \to 1$. 

We now show that $\delta^*(\gamma, \vec{\epsilon}, \vec{\sigma}) \to 1 - \alpha - \beta^*(\vec \sigma)$ as $\gamma \to 1$. Let $(\vec k^*(\gamma), \delta^*(\gamma))$ denote the optimal solution to \textbf{TOL}. Note that 
\begin{eqnarray*}
\gamma &\le & \mathbb{P}(
\tilde{\beta}^*(\vec{k}^*(\gamma), \vec{S}, \vec{\sigma}) \le \beta^*(\vec\sigma) + \delta^*(\gamma)) \\[1mm]
& = & \mathbb{P}\left(\max_{i \in [M]} \ \Phi\left( q_{1-\alpha} - \frac{\Delta_i \sqrt{N}}{\sigma_i} \sqrt{\frac{k^*_i(\gamma) \left(\frac{S_i}{\Delta_i}\right)^2}{\sum_{j \in [M]} k^*_j(\gamma) \left(\frac{S_j}{\Delta_j}\right)^2}} \right) \le \beta^*(\vec\sigma) + \delta^*(\gamma) \right) \\[1mm]
& \le & \mathbb{P}\left(\max_{i \in [2]} \ \Phi\left( q_{1-\alpha} - \frac{\Delta_i \sqrt{N}}{\sigma_i} \sqrt{\frac{k^*_i(\gamma) \left(\frac{S_i}{\Delta_i}\right)^2}{\sum_{j \in [M]} k^*_j(\gamma) \left(\frac{S_j}{\Delta_j}\right)^2}} \right) \le \beta^*(\vec\sigma) + \delta^*(\gamma) \right) \\[1mm]
& \le & \mathbb{P}\left(\max_{i \in [2]} \ \Phi\left( q_{1-\alpha} - \frac{\Delta_i \sqrt{N}}{\sigma_i} \sqrt{\frac{k^*_i(\gamma) \left(\frac{S_i}{\Delta_i}\right)^2}{\sum_{j \in [2]} k^*_j(\gamma) \left(\frac{S_j}{\Delta_j}\right)^2}} \right) \le \beta^*(\vec\sigma) + \delta^*(\gamma) \right),
\end{eqnarray*}
where the second inequality holds because the maximum of $M$ terms is at least as large as the maximum of two terms and the last inequality holds because $\sum_{j \in [2]} k^*_j(\gamma) (\frac{S_j}{\Delta_j})^2 \le \sum_{j \in [M]} k^*_j(\gamma) (\frac{S_j}{\Delta_j})^2$.

Now, define 
\begin{eqnarray*}
r^*_{ij}(\gamma) := \frac{k^*_i(\gamma)}{k^*_j(\gamma)}. 
\end{eqnarray*}
There are three cases:
\begin{itemize}
\item Case 1: $\liminf_{\gamma \to 1} r^*_{12}(\gamma) > 0$,
\item Case 2: $\liminf_{\gamma \to 1} r^*_{12}(\gamma) = \limsup_{\gamma \to 1} r^*_{12}(\gamma) = 0$;
\item Case 3: $\liminf_{\gamma \to 1} r^*_{12}(\gamma) = 0$ and  $\limsup_{\gamma \to 1} r^*_{12}(\gamma) > 0$.
\end{itemize}

We start with Case 1. In this case, there must exist $\bar r > 0$ such that $r^*_{12}(\gamma) > \bar r$ for all sufficiently large $\gamma$ (i.e., for $\gamma$ sufficiently close to 1). In other words, for sufficiently large $\gamma$, we can bound:
\begin{eqnarray*}
\gamma &\le & \mathbb{P}\left(\max_{i \in [2]} \ \Phi\left( q_{1-\alpha} - \frac{\Delta_i \sqrt{N}}{\sigma_i} \sqrt{\frac{k^*_i(\gamma) \left(\frac{S_i}{\Delta_i}\right)^2}{\sum_{j \in [2]} k^*_j(\gamma) \left(\frac{S_j}{\Delta_j}\right)^2}} \right) \le \beta^*(\vec\sigma) + \delta^*(\gamma) \right) \\[1mm]
&\le & \mathbb{P}\left(\Phi\left( q_{1-\alpha} - \frac{\Delta_2 \sqrt{N}}{\sigma_2} \sqrt{\frac{k^*_2(\gamma) \left(\frac{S_2}{\Delta_2}\right)^2}{\sum_{j \in [2]} k^*_j(\gamma) \left(\frac{S_j}{\Delta_j}\right)^2}} \right) \le \beta^*(\vec\sigma) + \delta^*(\gamma) \right) \\[1mm]
& = & \mathbb{P}\left(\Phi\left( q_{1-\alpha} - \frac{\Delta_2 \sqrt{N}}{\sigma_2} \sqrt{\frac{\left(\frac{S_2}{\Delta_2}\right)^2}{ r^*_{12}(\gamma) \left(\frac{S_1}{\Delta_1}\right)^2 + \left(\frac{S_2}{\Delta_2}\right)^2}} \right) \le \beta^*(\vec\sigma) + \delta^*(\gamma) \right) \\[1mm]
& \le & \mathbb{P}\left(\Phi\left( q_{1-\alpha} - \frac{\Delta_2 \sqrt{N}}{\sigma_2} \sqrt{\frac{\left(\frac{S_2}{\Delta_2}\right)^2}{ \bar r \left(\frac{S_1}{\Delta_1}\right)^2 + \left(\frac{S_2}{\Delta_2}\right)^2}} \right) \le \beta^*(\vec\sigma) + \delta^*(\gamma) \right). 
\end{eqnarray*}
Note that the random variable 
\begin{eqnarray*}
W := \Phi\left( q_{1-\alpha} - \frac{\Delta_2 \sqrt{N}}{\sigma_2} \sqrt{\frac{\left(\frac{S_2}{\Delta_2}\right)^2}{ \bar r \left(\frac{S_1}{\Delta_1}\right)^2 + \left(\frac{S_2}{\Delta_2}\right)^2}} \right)
\end{eqnarray*}
is independent of $\gamma$. It is a continuous random variable with support in $(0, 1-\alpha]$. Thus, as $\gamma \to 1$, to guarantee that the last inequality in the above holds, we must have $\delta^*(\gamma) \to 1-\alpha-\beta^*(\vec \sigma)$.  

We now consider Case 2. In this case, $\lim_{\gamma \to 1} r^*_{21}(\gamma) = \infty > 0$ exists. Thus, there must exist $\bar r > 0$ such that $r^*_{21}(\gamma) > \bar r$ for all sufficiently large $\gamma$. The remainder of the argument proceeds in the same way as in Case 1 but by focusing on experiment 1 instead of experiment 2.

As for case 3, note that it implies $\liminf_{\gamma \to 1} r^*_{21}(\gamma) > 0$ and  $\limsup_{\gamma \to 1} r^*_{21}(\gamma) = \infty > 0$. Since $\liminf_{\gamma \to 1} r^*_{21}(\gamma) > 0$, again, there must exist $\bar r > 0$ such that $r^*_{21}(\gamma) > \bar r$ for all sufficiently large $\gamma$. The remainder of the argument proceeds in the same way as in Case 1 but by focusing on experiment 1 instead of experiment 2.

\underline{Part (iii)}

Similar to part (i), it is sufficient to prove that $\delta^R(\gamma, \vec{\epsilon}, \vec{\sigma}) \to 0$ as $\gamma \to 0$. The claim $\delta^*(\gamma, \vec{\epsilon}, \vec \sigma) \to 0$ will immediately follow from  $\delta^*(\gamma, \vec{\epsilon}, \vec{\sigma}) \le \delta^R(\gamma, \vec{\epsilon}, \vec{\sigma})$.

By the same argument used in the proof of part (i), we can bound: 
\begin{eqnarray*}
\delta^R(\gamma, \vec \epsilon, \vec \sigma ) \ \le \ \Phi\left(q_{1-\alpha} - \sqrt{\frac{N}{\sum_{i=1}^M \kappa(\epsilon_i, \gamma^{1/M}) \left(\frac{\sigma_i}{\Delta_i}\right)^2}}\right) - \beta^*(\vec\sigma).
\end{eqnarray*}
By Lemma \ref{lem_kappa}, for each $i \in [M]$, $\kappa(\epsilon_i, \gamma^{1/M}) \to 1$ as $\epsilon_i \to \infty$. Thus, the quantity on the right hand side goes to 0 as $\epsilon_i \to \infty$ for all $i \in [M]$, which implies $\delta^R(\gamma, \vec \epsilon, \vec \sigma ) \to 0$.  \hfill $\blacksquare$

\subsection{Proof of Proposition~\ref{prop_A2}}

Fix a tolerance level $\delta$. Consider any vector $\vec{c} \in (0,1]^M$ that is feasible for \textbf{R-CONF}, with $\prod_{i=1}^M c_i = \gamma$. Let $k_i = \overline{\varphi}(\epsilon_i, c_i)$ for all $i \in [M]$. Lemmas \ref{lem_A1} and \ref{lem_A2}, together with the first constraint in \textbf{R-CONF}, imply  
\begin{eqnarray*}
\tilde{\beta}^*(\vec{k}, \vec{S}, \vec{\sigma}) \ \le \ \tilde{\beta}^{\text{robust}}(\vec{S}, \vec{\epsilon}, \vec{\sigma}, \vec{c})
\ \le \ 
\Phi\left(
q_{1-\alpha} - 
\sqrt{
\frac{N}
{\sum_{i=1}^M 
\kappa_i(\epsilon_i, c_i) 
\left( \frac{\sigma_i}{\Delta_i} \right)^2}
}
\right) \ \le \ \beta^*(\vec \sigma) + \delta
\end{eqnarray*}
with probability at least $\gamma$. In other words, 
\begin{eqnarray*}
\mathbb{P}(
\tilde{\beta}^*(\vec{k}, \vec{S}, \vec{\sigma}) \le \beta^*(\vec\sigma) + \delta) \ge \gamma,
\end{eqnarray*}
which means that $(\vec k, \gamma)$ is feasible for \textbf{CONF}. Since any $\gamma$ that is attainable under \textbf{R-CONF} is also feasible for \textbf{CONF}, we conclude that $\gamma^*(\delta, \vec{\epsilon}, \vec{\sigma}) \ge \gamma^R(\delta, \vec{\epsilon}, \vec{\sigma}).$

\underline{Part (i)}

The claim $\gamma^*(\delta, \vec{\epsilon}, \vec{\sigma}) \to 0$ as $\delta \to 0$ holds because, by definition, $\tilde \beta^*(\vec k, \vec S, \vec\sigma)$ is a continuous random variable and $\tilde \beta^*(\vec k, \vec S, \vec\sigma) \ge \beta^*(\vec\sigma)$ almost surely for all $\vec k \ge \vec 1$ and $\vec S$, which imply 
\begin{eqnarray*}
\mathbb{P}(\tilde \beta^*(\vec k, \vec S, \vec\sigma) \le \beta^*(\vec\sigma)) = 1 - \mathbb{P}(\tilde \beta^*(\vec k, \vec S, \vec\sigma) \ge \beta^*(\vec\sigma)) = 0. 
\end{eqnarray*}
The claim $\gamma^R(\delta, \vec{\epsilon}, \vec{\sigma}) \to 0$ as $\delta \to 0$ holds since $\gamma^*(\delta, \vec{\epsilon}, \vec{\sigma}) \ge \gamma^R(\delta, \vec{\epsilon}, \vec{\sigma}).$

\underline{Part (ii)}

For the claim $\gamma^R(\delta, \vec{\epsilon}, \vec{\sigma}) \to 1$ as $\delta \to 1-\alpha-\beta^*(\vec \sigma)$, note that if we set $\delta = 1-\alpha-\beta^*(\vec \sigma)$, by Lemma \ref{lem_kappa}, the first constraint in \textbf{R-CONF} is always satisfied for any $\vec c \in [0,1]^M$. Since the objective is to maximize $\prod_{i=1}^M c_i$, we can set $c_i \to 1$ for all $i \in [M]$. 

As for the claim $\gamma^*(\delta, \vec{\epsilon}, \vec{\sigma}) \to 1$ as $\delta \to 1-\alpha-\beta^*(\vec \sigma)$ holds since $\gamma^*(\delta, \vec{\epsilon}, \vec{\sigma}) \ge \gamma^R(\delta, \vec{\epsilon}, \vec{\sigma}).$  

\underline{Part (iii)}

First, note that the constraint
\begin{eqnarray*}
\delta \ \ge \ \Phi\left(q_{1-\alpha} - \sqrt{\frac{N}{\sum_{i=1}^M \kappa(\epsilon_i, c_i) \left(\frac{\sigma_i}{\Delta_i}\right)^2}}\right) - \beta^*(\vec\sigma)
\end{eqnarray*}
is equivalent to
\begin{eqnarray*}
\sum_{i=1}^M \kappa(\epsilon_i, c_i) \left(\frac{\sigma_i}{\Delta_i}\right)^2 \le \frac{N}{[q_{1-\alpha} - \Phi^{-1}(\delta + \beta^*(\vec\sigma))]^2}.
\end{eqnarray*}
Now, consider a solution $\vec c = (\hat c, \dots, \hat c)$, where $\hat c$ satisfies
\begin{eqnarray*}
\kappa\left(\min_{i \in [M]}\epsilon_i, \hat c\right) \cdot \sum_{i=1}^M \left(\frac{\sigma_i}{\Delta_i}\right)^2 = \frac{N}{[q_{1-\alpha} - \Phi^{-1}(\delta + \beta^*(\vec\sigma))]^2},
\end{eqnarray*}
or equivalently
\begin{eqnarray*}
\kappa\left(\min_{i \in [M]}\epsilon_i, \hat c\right) \ = \ \zeta \ := \  \frac{N}{[q_{1-\alpha} - \Phi^{-1}(\delta + \beta^*(\vec\sigma))]^2} \cdot \left[\sum_{i=1}^M \left(\frac{\sigma_i}{\Delta_i}\right)^2\right]^{-1}. 
\end{eqnarray*}
By construction, the vector $\vec c = (\hat c, \dots, \hat c)$ as defined above is feasible for \textbf{R-CONF}. So, we can bound $\gamma^R(\delta, \vec{\epsilon}, \vec{\sigma}) \ge \hat c^M$. More precisely, $\hat c$ is a function of $\min_{i \in [M]} \epsilon_i$ and $\zeta$, and we can express it as $\hat c(\min_{i\in[M]} \epsilon_i, \zeta)$. Now, $\epsilon_i \to \infty$ for all $i \in [M]$ implies $\min_{i\in[M]} \epsilon_i \to \infty$. By Lemma \ref{lem_kappa} part (iii), this further implies $\hat c(\min_{i\in[M]} \epsilon_i, \zeta) \ \to \ 1$. We conclude that $\gamma^R(\delta, \vec{\epsilon}, \vec{\sigma}) \to 1$ as $\epsilon_i \to \infty$ for all $i \in [M]$.

The claim $\gamma^*(\delta, \vec{\epsilon}, \vec{\sigma}) \to 1$ follows from $\gamma^*(\delta, \vec{\epsilon}, \vec{\sigma}) 
\ge \gamma^R(\delta, \vec{\epsilon}, \vec{\sigma}).$  \hfill $\blacksquare$

\subsection{Proof of Proposition~\ref{prop_A3}}

Fix $\vec k \ge \vec 1$. Due to the invertible mapping between $k_i$ and $c_i$, there exist $\vec c$ such that $k_i = \overline{\varphi}(\epsilon_i, c_i)$ for all $i \in [M]$. We can decompose the expected maximum 
Type 2 error as follows:
\begin{eqnarray*}
\mathbb{E}[
\tilde{\beta}^*(\vec{k}, \vec{S}, \vec{\sigma})] 
&=& \mathbb{E}[
\tilde{\beta}^*(\vec{k}, \vec{S}, \vec{\sigma}) \cdot 
\mathbf{1}\{\mathcal{E}(\vec{S}, \vec{\epsilon}, \vec{c})\}] 
+ \mathbb{E}[
\tilde{\beta}^*(\vec{k}, \vec{S}, \vec{\sigma}) \cdot 
\mathbf{1}\{\mathcal{E}^c(\vec{S}, \vec{\epsilon}, \vec{c})\}] \\
&\le& \mathbb{E}[
\tilde{\beta}^{\text{robust}}(\vec{S}, \vec{\epsilon}, \vec{\sigma}, \vec{c}) 
\cdot \mathbf{1}\{\mathcal{E}(\vec{S}, \vec{\epsilon}, \vec{c})\}] 
+ \mathbb{P}(\mathcal{E}^c(\vec{S}, \vec{\epsilon}, \vec{c})) \\
&\le& \Phi\left( 
q_{1-\alpha} - 
\sqrt{ 
\frac{N}
{\sum_{i=1}^M \kappa(\epsilon_i, c_i) 
\left( \frac{\sigma_i}{\Delta_i} \right)^2}
}
\right) \cdot \mathbb{P}(\mathcal{E}(\vec{S}, \vec{\epsilon}, \vec{c})) 
+ \mathbb{P}(\mathcal{E}^c(\vec{S}, \vec{\epsilon}, \vec{c})) \\
&=& 1 + \left[
\Phi\left( 
q_{1-\alpha} - 
\sqrt{ 
\frac{N}
{\sum_{i=1}^M \kappa(\epsilon_i, c_i) 
\left( \frac{\sigma_i}{\Delta_i} \right)^2}
}
\right) - 1
\right] \cdot \prod_{i=1}^M c_i,
\end{eqnarray*}
where the first inequality follows from Lemma~\ref{lem_A1} and the fact that $\tilde{\beta}^*(\vec{k}, \vec{S}, \vec{\sigma}) \le 1$, and the second inequality follows from Lemma~\ref{lem_A2}. We conclude that $g^*(\vec \epsilon, \vec \sigma) \le g^R(\vec \epsilon, \vec \sigma)$.

Next, note that we always have
\begin{eqnarray*}
1 + \left[ \Phi\left(q_{1-\alpha} - \sqrt{\frac{N}{\sum_{i=1}^M \kappa(\epsilon_i, c_i) \left(\frac{\sigma_i}{\Delta_i}\right)^2}}\right) - 1\right] \cdot \prod_{i=1}^M c_i & \ge & \Phi\left(q_{1-\alpha} - \sqrt{\frac{N}{\sum_{i=1}^M \kappa(\epsilon_i, c_i) \left(\frac{\sigma_i}{\Delta_i}\right)^2}}\right) \\[1mm]
& \ge & \beta^*(\vec \sigma).
\end{eqnarray*}
So, $g^R(\vec \epsilon, \vec \sigma) \ge \beta^*(\vec \sigma)$. We now show that $g^R(\vec \epsilon, \vec \sigma) \to \beta^*(\vec \sigma)$ as $\epsilon_i \to \infty$ for all $i \in [M]$. Let $\theta \in (0,1)$ be a number close to 1 and set $c_i$ such that $\kappa(\epsilon_i, c_i) = \theta$ for all $i \in [M]$. Since $c_i$ is a function of $\epsilon_i$ and $\theta$, we can write it as $c_i(\epsilon_i, \theta)$. By Lemma \ref{lem_kappa} part (iii), $c_i(\epsilon_i, \theta) \to 1$ as $\epsilon_i \to \infty$. Thus,
\begin{eqnarray*}
\limsup_{\epsilon_i \to \infty, \, \forall i \in [M]} g^R(\vec \epsilon, \vec \sigma) \ \le \ \Phi\left(q_{1-\alpha} - \sqrt{\frac{N}{\sum_{i=1}^M \theta \left(\frac{\sigma_i}{\Delta_i}\right)^2}}\right). 
\end{eqnarray*}
The above inequality holds for all $\theta \in (0,1)$ arbitrarily close to 1. Since  $g^R(\vec \epsilon, \vec \sigma) \ge \beta^*(\vec \sigma)$, we conclude that $g^R(\vec \epsilon, \vec \sigma) \to \beta^*(\vec \sigma)$ as $\epsilon_i \to \infty$ for all $i \in [M]$.  

The claim $g^*(\vec \epsilon, \vec \sigma) \to \beta^*(\vec \sigma)$ holds since, by definition, $g^*(\vec \epsilon, \vec \sigma) \ge \beta^*(\vec \sigma)$ and $g^*(\vec \epsilon, \vec \sigma) \to \beta^*(\vec \sigma)$.  \hfill $\blacksquare$

\subsection{Proof of Proposition \ref{prop_A4}}

Assume \(\epsilon_i=\epsilon\) for all \(i\in[M]\). Define \(x_i:=\log c_i\in(-\infty,0)\) and \(g_i(x):=\kappa(\epsilon,e^{x})\). 
Since \(\kappa(\epsilon,\cdot)\) is strictly increasing and strictly convex on \([0,1)\), we have
\[
g_i'(x)=e^{x}\kappa'(\epsilon,e^{x})>0,\qquad
g_i''(x)=e^{x}\kappa'(\epsilon,e^{x})+e^{2x}\kappa''(\epsilon,e^{x})>0,
\]
so each \(g_i\) is strictly increasing and strictly convex; in particular \(g_i'\) is strictly increasing on \((-\infty,0)\).

\underline{\textbf{R-TOL}}

From the analysis in Section \ref{subsec:solve_r-tol}, \textbf{R-TOL} can be simplified as
\begin{eqnarray*}
\min_{\vec{c}} && \sum_{i=1}^M g_i(x_i) 
\left(\frac{\sigma_i}{\Delta_i}\right)^2 \label{OBJ1} \\
\text{subject to} && \sum_{i=1}^M x_i \ge \log \gamma, \nonumber \\
&& x_i \in (-\infty, 0), \quad \forall i \in [M]. \nonumber
\end{eqnarray*}
the resulting KKT optimality condition for \textbf{R-TOL} must have
\[
g_i'(x_i)\,\left(\tfrac{\sigma_i}{\Delta_i}\right)^2-\lambda=0
\quad\Longleftrightarrow\quad
\boxed{\ g_i'(x_i)\,\left(\tfrac{\sigma_i}{\Delta_i}\right)^2=\lambda\ }
\]
with the right-hand side independent of \(i\). Since \(g_i'\) is strictly increasing, the above equation implies
\[
\left(\tfrac{\sigma_i}{\Delta_i}\right)^2<\left(\tfrac{\sigma_j}{\Delta_j}\right)^2 \Rightarrow x_i>x_j
\ \ (\text{hence }c_i>c_j),
\]

\underline{\textbf{R-CONF}}

Similar to \textbf{R-TOL}, we can rewrite the formulation of \textbf{R-CONF} as
\begin{eqnarray*}
\max_{\vec{c}} && \sum_{i=1}^M x_i \\
\text{subject to} && 
\sum_{i=1}^M g_i(x_i) \left(\frac{\sigma_i}{\Delta_i}\right)^2 \le d(\delta), \\[1mm]
&& x_i \in (-\infty,0), \quad \forall i \in [M]
\end{eqnarray*}
where $d(\delta)$ is defined in Section \ref{sec:unknown sigma 2exp}. The KKT condition yields
\[
-1+\lambda\,g_i'(x_i)\left(\tfrac{\sigma_i}{\Delta_i}\right)^2=0
\quad\Longleftrightarrow\quad
\boxed{\ g_i'(x_i)\,\left(\tfrac{\sigma_i}{\Delta_i}\right)^2=\frac{1}{\lambda}\ }
\]
which is independent of \(i\). Strict monotonicity of \(g_i'\) gives the same ordering as \textbf{R-TOL}.

\underline{\textbf{R-EXP}}

Let
\[
S(\vec x):=\sum_{i=1}^M g_i(x_i)\left(\tfrac{\sigma_i}{\Delta_i}\right)^2,
\qquad
P(\vec x):=\prod_{i=1}^M e^{x_i}=e^{\sum_{i=1}^M x_i},
\]
and define \(A(S):=\Phi\!\bigl(q_{1-\alpha}-\sqrt{N/S}\bigr)-1\).
Then \(A(S)\in(-1,0)\) and
\[
A'(S)=\varphi\!\bigl(q_{1-\alpha}-\sqrt{N/S}\bigr)\cdot\frac{\sqrt{N}}{2\,S^{3/2}}>0,
\]
where \(\varphi\) is the standard normal pdf. \textbf{EXP} can be written as
\begin{eqnarray*}
\min_{\vec{c}} && J(\vec x):=1+A\!\bigl(S(\vec x)\bigr)\,P(\vec x),\qquad \\
\text{subject to} && 
x_i \in (-\infty,0), \quad \forall i \in [M]
\end{eqnarray*}
Any optimum is interior (if some \(x_k=-\infty\) then \(P(\vec x)=0\) so \(J=1\); taking \(x_i=t<0\) for all \(i\) yields \(J(t)=1+A(S(t))e^{Mt}<1\)).  
First-order conditions give, for each \(i\),
\[
A'(S)\,g_i'(x_i)\left(\tfrac{\sigma_i}{\Delta_i}\right)^2 P(\vec x)+A(S)\,P(\vec x)=0
\quad\Longleftrightarrow\quad
\boxed{\ g_i'(x_i)\,\left(\tfrac{\sigma_i}{\Delta_i}\right)^2= -\frac{A(S)}{A'(S)}\ }
\]
with the right-hand side independent of \(i\). Again, strict increase of \(g_i'\) implies the same ordering as \textbf{R-TOL}. \hfill $\blacksquare$

\end{document}